\documentclass[acmsmall,screen]{acmart}
\AtBeginDocument{%
  }

\setcopyright{rightsretained}
\acmDOI{10.1145/3808257}
\acmYear{2026}
\acmJournal{PACMPL}
\acmVolume{10}
\acmNumber{PLDI}
\acmArticle{179}
\acmMonth{6}
\acmSubmissionID{pldi26main-p28-p}
\received{2025-11-14}
\received[accepted]{2026-04-03}

\usepackage{mathtools}
\mathtoolsset{showonlyrefs=true}
\usepackage{stmaryrd}
\usepackage{mathpartir}
\usepackage{extarrows}
\usepackage{multirow}
\usepackage{tikz,pgfplots}
\usetikzlibrary{cd,automata,positioning}
\tikzset{
	every state/.style={inner sep=2pt, minimum size=0pt},
	every path/.style={thick},
	initial text={},
}
\usepackage{booktabs}
\usepackage[beginLComment=//~, endLComment={}]{algpseudocodex}

\algnewcommand{\Break}{\textbf{break}}

\newif\ifdraft
\draftfalse

\newif\iflong
\longtrue

\iflong
\usepackage[bibliography=common]{apxproof}
\else
\usepackage[bibliography=common,appendix=strip]{apxproof}
\fi

\usepackage[arxiv=arxiv]{macros.apxproof}

\newcommand{\comp}{\mathrel{\circ}}
\newcommand{\Giry}{\mathcal{G}}

\newcommand{\Characteristic}[1]{\chi_{#1}}

\newcommand{\ifexpr}[3]{\mathbf{if}\ #1\ \mathbf{then}\ #2\ \mathbf{else}\ #3} %
\newcommand{\ifstmt}[3]{\mathtt{if}\ #1\ \mathtt{then}\ #2\ \mathtt{else}\ #3} %
\newcommand{\while}[2]{\mathtt{while}\ (#1)\ \{\ #2\ \}}
\newcommand{\assign}[2]{#1 \coloneqq #2}

\newcommand{\nexttime}{\mathbb{X}}
\newcommand{\compl}[1]{#1^c}

\newcommand{\OneToN}[1]{\{ 1, \dots, #1 \}}

\ifdraft
\newcommand{\todo}[1]{\textcolor{red}{[ToDo: #1]}}
\else
\newcommand{\todo}[1]{}
\fi

\AtEndPreamble{%
\theoremstyle{acmdefinition}

\newtheorem{remark}[theorem]{Remark}}

\begin{document}
\title{A Hierarchy of Supermartingales for \texorpdfstring{$\omega$}{ω}-Regular Verification}
\author{Satoshi Kura}
\email{satoshikura@acm.org}
\orcid{0000-0002-3954-8255}
\affiliation{%
  \institution{Waseda University}
  \city{Tokyo}
  \country{Japan}
}

\author{Hiroshi Unno}
\email{hiroshi.unno@acm.org}
\orcid{0000-0002-4225-8195}
\affiliation{%
  \institution{Tohoku University}
  \city{Sendai}
  \country{Japan}
}

\begin{abstract}
	We propose new supermartingale-based certificates for verifying almost sure satisfaction of $\omega$-regular properties: (1) \emph{generalised Streett supermartingales} (GSSMs) and their lexicographic extension (LexGSSMs), (2) \emph{distribution-valued Streett supermartingales} (DVSSMs), and (3) \emph{progress-measure supermartingales} (PMSMs) and their lexicographic extension (LexPMSMs).
	GSSMs, LexGSSMs, and DVSSMs are derived from least-fixed point characterisations of positive recurrence and null recurrence of Markov chains with respect to given Streett conditions; and PMSMs and LexPMSMs are probabilistic extensions of parity progress measures.
	We study the hierarchy among these certificates and existing certificates, namely Streett supermartingales, by comparing the classes of problems that can be verified by each type of certificates.
	Notably, we show that our certificates are strictly more powerful than Streett supermartingales.
	We also prove completeness of GSSMs for positive recurrence and of DVSSMs for null recurrence: DVSSMs are, in theory, the most powerful certificates in the sense that for any Markov chain that almost surely satisfies a given $\omega$-regular property, there exists a DVSSM certifying it.
	We provide a sound and relatively complete algorithm for synthesising LexPMSMs, the second most powerful certificates in the hierarchy.
	We have implemented a prototype tool based on this algorithm, and our experiments show that our tool can successfully synthesise certificates for various examples including those that cannot be certified by existing supermartingales.
\end{abstract}



\maketitle

\section{Introduction}
\label{sec:introduction}

\subsection{Background and Motivation}

Verification of probabilistic programs has been actively studied in recent years.
\emph{Martingale-based methods} are one of the major approaches in this area.
These methods guarantee properties of probabilistic programs by finding a function called a \emph{super-} or \emph{sub-martingale}, which is required to satisfy a certain inequality between its current value and its expected value after one step of execution.
Even when a target program has an infinite state space, martingale-based methods often admit automated synthesis via reduction to constraint solving.
Consequently, various types of martingales have been proposed to verify different types of properties of probabilistic programs, such as almost-sure termination \cite{ChakarovCAV2013,HuangOOPSLA2019,McIverPOPL2018,Kenyon-RobertsLICS2021,AgrawalPOPL2018,ChatterjeeFAC2023,TakisakaCAV2024}, quantitative termination \cite{ChatterjeeCAV2022,ChatterjeePOPL2017,TakisakaTOPLAS2021,UrabeLICS2017}, and cost analysis \cite{KuraTACAS2019,WangPLDI2019}.

Recently, supermartingales have also been studied for verifying \emph{$\omega$-regular properties} of probabilistic programs \cite{AbateCAV2024,AbateCAV2025,HenzingerCAV2025,ChakarovTACAS2016}.
For instance, consider a specification given as a linear temporal logic formula, such as one in LTL, over the sequence of observable events produced by a probabilistic program.
For example, one may ask whether the position $x = 0$ is visited infinitely often (i.e., the LTL formula $G F (x = 0)$ holds) in a one-dimensional fair random walk (which holds almost surely, as is well known).
In such a setting, the probabilistic program induces a Markov chain (or Markov decision process if the program has nondeterminism), which typically has an infinite state space.
The LTL formula can be translated into an $\omega$-automaton, e.g., a deterministic Streett automaton.
By applying the product construction to the Markov chain and the deterministic Streett automaton, taking into account the emitted events during transitions, we obtain a Markov chain equipped with a Streett acceptance condition.
The verification problem then asks whether the Markov chain satisfies the associated Streett condition with probability~1 (\emph{qualitative} verification) or with a probability above a given threshold (\emph{quantitative} verification).
Streett supermartingales~\cite{AbateCAV2024} serve as certificates for such qualitative verification problems.
Moreover, when combined with stochastic invariants \cite{ChatterjeePOPL2017}, they can also be used to support quantitative verification \cite{AbateCAV2025,HenzingerCAV2025}.

Although the above problem setting can be seen as a natural extension of the well-studied $\omega$-regular verification problem for non-probabilistic programs (e.g., \cite{CookPOPL2007,UnnoPOPL2023,ChatterjeeFM2025}) and probabilistic model checking over finite state spaces \cite{VardiSFCS1985} (see also the modern textbook treatment in \cite{Baier2008}), much less is known about the verification of $\omega$-regular properties of probabilistic programs with infinite state spaces.

\paragraph{Limitations of existing work}
To the best of our knowledge, Streett supermartingales \cite{AbateCAV2024} 
are essentially the only known supermartingale-based certificate for qualitative verification of general $\omega$-regular properties.
However, their verification power has not been well understood.
In fact, there exists a simple $\omega$-regular verification problem that cannot be verified by Streett supermartingales.
For example, consider the following simple program, which is not even a probabilistic program:
\begin{equation}
	\while{\mathtt{true}}{\assign{m}{n};\ \while{m > 0}{\assign{m}{m - 1}};\ \assign{n}{n + 1}}
	\label{eq:intro-ssm-example}
\end{equation}
Suppose that we want to verify ``either the inner loop is executed finitely often, or the outer loop is executed infinitely often'', which can naturally be expressed as a Streett acceptance condition.
Although this property clearly holds, no Streett supermartingale can verify it, as we will show in Example~\ref{ex:no-streett-supermartingale}.
This limitation raises the following questions: (1) Is there any supermartingale-based certificate more powerful than Streett supermartingales? (2) Can we characterise the verification power of supermartingales for $\omega$-regular properties?

\begin{figure}
	\begin{tikzpicture}[scale=0.9, transform shape]
		\node (DVGSSM) at (0, 3) {Streett-condition based};
		\node (DVGSSM) at (6, 3) {Parity-condition based};
		\node[draw, align=center] (SSM) at (0, -3.6) {Streett supermartingale \\ \cite{AbateCAV2024}};
		\node[draw, align=center] (GSSM) at (0, -2) {Generalised \\ Streett supermartingale \\ (Definition~\ref{def:generalised-streett-supermartingale})};
		\node[draw, align=center] (DVGSSM) at (0, 1.8) {Distribution-valued \\ Streett supermartingale \\ (Definition~\ref{def:distribution-valued-streett-supermartingale})};
		\node[draw, align=center] (LexGSSM) at (0, 0) {Lexicographic generalised \\ Streett supermartingale \\ (Definition~\ref{def:lexicographic-generalised-streett-supermartingale})};
		\node[draw, align=center] (PPM) at (6, -2) {Progress-measure \\ supermartingale \\ (Definition~\ref{def:progress-measure})};
		\node[draw, align=center] (LexPPM) at (6, 0) {Lexicographic \\ progress-measure \\ supermartingale \\ (Definition~\ref{def:lexpmsm})};
		\node[align=center] (PR) at (-4.9, -2) {Positive \\ recurrence};
		\node[align=center] (NR) at (-5.9, 1.8) {Null recurrence \\ (= $\omega$-regular properties)};
		\draw[dashed] (-3, 2.5) -- (-3, -3.8);
		\path (SSM) -- node[sloped, auto=false, allow upside down] {$\subsetneq$} node[midway, right=1ex, font=\footnotesize] {Ex.~\ref{ex:no-streett-supermartingale}} (GSSM);
		\path (GSSM) -- node[sloped, auto=false, allow upside down] {$\subsetneq$} node[midway, right=1ex, font=\footnotesize] {Ex.~\ref{ex:generalised-streett-supermartingale}} (LexGSSM);
		\path (LexGSSM) -- node[sloped, auto=false, allow upside down] {$\subsetneq$} node[midway, right=1ex, font=\footnotesize] {Ex.~\ref{ex:lexgssm-not-complete}} (DVGSSM);
		\path (PPM) -- node[sloped, auto=false, allow upside down] {$\subseteq$} (LexPPM);
		\path (GSSM) -- node[sloped, auto=false, allow upside down] {$\subseteq$} node[above=1ex, font=\footnotesize] {(Prop.~\ref{prop:lexgssm-to-lexpmsm})} (PPM);
		\path (LexGSSM) -- node[sloped, auto=false, allow upside down] {$=$} node[midway, above=1ex, font=\footnotesize] {Prop.~\ref{prop:lexpmsm-to-lexgssm},~\ref{prop:lexgssm-to-lexpmsm}} (LexPPM);
		\node[fill=white, align=center] at (-3, -2) {{\footnotesize Thm.~\ref{thm:gssm-sound-complete}} \\ $\iff$};
		\node[fill=white, align=center] at (-3, 1.8) {{\footnotesize Cor.~\ref{cor:dvssm-sound-complete}} \\ $\iff$};
	\end{tikzpicture}
	\caption{The hierarchy of supermartingale certificates for almost sure satisfaction of $\omega$-regular properties. The middle column shows supermartingales for Streett conditions, and the right column shows those for parity conditions. Three proper inclusions $\subsetneq$ in the figure are witnessed by Example~\ref{ex:no-streett-supermartingale},~\ref{ex:generalised-streett-supermartingale}, and~\ref{ex:lexgssm-not-complete}.}
	\label{fig:hierarchy-of-supermartingales}
	\Description{The hierarchy of supermartingale certificates for almost sure satisfaction of $\omega$-regular properties. The middle column shows supermartingales for Streett conditions: Streett supermartingales, generalised Streett supermartingales, lexicographic generalised Streett supermartingales, and distribution-valued Streett supermartingales.
	Streett supermartingales are strictly less powerful than generalised Streett supermartingales, which are strictly less powerful than lexicographic generalised Streett supermartingales, which are strictly less powerful than distribution-valued Streett supermartingales.
	The right column shows those for parity conditions: progress-measure supermartingales and lexicographic progress-measure supermartingales.
	Lexicographic progress-measure supermartingales are as powerful as lexicographic generalised Streett supermartingales.
	The power of progress-measure supermartingales is between that of generalised Streett supermartingales and lexicographic generalised Streett supermartingales.
	The left column shows the corresponding classes of verification problems: generalised Streett supermartingales are sound and complete for positive recurrence, and distribution-valued Streett supermartingales are sound and complete for null recurrence.}
\end{figure}

\subsection{Our Contributions}
We propose several new supermartingale-based certificates for qualitative $\omega$-regular verification problems and systematically investigate their verification power (Fig.~\ref{fig:hierarchy-of-supermartingales}).
Here, a \emph{qualitative $\omega$-regular verification problem} asks, given an infinite-state Markov chain and an acceptance condition for an $\omega$-regular property, whether the probability of satisfying the acceptance condition is 1.
We consider two types of acceptance conditions: Streett conditions and parity conditions.
In Fig.~\ref{fig:hierarchy-of-supermartingales}, the middle column shows supermartingales for Streett conditions, while the right column shows those for parity conditions.
Each type of supermartingales induces a class of problems that can be verified by it.
We establish inclusion relations between these classes via translations between different types of supermartingales and show that some of these inclusions are proper by providing examples that separate the classes.

\paragraph{Ideas behind generalised Streett supermartingales}
There is a close connection between Streett conditions and recurrence properties of Markov chains.
A \emph{Streett condition} is specified by (a finite set of) pairs $(A, B)$ where $A$ and $B$ are subsets of states of a Markov chain.
An infinite trace satisfies the Streett condition if either $A$ is visited only finitely often or $B$ is visited infinitely often for each pair $(A, B)$ in the Streett condition.
This condition for a single pair $(A, B)$ is equivalent to saying that at any point in the trace, the number of remaining visits to $A$ before reaching $B$ is finite.
Following terminology of Markov chain theory, we say a Markov chain is \emph{$(A, B)$-null recurrent} if the random variable $\mathrm{step}^{(A, B)}$ that counts the number of visits to $A$ before reaching $B$ is almost surely finite (i.e., $\mathbb{P}[\mathrm{step}^{(A, B)} < \infty] = 1$).
Then, we have that almost sure satisfaction of a Streett condition is equivalent to $(A, B)$-null recurrence.

The notion of \emph{generalised Streett supermartingales} (GSSMs) is obtained by adapting the order-theoretic characterisation of ranking supermartingales \cite{TakisakaTOPLAS2021,UrabeLICS2017} to \emph{$(A, B)$-positive recurrence}, which is a stronger notion than $(A, B)$-null recurrence and defined by $\mathbb{E}[\mathrm{step}^{(A, B)}] < \infty$.
It is straightforward to characterise $\mathbb{E}[\mathrm{step}^{(A, B)}]$ as the least fixed point $\mu K_{\mathbb{E}}$ of a monotone function $K_{\mathbb{E}}$.
Thus, if there exists a prefixed point $r \ge K_{\mathbb{E}} (r)$ such that $r < \infty$, then by Park induction (or the Knaster-Tarski theorem), we have $\mathbb{E}[\mathrm{step}^{(A, B)}] = \mu K_{\mathbb{E}} \le r < \infty$, which implies $(A, B)$-positive recurrence.
We call such an $r$ a \emph{generalised Streett supermartingale}.
Concretely, a GSSM is a non-negative function $r \colon S \to [0, \infty)$ over the state space $S$ of a Markov chain that satisfies the following condition:
\[ \forall x \in A \setminus B,\quad (\nexttime r)(x) \le r(x) - 1, \qquad\qquad\qquad \forall x \notin A \cup B,\quad (\nexttime r)(x) \le r(x) \]
where $\nexttime r$ is the expected value of $r$ after one step.
Intuitively, $r$ decreases by at least 1 in expectation when the current state is in $A \setminus B$, and does not increase in expectation when outside $A \cup B$.
No constraint is imposed when $x \in B$, which is the key difference from Streett supermartingales \cite{AbateCAV2024}.
This order-theoretic argument not only yields soundness but also completeness of GSSMs for $(A, B)$-positive recurrence: if a Markov chain is $(A, B)$-positive recurrent, then there exists a GSSM.
Moreover, the program in~\eqref{eq:intro-ssm-example} admits a GSSM but no Streett supermartingale, showing that GSSMs are strictly more powerful.

A similar order-theoretic argument can also be applied to $(A, B)$-null recurrence, which yields \emph{distribution-valued Streett supermartingales} (DVSSMs).
DVSSMs are complete for null recurrence, and therefore the most powerful among the supermartingales for $\omega$-regular properties in theory.
However, their automated synthesis is difficult, which limits their practical usefulness.

\paragraph{Expressive and automation-friendly supermartingales}
To provide supermartingales that can verify null recurrence and are more amenable to automation, we introduce \emph{lexicographic generalised Streett supermartingales} (LexGSSMs), \emph{progress-measure supermartingales} (PMSMs), and a lexicographic extension of PMSMs (LexPMSMs).
PMSMs are supermartingales for parity conditions and inspired by \emph{parity progress measures} \cite{JurdzinskiSTACS2000}, which are used in algorithms for solving parity games.
We adapt this notion to the probabilistic systems by identifying the essence of the correctness proof for parity games and parity graphs, and by combining it with ideas from lexicographic ranking supermartingales \cite{AgrawalPOPL2018}.
After proving the soundness of PMSMs, we extend them to LexPMSMs by considering nested lexicographic orderings.
LexPMSMs are more expressive and possess better theoretical properties: we show that LexPMSMs are equivalent to LexGSSMs, which are a lexicographic extension of GSSMs.

Lexicographic ranking supermartingales have proven effective in automated verification \cite{AgrawalPOPL2018,ChatterjeeFAC2023,TakisakaCAV2024}, and our LexPMSMs inherit this advantage.
We extend the algorithm in \cite{AgrawalPOPL2018} to synthesise LexPMSMs and implement it using an off-the-shelf constraint solver.
Our experiments show that our prototype implementation can successfully synthesise LexPMSMs for examples that cannot be handled by Streett supermartingales.

Our contributions are summarised as follows:
\begin{itemize}
	\item We introduce several new supermartingale-based certificates for almost-sure satisfaction of $\omega$-regular properties, including \emph{generalised Streett supermartingales} (GSSMs) and their lexicographic extension (LexGSSMs); \emph{distribution-valued Streett supermartingales} (DVSSMs); and \emph{progress-measure supermartingales} (PMSMs) and their lexicographic extension (LexPMSMs).
	We show that these are strictly more powerful than the existing approach, namely Streett supermartingales.
	\item We systematically study the verification power of these supermartingales and establish a hierarchy among them (Fig.~\ref{fig:hierarchy-of-supermartingales}).
	We prove inclusion relations between the corresponding classes of verification problems, provide translations witnessing these inclusions, and construct separating examples showing that some inclusions are strict.
	We also prove completeness of GSSMs for positive recurrence and of DVSSMs for null recurrence.
	\item We develop a sound and relatively complete algorithm for synthesising LexPMSMs, implement a prototype tool based on it, and demonstrate its effectiveness through experiments.
\end{itemize}

Omitted proofs can be found in
\iflong
the appendix.
\else
the long version \cite{arxiv}.
\fi

\section{Preliminaries}
\label{sec:preliminaries}

\subsection{Markov Chains}\label{sec:markov-chains}
As a model of probabilistic programs, we consider Markov chains whose state space may be infinite.

Let $A = (A, \Sigma_A)$ and $B = (B, \Sigma_B)$ be measurable spaces where $\Sigma_A$ and $\Sigma_B$ are $\sigma$-algebras on $A$ and $B$, respectively.
A \emph{Markov kernel} from $A$ to $B$ is a function $F \colon A \times \Sigma_B \to [0, 1]$ such that (a) for each $x \in A$, $F(x, {-})$ is a probability measure on $B$ and (b) for each measurable set $E \in \Sigma_B$, the function $F(-, E) \colon A \to [0, 1]$ is measurable.
Markov kernels can be equivalently defined as measurable functions $F \colon A \to \Giry B$ where $\Giry B$ is the set of probability measures on $B$ equipped with the smallest $\sigma$-algebra that makes the evaluation map $\mu \mapsto \mu(E)$ measurable for each $E \in \Sigma_B$.
The construction of the measurable space $\Giry B$ is known as the \emph{Giry monad}.
A \emph{Markov chain} over a measurable space $S$ is a Markov kernel $F \colon S \to \Giry S$.
We do not assume that the state space $S$ is finite.
We write the Dirac distribution that assigns probability $1$ to a point $s \in S$ as $\delta_s \in \Giry S$.

Let $g \colon B \to [0, \infty]$ be a measurable function and $\mu$ be a measure on $B$ where $[0, \infty]$ denotes the set of non-negative extended real numbers equipped with the Borel $\sigma$-algebra.
We write the integral of $g$ with respect to $\mu$ as $\int g \, \mathrm{d} \mu$ or $\int g(x) \, \mathrm{d} \mu(x)$ where $x$ is a variable ranging over $B$.
If $F \colon A \to \Giry B$ is a Markov kernel and $x \in A$, then the integral with respect to $F(x)$ is denoted as $\int g \, \mathrm{d} F_x$ or $\int g(y) \, \mathrm{d} F_x(y)$ (note the subscript).
By definition of Markov kernels, it follows that $x \mapsto \int g \, \mathrm{d} F_x$ is a measurable function.
If $F \colon S \to \Giry S$ is a Markov chain, then the \emph{next-time operator} $\nexttime_F$ is defined as the mapping from a measurable function $g \colon S \to [0, \infty]$ to a measurable function $\nexttime_F g \colon S \to [0, \infty]$ where $(\nexttime_F g)(s) = \int g \, \mathrm{d} F_s$.
The subscript of $\nexttime_F$ is omitted when $F$ is clear from the context.

In the context of probabilistic program verification, one typical way to obtain Markov chains is through \emph{probabilistic control-flow graphs} (\emph{pCFGs}) \cite{ChatterjeeCAV2016,AgrawalPOPL2018}, which are a common formalism for representing imperative probabilistic programs.
Since our theoretical development is based on Markov chains, we do not consider nondeterminism in pCFGs.
Here, instead of giving a formal definition of pCFGs, we illustrate how imperative probabilistic programs are modelled as pCFGs and how pCFGs induce Markov chains and its next-time operators.
\begin{example}\label{ex:pCFG}
	Consider the following probabilistic program.
	\[ \while{x > 0}{\ifstmt{\mathtt{prob}(0.5)}{x \coloneqq x - 1}{x \coloneqq x + \mathtt{unif}(0, 1)}} \]
	Here, $\mathtt{prob}(p)$ returns $\mathtt{true}$ with probability $p$ and $\mathtt{false}$ with probability $1 - p$, and $\mathtt{unif}(a, b)$ is a uniform sampling from the interval $[a, b]$.
	This program can be represented as the following pCFG, which has five locations $L = \{ l_0, l_1, l_2, l_3, l_4 \}$ and one real-valued program variable $x$.
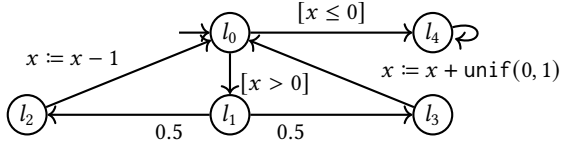
	\begin{center}
		\begin{tikzpicture}[scale=0.9, every node/.style={transform shape}]
			\node[state, initial] (l0) at (0, 0) {$l_0$};
			\node[state] (l1) at (0, -1.2) {$l_1$};
			\node[state] (l2) at (-3, -1.2) {$l_2$};
			\node[state] (l3) at (3, -1.2) {$l_3$};
			\node[state] (l4) at (3, 0) {$l_4$};

			\draw[->] (l0) -- node[right, near end] {$[x > 0]$} (l1);
			\draw[->] (l0) -- node[above] {$[x \le 0]$} (l4);
			\draw[->] (l1) -- node[below, near start] {$0.5$} (l2);
			\draw[->] (l1) -- node[below, near start] {$0.5$} (l3);
			\draw[->] (l2) -- node[above left] {$\assign{x}{x - 1}$} (l0);
			\draw[->] (l3) -- node[above right, near start] {$\assign{x}{x + \mathtt{unif}(0, 1)}$} (l0);
			\draw[->] (l4) edge[loop right] (l4);
		\end{tikzpicture}
	\end{center}
	The location $l_0$ corresponds to the conditional branching for the while loop, $l_1$ is the probabilistic branching for $\mathtt{prob}(0.5)$, $l_2$ is the deterministic assignment $x \coloneqq x - 1$, $l_3$ is the random assignment $x \coloneqq x + \mathtt{unif}(0, 1)$, and $l_4$ is the terminal location with a self-loop.
	This pCFG can be interpreted as an infinite-state Markov chain with the state space $S = L \times \mathbb{R}$ where $\mathbb{R}$ represents valuations of the program variable $x$.
	Specifically, the Markov chain $F \colon S \to \Giry S$ is given as follows.
	\[ F(l_0, x)\ =\ \begin{cases}
		\delta_{(l_1, x)} & \text{if } x > 0 \\
		\delta_{(l_4, x)} & \text{if } x \le 0
	\end{cases}
	\qquad
	\begin{aligned}
		F(l_1, x)\ &=\ 0.5 \cdot \delta_{(l_2, x)} + 0.5 \cdot \delta_{(l_3, x)} \\
		F(l_3, x)\ &=\ \lambda E. \mu (E \cap [x, x + 1])
	\end{aligned}
	\qquad
	\begin{aligned}
		F(l_2, x)\ &=\ \delta_{(l_0, x - 1)} \\
		F(l_4, x)\ &=\ \delta_{(l_4, x)}
	\end{aligned}
	\]
	Here, $\mu$ is the Lebesgue measure on $\mathbb{R}$ and $\delta$ is the Dirac distribution.
	The corresponding next-time operator $\nexttime$ is given as follows for any measurable function $r \colon L \times \mathbb{R} \to [0, \infty]$.
	\begin{align}
		(\nexttime r)(l_0, x)\ &=\ \begin{cases}
			r(l_1, x) & \text{if } x > 0 \\
			r(l_4, x) & \text{if } x \le 0
		\end{cases} &
		(\nexttime r)(l_3, x)\ &=\ \int_0^1 r(l_0, x + y) \, \mathrm{d} y \\
		(\nexttime r)(l_1, x)\ &=\ 0.5 \cdot r(l_2, x) + 0.5 \cdot r(l_3, x) &
		(\nexttime r)(l_2, x)\ &=\ r(l_0, x - 1) \quad (\nexttime r)(l_4, x)\ =\ r(l_4, x) \tag*{\qed}
	\end{align}
\end{example}

Markov chains induce probability measures on infinite traces, which can be found in standard textbooks, e.g., \cite{Baier2008}.
Let $F \colon S \to \Giry S$ be a Markov chain.
Then, the set $S^{\omega}$ of infinite traces over $S$ is a measurable space with the $\sigma$-algebra $\Sigma_{S^{\omega}}$ generated by the cylinder sets.
Given an initial state $s_0 \in S$, a probability measure on $S^{\omega}$ is defined by the Kolmogorov extension theorem.
That is, $F \colon S \to \Giry S$ can be extended to $F_{\omega} \colon S \to \Giry S^{\omega}$ as follows.

\begin{lemma}
	Given a Markov chain $F \colon S \to \Giry S$, we have a unique Markov kernel $F_{\omega} \colon S \to \Giry S^{\omega}$ such that for each measurable function $g \colon S^n \to [0, \infty]$ and $s_0 \in S$, we have the following equation.
	\[ \int g(s_1, \dots, s_n) \, \mathrm{d} (F_{\omega})_{s_0}(s_1 s_2 \dots) = \int \dots \int g(s_1, \dots, s_n) \, \mathrm{d} F_{s_{n - 1}}(s_n) \dots \, \mathrm{d} F_{s_1}(s_2)\, \mathrm{d} F_{s_0}(s_1) \tag*{\qed} \]
\end{lemma}

For simplicity of notation, we write $\mathbb{P}_{s_0}[E] = F_{\omega}(s_0)(E)$ for the probability of a measurable set $E \subseteq S^{\omega}$ and $\mathbb{E}_{s_0}[g] = \int g \,\mathrm{d} (F_{\omega})_{s_0}$ for the expected value of a measurable function $g \colon S^{\omega} \to [0, \infty]$ when a Markov chain $F$ is clear from the context.

\subsection{\texorpdfstring{$\omega$}{ω}-Regular Verification for Markov Chains}\label{sec:omega-regular-verification}

In what follows, we formalise the target problem of this paper, which is the problem of verifying whether a given Markov chain satisfies an $\omega$-regular property almost surely.
It is well-known that $\omega$-regular languages are characterised by various kinds of $\omega$-automata \cite{Gradel2002}, where an $\omega$-automaton is defined as a finite-state automaton that runs on infinite words and uses an \emph{acceptance condition} to determine whether an infinite word is accepted or not.
We consider two of such acceptance conditions, \emph{Streett conditions} and \emph{parity conditions}, and adapt them to possibly infinite-state Markov chains.

\begin{definition}[Streett condition]
	A \emph{Streett pair} over a set $S$ is a tuple $(A, B)$ of subsets $A, B \subseteq S$.
	A \emph{Streett condition} is a finite set $\{ (A_1, B_1), \dots, (A_n, B_n) \}$ of Streett pairs.
	A trace $s_0 s_1 \dots \in S^{\omega}$ satisfies the Streett condition if for any $i$, either $A_i$ is visited finitely often or $B_i$ is visited infinitely often.
	\begin{align*}
		&\mathbf{Streett}(A, B) &&\coloneqq\quad \{ s_0 s_1 \dots \in S^{\omega} \mid \# \{ i \mid s_i \in A \} < \infty \lor \# \{ i \mid s_i \in B \} = \infty \} \\
		&\mathbf{Streett}\big((A_1, B_1), \dots, (A_n, B_n)\big) &&\coloneqq\quad \bigcap_{i=1}^n \mathbf{Streett}(A_i, B_i)
	\end{align*}
	If $S$ is a measurable space, then a Streett condition is said to be \emph{measurable} if each $A_i$ and $B_i$ is a measurable subset of $S$.
	We call a pair of a Markov chain $F \colon S \to \Giry S$ and a measurable Streett condition over $S$ a \emph{Streett Markov chain}.
\end{definition}

\begin{example}
	Termination of a program can be expressed as a Streett condition.
	For example, consider the program in Example~\ref{ex:pCFG}.
	In this case, termination is equivalent to saying that the location $l_4$ is visited in finite steps and thus can be expressed as the Streett pair $(\{ l_0, l_1, l_2, l_3 \} \times \mathbb{R}, \emptyset)$.
	Recurrence (i.e., visiting a certain set of states infinitely often) can also be expressed as a Streett condition.
	The property of visiting the location $l_0$ in Example~\ref{ex:pCFG} infinitely often can be expressed as the Streett pair $(\{ l_1, l_2, l_3, l_4 \} \times \mathbb{R}, \{ l_0 \} \times \mathbb{R})$.
	\qed
\end{example}

\begin{definition}[parity condition]\label{def:parity-condition}
	A \emph{parity condition} over a set $S$ is specified by a \emph{priority function} $p \colon S \to \OneToN{d}$ that assigns a priority (a natural number between $1$ and $d$) to each element of $S$ where $d$ is some positive integer.
	An infinite sequence $s_0 s_1 \dots \in S^{\omega}$ satisfies the parity condition $p$ if the minimum priority that occurs infinitely often in the sequence is even.
	\[ \mathbf{Parity}(p) \coloneqq \{ s \in S^{\omega} \mid \min \mathbf{Inf}_p(s) \text{ is even} \} \quad\text{where}\quad \mathbf{Inf}_p(s_0 s_1 \dots) \coloneqq \{ q \mid \# \{ i \mid p(s_i) = q \} = \infty \} \]
	We call a pair of a Markov chain $F \colon S \to \Giry S$ and a measurable priority function $p \colon S \to \OneToN{d}$ a \emph{parity Markov chain}.
	Here, we consider the discrete $\sigma$-algebra on $\OneToN{d}$, which is the powerset of $\OneToN{d}$.
\end{definition}

Translating parity conditions to Streett conditions is easy.
\begin{lemma}\label{lem:parity-streett}
	Let $p \colon S \to \{ 1, \dots, 2 d \}$ be a priority function.
	Then, the parity condition for $p$ is equivalent to the Streett condition $\{ (S_{\le 2 i - 1}, S_{\le 2 i - 2}) \mid i = 1, \dots, d \}$ where $S_{\le k} \coloneqq \{ x \in S \mid p(x) \le k \}$.
	\[ \mathbf{Parity}(p) \quad=\quad \mathbf{Streett}\big( (S_{\le 1}, S_{\le 0}), (S_{\le 3}, S_{\le 2}), \dots, (S_{\le 2 d - 1}, S_{\le 2 d - 2}) \big) \tag*{\qed} \]
\end{lemma}

Conversely, translating a Streett condition into a parity condition is computationally expensive \cite{Gradel2002}.
However, translating a Streett condition into a \emph{conjunction} of parity conditions is straightforward.

\begin{lemma}\label{lem:streett-pair-parity}
	Any Streett pair $(A, B)$ is equivalent to the parity condition $p \colon S \to \{ 1, 2, 3, 4 \}$ defined by $p(x) = 2$ if $x \in B$, $p(x) = 3$ if $x \in A \setminus B$, and $p(x) = 4$ otherwise.
	\qed
\end{lemma}

\paragraph{Verification Problems}
We consider the following verification problem in this paper.
Suppose that we have a parity Markov chain $F \colon S \to \Giry S$ (or a Streett Markov chain).
A \emph{qualitative verification problem} asks whether the probability of satisfying the parity condition $p$ is $1$ (or $0$) for \emph{any} initial state $s_0 \in S$.
\[ \mathbb{P}_{s_0}[\mathbf{Parity}(p)] = 1 \quad(\text{or}\quad \mathbb{P}_{s_0}[\mathbf{Parity}(p)] = 0) \qquad \text{for any $s_0 \in S$} \]
On the other hand, a \emph{quantitative verification problem} asks whether the probability of satisfying the parity condition is lower bounded by $l \colon S \to [0, 1]$ (or upper bounded by $u \colon S \to [0, 1]$).
\[ l(s_0) \le \mathbb{P}_{s_0}[\mathbf{Parity}(p)] \quad(\text{or}\quad \mathbb{P}_{s_0}[\mathbf{Parity}(p)] \le u(s_0)) \qquad \text{for any $s_0 \in S$} \]
Qualitative verification is a special case of quantitative verification where $l = 1$ and $u = 0$.
Note that lower-bound problems and upper-bound problems are inter-reducible by considering the complement of the acceptance condition.

In this paper, we focus on qualitative verification problems for $\mathbb{P}_{s_0}[\mathbf{Parity}(p)] = 1$.
However, combining our supermartingale-based certificates with existing techniques for quantitative verification \cite{AbateCAV2025,HenzingerCAV2025} is straightforward because this is achieved by combination with a certificate for stochastic invariants \cite{ChatterjeePOPL2017,ChatterjeeCAV2022,TakisakaTOPLAS2021}.

In practice, we often want to restrict our attention to the states in an invariant $I \subseteq S$ of a Markov chain $F \colon S \to \Giry S$ and ask whether the parity condition is satisfied almost surely from any initial state in the invariant $I$.
This situation can be reduced to the above qualitative verification problem by considering the sub-Markov chain $F|_I \colon I \to \Giry I$ induced by the invariant.

\subsection{Streett Supermartingales}

We review \emph{Streett supermartingales} \cite{AbateCAV2024}, which are a type of supermartingale-based certificate for almost sure satisfaction of Streett conditions.
We focus on the case of a single Streett pair $(A, B)$, since Streett supermartingales for a general Streett condition $\{ (A_1, B_1), \dots, (A_n, B_n) \}$ are defined simply as tuples $(r_1, \dots, r_n)$ of Streett supermartingales $r_i$ for each Streett pair $(A_i, B_i)$.

\begin{definition}\label{def:streett-ranking-supermartingale}
	A \emph{Streett supermartingale} (or \emph{SSM} for short) for a Streett Markov chain $F \colon S \to \Giry S$ with a Streett pair $(A, B)$ is a measurable function $r \colon S \to [0, \infty)$ such that for each $x \in S$,
	\begin{equation}
		(\nexttime r)(x) \quad\le\quad r(x) - \epsilon \Characteristic{A \setminus B}(x) + M \Characteristic{B}(x)
		\label{eq:streett-ranking-supermartingale}
	\end{equation}
	where $\Characteristic{B}$ is the indicator function of $B$ and $\epsilon, M > 0$ are constants.
\end{definition}
The condition~\eqref{eq:streett-ranking-supermartingale} says that if $x \in A \setminus B$, then $r$ must decrease on average by at least $\epsilon$; if $x \in B$, then $r$ may increase on average by at most $M$; and otherwise, $r$ is non-increasing on average.
Note that we can assume $\epsilon = 1$ without loss of generality by multiplying $1 / \epsilon$ to $r$.
The existence of Streett supermartingales ensures that the Streett condition is satisfied almost surely.
This is intuitively because if $A$ is visited infinitely often but $B$ is visited finitely often with positive probability, then the non-negativity of $r$ must be violated by the ranking condition~\eqref{eq:streett-ranking-supermartingale}.

\begin{proposition}
	If there exists a Streett supermartingale for a Streett Markov chain, then the Streett condition is satisfied almost surely from any initial state \cite[Theorem~2]{AbateCAV2024}.
	\qed
\end{proposition}

Supermartingales for recurrence \cite[Section~3.2]{ChakarovTACAS2016} are also used for $\omega$-regular verification of probabilistic systems \cite{HenzingerCAV2025}.
These can be seen as a special case of Streett supermartingales for the Streett pair $(S \setminus B, B)$, i.e., the B\"uchi condition for $B$.

\section{Positive Recurrence and Streett Supermartingales}
\label{sec:positive-recurrence}

In this section, we first describe the relationship between Streett conditions and positive/null recurrence for Markov chains.
Then, we introduce a novel notion of generalised Streett supermartingales (GSSMs) based on the characterisation of positive recurrence as a least fixed point.
We show that GSSMs are sound and complete certificates for positive recurrence.
We also provide a qualitative verification problem that can be verified by GSSMs but not by Streett supermartingales (Definition~\ref{def:streett-ranking-supermartingale}).
This shows that the class of qualitative verification problems that can be verified by Streett supermartingales is a proper subclass of positively recurrent Markov chains.

\subsection{Recurrence and Streett Conditions}

Given a Markov chain, we consider the number of steps in a measurable set $A$ until reaching another measurable set $B$ and define positive/null recurrence with respect to $(A, B)$.
\begin{definition}
	Let $F \colon S \to \Giry S$ be a Markov chain.
	Recall that a Markov chain $F$ induces a probability measure $\mathbb{P}_{s_0}$ on $S^{\omega}$ for each initial state $s_0 \in S$.
	For any measurable sets $A, B \subseteq S$, we define a random variable $\mathrm{step}_{s_0}^{(A, B)} \colon S^{\omega} \to \mathbb{N}_{\infty}$ as follows.
	\[ \mathrm{step}_{s_0}^{(A, B)}(s_1 s_2 \cdots) \quad\coloneqq\quad |\{ i \mid 0 \le i < \min \{ j \mid s_j \in B \land j \ge 0 \} \land s_i \in A \}| \]
	The function $\mathrm{step}_{s_0}^{(A, B)}(s_1 s_2 \cdots)$ counts the number of $A$-steps in $s_0 s_1 s_2 \cdots$ up to the first $B$-step $s_j \in B$.
	Note that $s_0$ is counted as an $A$-step if $s_0 \in A$ and as a $B$-step if $s_0 \in B$.
	We say a Markov chain $F$ is \emph{positively $(A, B)$-recurrent} from an initial state $s_0$ if the expected number of $A$-steps until reaching $B$ is finite, i.e., $\mathbb{E}_{s_0}[\mathrm{step}_{s_0}^{(A, B)}] < \infty$.
	Similarly, we say $F$ is \emph{null $(A, B)$-recurrent} from an initial state $s_0$ if the number of $A$-steps until reaching $B$ is almost surely finite, i.e., $\mathbb{P}_{s_0}[\mathrm{step}_{s_0}^{(A, B)} < \infty] = 1$.
\end{definition}

The relationship between null recurrence and positive recurrence is analogous to that between almost sure termination and positive almost sure termination.
In fact, if $T \subseteq S$ is a set of terminating states with self-loops, then positive (resp.\ null) $(S \setminus T, \emptyset)$-recurrence from an initial state $s_0$ coincides with the notion of positive almost sure termination (resp.\ almost sure termination) from $s_0$.

It is easy to see that positive recurrence is a strictly stronger condition than null recurrence because if $\mathbb{P}_{s_0}[\mathrm{step}_{s_0}^{(A, B)} = \infty] > 0$, then we have $\mathbb{E}_{s_0}[\mathrm{step}_{s_0}^{(A, B)}] = \infty$.

\begin{lemma}\label{lem:positive-implies-null-recurrence}
	If a Markov chain $F$ is positively $(A, B)$-recurrent from an initial state $s_0$, then it is null $(A, B)$-recurrent from $s_0$.
	\qed
\end{lemma}

On the other hand, there exists a Markov chain $F$ that is null $(A, B)$-recurrent but not positively $(A, B)$-recurrent.
A well-known example is a fair random walk on $\mathbb{Z}$, which is null recurrent but not positively recurrent for $(A, B) = (Z \setminus \{ 0 \}, \{ 0 \})$.

Null recurrence corresponds to almost sure satisfaction of Streett conditions as follows.
\begin{lemma}\label{lem:streett-null-recurrence}
	Given a Streett pair $(A, B)$, the Streett condition is satisfied almost surely from any initial state $s_0$ if and only if the Markov chain is null $(A, B)$-recurrent from any initial state $s_0$.
	\[ \forall s_0, \mathbb{P}_{s_0}[\mathbf{Streett}(A, B)] = 1 \quad\iff\quad \forall s_0, \mathbb{P}_{s_0}[\mathrm{step}_{s_0}^{(A, B)} < \infty] = 1 \tag*{\qed} \]
\end{lemma}
\begin{appendixproof}[Proof of Lemma \ref{lem:streett-null-recurrence}]
	Let $\#^A(s)$ be the number of visits to $A$ in the trace $s$.
	\[ \#^A(s_0 s_1 s_2 \cdots) \quad\coloneqq\quad \# \{ i \mid s_i \in A \} \]
	We define $\#^B(s)$ similarly.
	Note that we have $\mathrm{step}_{s_0}^{(A, B)}(s) = \infty$ if and only if $\#^A(s_0 \cdot s) = \infty$ and $\#^B(s_0 \cdot s) = 0$.
	This further implies the following equation.
	\[ \mathbb{P}_{s_0}[\#^A = \infty \land \#^B = 0] \quad=\quad \int \mathbb{P}_{s_1}[\mathrm{step}_{s_1}^{(A, B)} = \infty] \,\mathrm{d} F_{s_0}(s_1) \]
	\begin{itemize}
		\item ($\implies$) We have $\mathbf{Streett}(A, B) \subseteq \{ s \mid \mathrm{step}_{s_0}^{(A, B)}(s) < \infty \}$ because if $\mathrm{step}_{s_0}^{(A, B)}(s) = \infty$, then we have $\#^A(s_0 \cdot s) = \infty$ and $\#^B(s_0 \cdot s) = 0$, which further implies $\#^A(s) = \infty$ and $\#^B(s) = 0$.
		\item ($\impliedby$) We show the contrapositive.
			Assume that $\mathbb{P}_{s_0}[\mathbf{Streett}(A, B)] < 1$ for some $s_0$.
			Then, $\mathbb{P}_{s_0}[\#^A = \infty \land \#^B < \infty] > 0$.
			Since we have
			\[ \{ s \mid \#^A = \infty \land \#^B < \infty \} = \bigcup_{N = 0}^{\infty} \{ s \mid \#^A (s_{N + 1} s_{N + 2} \cdots) = \infty \land \#^B (s_{N + 1} s_{N + 2} \cdots) = 0 \}, \]
			there exists $N$ such that $\mathbb{P}_{s_0}[\#^A (s_{N + 1} s_{N + 2} \cdots) = \infty \land \#^B (s_{N + 1} s_{N + 2} \cdots) = 0] > 0$.
			Let $X_N = \{ s \mid \#^A (s_{N + 1} s_{N + 2} \cdots) = \infty \land \#^B (s_{N + 1} s_{N + 2} \cdots) = 0 \}$.
			\begin{align}
				\mathbb{P}_{s_0}[X_N] &= \int \chi_{X_N}(s) \,\mathrm{d} (F_{\omega})_{s_0}(s) \\
				&= \iint \chi_{X_N}(s_1 s_2 \cdots s_N s_{N + 1} \cdots) \,\mathrm{d} (F_{\omega})_{s_N}(s_{N + 1} s_{N + 2} \cdots) \,\mathrm{d} (F_N)_{s_0}(s_1 s_2 \cdots s_N) \\
				&= \int \mathbb{P}_{s_N}[\#^A = \infty \land \#^B = 0] \,\mathrm{d} (F_N)_{s_0}(s_1 s_2 \cdots s_N) \\
				&= \iint \mathbb{P}_{s_{N + 1}}[\mathrm{step}_{s_{N + 1}}^{(A, B)} = \infty] \,\mathrm{d} F_{s_N}(s_{N + 1}) \,\mathrm{d} (F_N)_{s_0}(s_1 s_2 \cdots s_N) \\
				&> 0
			\end{align}
			Therefore, there exists $s_{N + 1}$ such that $\mathbb{P}_{s_{N + 1}}[\mathrm{step}_{s_{N + 1}}^{(A, B)} = \infty] > 0$.
	\end{itemize}
	A similar result using the Markov property can be found in \cite[Proposition~9.1.1]{Meyn1996}.
\end{appendixproof}

\subsection{Generalised Streett Supermartingales}

In the literature on termination verification for probabilistic programs, it is known that ranking supermartingales \cite{ChakarovCAV2013} can be understood as a prefixed point \cite{TakisakaTOPLAS2021}.
Specifically, the expected number of steps until termination can be characterised as the least fixed point of a certain monotone function, and ranking supermartingales are its finite prefixed points.
By Park induction (or the Knaster--Tarski theorem), ranking supermartingales are sound and complete certificates for positive almost sure termination.

Here, we extend this idea to positive recurrence.
It is straightforward to see that the expected number $\mathbb{E}_{s_0}[\mathrm{step}_{s_0}^{(A, B)}]$ of $A$-steps until reaching $B$ can be characterised as a least fixed point as well.
By considering its finite prefixed points, we obtain a novel notion of supermartingales, which is sound and complete for positive recurrence.
It turns out that this new notion of supermartingales is a generalisation of Streett supermartingales in \cite{AbateCAV2024}.

We briefly recall some order-theoretic notions.
An \emph{$\omega$-complete partial order} (\emph{$\omega$-cpo}) is a poset in which every $\omega$-chain (i.e., countable increasing sequence) $\{ x_n \}_{n \in \mathbb{N}}$ has a least upper bound $\sup_n x_n$.
The main example of $\omega$-cpos in this paper is the set $\mathbf{Meas}(S, [0, \infty])$ of measurable functions from $S$ to $[0, \infty]$ ordered by the pointwise order.
A function between $\omega$-cpos is \emph{Scott continuous} if it preserves least upper bounds of $\omega$-chains.
If $X$ is an $\omega$-cpo with the least element $\bot$, then any Scott continuous function $f \colon X \to X$ has a least fixed point, which is given by $\mu f = \sup_n f^{n}(\bot)$.

\begin{proposition}\label{prop:expected-number-of-A-steps-as-lfp}
	The expected number of $A$-steps until reaching $B$ is expressed as the least fixed point of the following monotone function $K_{\mathbb{E}} \colon \mathbf{Meas}(S, [0, \infty]) \to \mathbf{Meas}(S, [0, \infty])$.
	\begin{equation}
		\mathbb{E}_{s_0}[\mathrm{step}_{s_0}^{(A, B)}]\ =\ \mu K_{\mathbb{E}} (s_0) \qquad\text{where}\qquad
		(K_{\mathbb{E}} \eta)(x)\ \coloneqq\ \begin{cases}
			(\nexttime \eta) (x) + 1 & x \in A \setminus B \\
			0 & x \in B \\
			(\nexttime \eta)(x) & \text{otherwise}
		\end{cases}
		\label{eq:streett-step}
	\end{equation}
\end{proposition}
\begin{appendixproof}[Proof of Proposition~\ref{prop:expected-number-of-A-steps-as-lfp}]
	It is easy to see the following properties.
	\begin{itemize}
		\item $\sup_n \mathrm{step}_{s_0, n}^{(A, B)} \comp \pi_n = \mathrm{step}_{s_0}^{(A, B)}$ where $\pi_n \colon S^{\omega} \to S^{n}$ is the projection onto the first $n$ steps.
		\item If $s_0 \in A \setminus B$, then $\mathrm{step}_{s_0, n}^{(A, B)}(s_1 \dots s_n) = 1 + \mathrm{step}_{s_1, n - 1}^{(A, B)}(s_2 \dots s_n)$.
		\item If $s_0 \in B$, then $\mathrm{step}_{s_0, n}^{(A, B)}(s_1 \dots s_n) = 0$.
		\item If $s_0 \notin A \cup B$, then $\mathrm{step}_{s_0, n}^{(A, B)}(s_1 \dots s_n) = \mathrm{step}_{s_1, n}^{(A, B)}(s_2 \dots s_n)$.
	\end{itemize}
	Now, we show the following claim by induction on $n$.
	\[ K_{\mathbb{E}}^{n + 1}(\bot)(s_0) = \mathbb{E}_{s_0}[\mathrm{step}_{s_0, n}^{(A, B)} \comp \pi_n] \]
	\begin{itemize}
		\item Base case ($n = 0$): we have $K_{\mathbb{E}}(\bot)(s_0) = 1$ if and only if $s_0 \in A \setminus B$.
		On the other hand, we have $\mathrm{step}_{s_0, 0}^{(A, B)}() = 1$ if and only if $s_0 \in A \setminus B$.
		\item Step case: 
		\begin{itemize}
			\item If $s_0 \in A \setminus B$, then
			\begin{align}
				K_{\mathbb{E}}^{n + 2}(\bot)(s_0) &= \nexttime (K_{\mathbb{E}}^{n + 1}(\bot))(s_0) + 1 \\
				&= \int K_{\mathbb{E}}^{n + 1}(\bot)(s_1) \,\mathrm{d} F_{s_0}(s_1) + 1 \\
				&= \int \mathbb{E}_{s_1}[\mathrm{step}_{s_1, n}^{(A, B)} \comp \pi_n] \,\mathrm{d} F_{s_0}(s_1) + 1 \\
				&= \int \dots \int \mathrm{step}_{s_1, n}^{(A, B)}(s_2 \dots s_{n + 1}) \,\mathrm{d} F_{s_n}(s_{n+1}) \dots \,\mathrm{d} F_{s_0}(s_1) + 1 \\
				&= \int \dots \int \mathrm{step}_{s_0, n + 1}^{(A, B)}(s_1 s_2 \dots s_{n + 1}) \,\mathrm{d} F_{s_n}(s_{n+1}) \dots \,\mathrm{d} F_{s_0}(s_1) \\
				&= \mathbb{E}_{s_0}[\mathrm{step}_{s_0, n + 1}^{(A, B)} \comp \pi_{n + 1}]
			\end{align}
			\item If $s_0 \in B$, then
			\[ K_{\mathbb{E}}^{n + 2}(\bot)(s_0) = 0 = \mathbb{E}_{s_0}[\mathrm{step}_{s_0, n + 1}^{(A, B)} \comp \pi_{n + 1}] \]
			\item If $s_0 \notin A \cup B$, then
			\begin{align}
				K_{\mathbb{E}}^{n + 2}(\bot)(s_0) &= \nexttime (K_{\mathbb{E}}^{n + 1}(\bot))(s_0) \\
				&= \int K_{\mathbb{E}}^{n + 1}(\bot)(s_1) \,\mathrm{d} F_{s_0}(s_1) \\
				&= \int \mathbb{E}_{s_1}[\mathrm{step}_{s_1, n}^{(A, B)} \comp \pi_n] \,\mathrm{d} F_{s_0}(s_1) \\
				&= \int \dots \int \mathrm{step}_{s_1, n}^{(A, B)}(s_2 \dots s_{n + 1}) \,\mathrm{d} F_{s_n}(s_{n+1}) \dots \,\mathrm{d} F_{s_0}(s_1) \\
				&= \int \dots \int \mathrm{step}_{s_0, n + 1}^{(A, B)}(s_1 s_2 \dots s_{n + 1}) \,\mathrm{d} F_{s_n}(s_{n+1}) \dots \,\mathrm{d} F_{s_0}(s_1) \\
				&= \mathbb{E}_{s_0}[\mathrm{step}_{s_0, n + 1}^{(A, B)} \comp \pi_{n + 1}]
			\end{align}
		\end{itemize}
	\end{itemize}
	By the monotone convergence theorem, we have
	\[ \mu K_{\mathbb{E}} (s_0) = \sup_n K_{\mathbb{E}}^n(\bot)(s_0) = \mathbb{E}_{s_0}[\mathrm{step}_{s_0}^{(A, B)}]. \qedhere \]
\end{appendixproof}

We now introduce new supermartingales as finite prefixed points of $K_{\mathbb{E}}$ in~\eqref{eq:streett-step}.
\begin{definition}\label{def:generalised-streett-supermartingale}
	Let $F \colon S \to \Giry S$ be a Streett Markov chain with a Streett pair $(A, B)$.
	A \emph{generalised Streett supermartingale} (GSSM) is a finite-valued measurable function $r \colon S \to [0, \infty)$ that is a prefixed point of the function $K_{\mathbb{E}}$ in~\eqref{eq:streett-step}:
	concretely, $r$ is such that
	\[ (\nexttime r)(x) \le r(x) - \Characteristic{A \setminus B}(x) \qquad\text{for any $x \notin B$} \]
	where $\Characteristic{A \setminus B} \colon S \to \{ 0, 1 \}$ is the characteristic function of $A \setminus B$.
\end{definition}
By Park induction, GSSMs witness positive recurrence and hence almost sure satisfaction of the Streett condition.

\begin{theorem}\label{thm:gssm-sound-complete}
	Let $F \colon S \to \Giry S$ be a Streett Markov chain with a Streett pair $(A, B)$.
	There exists a generalised Streett supermartingale if and only if $F$ is positively $(A, B)$-recurrent from any initial state.
\end{theorem}
\begin{proof}
	If there exists a generalised Streett supermartingale $r$, then by Park induction, we have $\mathbb{E}_{s_0}[\mathrm{step}_{s_0}^{(A, B)}] \le r(s_0) < \infty$.
	Conversely, if $F$ is positively $(A, B)$-recurrent, then the measurable function $s_0 \mapsto \mathbb{E}_{s_0}[\mathrm{step}_{s_0}^{(A, B)}]$ is finite, and by Proposition~\ref{prop:expected-number-of-A-steps-as-lfp}, it is also a prefixed point of $K_{\mathbb{E}}$.
\end{proof}

\begin{corollary}
	If a GSSM exists, then the Streett condition is satisfied almost surely from any initial state.
	\qed
\end{corollary}

\begin{example}\label{ex:generalised-streett-supermartingale}
	Consider the Markov chain $F \colon \mathbb{N} \to \Giry \mathbb{N}$ defined by $F(x) \coloneqq \delta_{x - 1}$ if $x > 0$ and $F(x) \coloneqq \sum_{i = 0}^{\infty} 2^{-(i + 1)} \delta_{2^{i + 1}}$ if $x = 0$.
	That is, if the current position is $x > 0$, then the next position is $x - 1$ deterministically; and if the current position is $0$, then the next position is $2^{i + 1}$ with probability $2^{-(i + 1)}$ for each $i \in \mathbb{N}$.
	If $A = \{ x \in \mathbb{N} \mid x > 0 \}$ and $B = \{ 0 \}$, then $r(x) = x$ is a GSSM because $r$ satisfies
	$x > 0 \implies r(x - 1) \le r(x) - 1$.
	On the other hand, if $A = \{ x \in \mathbb{N} \mid x \ge 2 \}$ and $B = \{ 1 \}$, then a GSSM $r$ must satisfy $x \ge 2 \implies r(x - 1) \le r(x) - 1$ and $x = 0 \implies \sum_{i = 0}^{\infty} 2^{-(i + 1)} r(2^{i + 1}) \le r(x)$.
	However, this is impossible because the Markov chain is not positively $(A, B)$-recurrent when $x = 0$ is the initial state:
	$\mathbb{E}_0[\mathrm{step}_{0}^{(A, B)}] = \sum_{i = 0}^{\infty} 2^{-(i + 1)} \cdot (2^{i + 1} - 1) = \infty$.
	\qed
\end{example}

\subsection{Comparison with Streett Supermartingales}

By definition, it is clear that GSSMs (Definition~\ref{def:generalised-streett-supermartingale}) generalise SSM (Definition~\ref{def:streett-ranking-supermartingale}) by dropping the requirement for $x \in B$.
This relaxation strictly increases the verification power of GSSMs compared to SSMs.
The following example shows the difference between two notions of supermartingales.

\begin{example}\label{ex:no-streett-supermartingale}
	Consider the following (deterministic) program where $n, m$ are natural numbers.
	\[ \while{\mathtt{true}}{\assign{m}{n}; \while{m > 0}{\assign{m}{m - 1}}; \assign{n}{n + 1}} \]
	This program corresponds to the pCFG shown below.
	Here, we merge locations for conditional branching and assignments for simplicity.
	The pCFG defines a Markov chain $F \colon S \to \Giry S$ over the state space $S = \{ l_0, l_1 \} \times \mathbb{N}^2$.
	
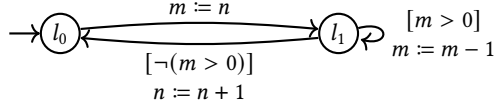
\begin{center}
		\begin{tikzpicture}[scale=0.9, transform shape]
			\node[initial, state] (s0) {$l_0$};
			\node[state, right=10em of s0] (s1) {$l_1$};
			\draw[->] (s0) edge[bend left=5, transform canvas={yshift=1pt}] node[above] {$\assign{m}{n}$} (s1);
			\draw[->] (s1) edge[bend left=5, transform canvas={yshift=-1pt}] node[below, align=center] {$[\lnot (m > 0)]$\\$\assign{n}{n + 1}$} (s0);
			\draw[->] (s1) edge[loop right] node[right, align=center] {$[m > 0]$\\$\assign{m}{m - 1}$} (s1);
		\end{tikzpicture}
	\end{center}
	Consider the Streett condition given by $(A, B) = (\{ l_1 \} \times \mathbb{N}^2, \{ l_0 \} \times \mathbb{N}^2)$.
	Then, there exists a GSSM $r \colon S \to [0, \infty)$ for this Streett Markov chain as follows.
	\[ r(l_0, m, n) = 0 \qquad\qquad r(l_1, m, n) = m + 1 \tag*{\qed} \]
\end{example}

Now, consider SSMs for the Streett Markov chain in Example~\ref{ex:no-streett-supermartingale}.
For $r$ to be an SSM, it must satisfy the following additional condition:
\[ \exists M > 0,\qquad \forall m, n \in \mathbb{N},\qquad r(l_1, n, n) \quad\le\quad r(l_0, m, n) + M. \]
The GSSM presented in Example~\ref{ex:no-streett-supermartingale} is not an SSM because the variable $n$ becomes bigger and bigger as the program executes, and the inequality $m + 1 \le 0 + M$ will be eventually violated for any fixed $M$.
Similarly, we can show that it is impossible to satisfy this condition no matter how we define $r$.

\begin{proposition}\label{prop:no-streett-ranking-supermartingale}
	There is no Streett supermartingale for Example~\ref{ex:no-streett-supermartingale}.
\end{proposition}
\begin{proofsketch}
	Suppose that there exists an SSM $r \colon S \to [0, \infty)$.
	Then, $r$ decreases by at least $\epsilon$ whenever $l_1$ is visited, while it increases by at most $M$ whenever $l_0$ is visited.
	In each iteration of the outer loop, the inner loop is executed $n$ times, and $n$ is incremented by 1 afterward.
	The larger $n$ becomes, the more times $r$ is decreased by $\epsilon$.
	Thus, after sufficiently many iterations, the non-negativity of $r$ is eventually violated.
	A more formal proof is given in
	\iflong
	the appendix.
	\else
	the long version \cite{arxiv}.
	\fi
\end{proofsketch}
\begin{appendixproof}[Proof of Proposition~\ref{prop:no-streett-ranking-supermartingale}]
	Suppose that there exists a Streett supermartingale $r \colon S \to [0, \infty)$.
	Then, the following conditions must hold.
	\begin{align}
		&r(l_0, m, n) + M \ge r(l_1, n, n) \\
		m > 0 \implies &r(l_1, m, n) \ge \epsilon + r(l_1, m - 1, n) \\
		m = 0 \implies &r(l_1, m, n) \ge \epsilon + r(l_0, m, n + 1)
	\end{align}
	For each $n$, we have the following.
	\begin{align}
		r(l_0, 0, n) &\ge r(l_1, n, n) - M \\
		&\ge r(l_1, n - 1, n) + \epsilon - M \\
		&\ge \dots \\
		&\ge r(l_1, 0, n) + n \epsilon - M \\
		&\ge r(l_0, 0, n + 1) + (n + 1) \epsilon - M
	\end{align}
	Thus, $r(l_0, 0, 0)$ satisfies the following condition for each $n$.
	\begin{align}
		r(l_0, 0, 0) &\ge r(l_0, 0, 1) + \epsilon - M \\
		&\ge \dots \\
		&\ge r(l_0, 0, n) + \sum_{i = 1}^n i \epsilon - n \cdot M \\
		&= r(l_0, 0, n) + n \left( \frac{n + 1}{2} \epsilon - M \right)
	\end{align}
	For sufficiently large $n$, we have
	\[ r(l_0, 0, 0) < n \left( \frac{n + 1}{2} \epsilon - M \right), \]
	which implies that $r(l_0, 0, n) < 0$.
	This contradicts the non-negativity of $r_{l_0}$.
\end{appendixproof}

\section{Null Recurrence and Streett Supermartingales}
\label{sec:null-recurrence}

Similarly to the case of positive recurrence, null recurrence can be characterised as the least fixed point of a certain monotone function, following the idea of \emph{distribution-valued ranking supermartingales} in \cite{UrabeLICS2017}.
This leads to a sound and complete supermartingale-based certificate for Streett conditions, which we call \emph{distribution-valued Streett supermartingales} (DVSSMs).
Although this result is theoretically interesting, it has a drawback in practical applicability, since there is no existing synthesis algorithm for distribution-valued supermartingales as far as we know.
To address this issue, we also introduce \emph{lexicographic generalised Streett supermartingales} (LexGSSMs).
Although LexGSSMs are not complete for null recurrence, they are more amenable to automatic synthesis and able to prove qualitative problems beyond positive recurrence.

\subsection{Distribution-Valued Streett Supermartingales}

Following \cite{UrabeLICS2017}, we define an $\omega$-cpo of distributions over $\mathbb{N}_{\infty} = \mathbb{N} \cup \{ \infty \}$ and characterise the distribution of the number of $A$-steps until reaching $B$ as a least fixed point.

\begin{definition}
	Consider the set $\mathbb{N}_{\infty}$ of extended natural numbers equipped with discrete $\sigma$-algebra.
	Let $\Giry \mathbb{N}_{\infty}$ be the set of all distributions over $\mathbb{N}_{\infty}$ as defined in Section~\ref{sec:markov-chains}.
	We define an order relation $\le$ on $\Giry \mathbb{N}_{\infty}$ as follows, which is known as the \emph{stochastic order}.
	\[ \zeta_1 \le \zeta_2 \quad\iff\quad \forall a \in \mathbb{N}, \zeta_1([a, \infty]) \le \zeta_2([a, \infty]) \]
	Here, $[a, \infty] = \{ n \in \mathbb{N}_{\infty} \mid n \ge a \}$ denotes the interval.
	Then, $\Giry \mathbb{N}_{\infty}$ is an $\omega$cpo: 
	the least upper bound of $\zeta_0 \le \zeta_1 \le \dots$ is such that $(\sup_n \zeta_n)([a, \infty]) = \sup_n \zeta_n([a, \infty])$ for any $a \in \mathbb{N}$.
	\[ (\sup_n \zeta_n)(a) = \sup_n \zeta_n([a, \infty]) - \sup_n \zeta_n([a + 1, \infty]) \qquad (\sup_n \zeta_n)(\infty) = \inf_a \sup_n \zeta_n([a, \infty]) \]
	The Dirac distribution at 0, denoted by $\delta_0$, is the least element of $\Giry \mathbb{N}_{\infty}$.
\end{definition}

Let $\mathbf{Meas}(S, \Giry \mathbb{N}_{\infty})$ be the set of measurable functions from $S$ to $\Giry \mathbb{N}_{\infty}$ ordered by the pointwise order.
We define a monotone function on $\Giry \mathbb{N}_{\infty}$ similarly to positive recurrence (Proposition~\ref{prop:expected-number-of-A-steps-as-lfp}).

\begin{definition}
	We define the \emph{next-time operator} $\nexttime_F \colon \mathbf{Meas}(S, \Giry \mathbb{N}_{\infty}) \to \mathbf{Meas}(S, \Giry \mathbb{N}_{\infty})$ of a Markov chain $F \colon S \to \Giry S$, and an action $\oplus$ of the monoid $(\mathbb{N}, 0, {+})$ on $\Giry \mathbb{N}_{\infty}$ as follows.
	\[ (\nexttime_F \zeta)(s)(m) \quad\coloneqq\quad \int \zeta(s')(m) \,\mathrm{d} F_s(s') \qquad\qquad
	(n \oplus \zeta)(m) \quad\coloneqq\quad \begin{cases}
		\zeta(m - n) & m \ge n \\
		0 & m < n
	\end{cases} \]
	Intuitively, given a distribution $\zeta$ of the number of steps, $n \oplus \zeta$ is the distribution obtained by adding $n$ to each outcome of $\zeta$.
	We extend $\oplus$ to an action of $(\mathbb{N}, 0, {+})$ on $\mathbf{Meas}(S, \Giry \mathbb{N}_{\infty})$ pointwise.
\end{definition}

\begin{proposition}\label{prop:streett-distribution-lfp}
	Let $F \colon S \to \Giry S$ be a Streett Markov chain with a Streett pair $(A, B)$.
	The distribution of the number of $A$-steps until reaching $B$ is characterised as the least fixed point of the monotone function $K_{\mathbb{P}} \colon \mathbf{Meas}(S, \Giry \mathbb{N}_{\infty}) \to \mathbf{Meas}(S, \Giry \mathbb{N}_{\infty})$ defined by $K_{\mathbb{P}} (r)(x) \coloneqq 1 \oplus \nexttime_F (r)(x)$ if $x \in A \setminus B$, $K_{\mathbb{P}} (r)(x) \coloneqq \delta_{0}$ if $x \in B$, and $K_{\mathbb{P}} (r)(x) \coloneqq \nexttime_F (r)(x)$ otherwise.
	\[ \mathbb{P}_{s_0}[\mathrm{step}_{s_0}^{(A, B)} = m]\ =\ \mu K_{\mathbb{P}}(s_0)(m) \]
\end{proposition}
\begin{proofsketch}
	It is straightforward to show that $K_{\mathbb{P}}^{n}(\bot)$ gives the $n$-step approximation of the distribution of the number of $A$-steps until reaching $B$.
	Then, by taking the limit, we obtain the desired result.
	Details of the proof are given in
	\iflong
	the appendix.
	\else
	the long version \cite{arxiv}.
	\fi
\end{proofsketch}
\begin{appendixproof}[Proof of Proposition~\ref{prop:streett-distribution-lfp}]
	We show the following claim by induction on $n$.
	\[ K_{\mathbb{P}}^{n + 1}(\bot)(s_0)(m) \quad=\quad \mathbb{P}_{s_0}[\mathrm{step}_{s_0, n}^{(A, B)} \comp \pi_n = m] \]
	\begin{itemize}
		\item Base case ($n = 0$): we have $K_{\mathbb{P}}(\bot)(s_0) = \delta_1$ if $s_0 \in A \setminus B$ and $K_{\mathbb{P}}(\bot)(s_0) = \delta_0$ otherwise.
		On the other hand, we have $\mathrm{step}_{s_0, 0}^{(A, B)}() = 1$ if $s_0 \in A \setminus B$ and $\mathrm{step}_{s_0, 0}^{(A, B)}() = 0$ otherwise.
		\item Step case:
		\begin{itemize}
			\item If $s_0 \in A \setminus B$, then $\mathrm{step}_{s_0, n + 1}^{(A, B)}(s_1 \dots s_{n + 1}) = 1 + \mathrm{step}_{s_1, n}^{(A, B)}(s_2 \dots s_{n + 1})$.
			\begin{itemize}
				\item If $m = 0$, then
				\[ K_{\mathbb{P}}^{n + 2}(\bot)(s_0)(m) = 0 = \mathbb{P}_{s_0}[\mathrm{step}_{s_0, n + 1}^{(A, B)} \comp \pi_{n + 1} = 0] \]
				\item If $m > 0$, then
				\begin{align}
					K_{\mathbb{P}}^{n + 2}(\bot)(s_0)(m) &= (1 \oplus \nexttime_F (K_{\mathbb{P}}^{n + 1}(\bot))(s_0))(m) \\
					&= \nexttime_F (K_{\mathbb{P}}^{n + 1}(\bot))(s_0)(m - 1) \\
					&= \int K_{\mathbb{P}}^{n + 1}(\bot)(s_1)(m - 1) \,\mathrm{d} F_{s_0}(s_1) \\
					&= \int \mathbb{P}_{s_1}[\mathrm{step}_{s_1, n}^{(A, B)} \comp \pi_n = m - 1] \,\mathrm{d} F_{s_0}(s_1) \\
					&= \mathbb{P}_{s_0}[\mathrm{step}_{s_0, n + 1}^{(A, B)} \comp \pi_{n + 1} = m]
				\end{align}
			\end{itemize}
			\item If $s_0 \in B$, then $\mathrm{step}_{s_0, n + 1}^{(A, B)}(s_1 \dots s_{n + 1}) = 0$.
			\[ K_{\mathbb{P}}^{n + 2}(\bot)(s_0)(m) = \delta_0(m) = \mathbb{P}_{s_0}[\mathrm{step}_{s_0, n + 1}^{(A, B)} \comp \pi_{n + 1} = m] \]
			\item If $s_0 \notin A \cup B$, then $\mathrm{step}_{s_0, n + 1}^{(A, B)}(s_1 \dots s_{n + 1}) = \mathrm{step}_{s_1, n}^{(A, B)}(s_2 \dots s_{n + 1})$.
			\begin{align}
				K_{\mathbb{P}}^{n + 2}(\bot)(s_0)(m) &= \nexttime_F (K_{\mathbb{P}}^{n + 1}(\bot))(s_0)(m) \\
				&= \int K_{\mathbb{P}}^{n + 1}(\bot)(s_1)(m) \,\mathrm{d} F_{s_0}(s_1) \\
				&= \int \mathbb{P}_{s_1}[\mathrm{step}_{s_1, n}^{(A, B)} \comp \pi_n = m] \,\mathrm{d} F_{s_0}(s_1) \\
				&= \mathbb{P}_{s_0}[\mathrm{step}_{s_0, n + 1}^{(A, B)} \comp \pi_{n + 1} = m]
			\end{align}
		\end{itemize}
	\end{itemize}
	Now, we have
	\[ \sup_n K_{\mathbb{P}}^{n}(\bot)(s_0)([m, \infty]) = \sup_n \mathbb{P}_{s_0}[\mathrm{step}_{s_0, n}^{(A, B)} \comp \pi_n \ge m] = \mathbb{P}_{s_0}[\mathrm{step}_{s_0}^{(A, B)} \ge m] \]
	and thus
	\[ \mu K_{\mathbb{P}} (s_0)(m) = \mathbb{P}_{s_0}[\mathrm{step}_{s_0}^{(A, B)} = m]. \qedhere \]
\end{appendixproof}
Now, we define a new supermartingale as a ``finite'' prefixed point of $K_{\mathbb{P}}$.
\begin{definition}\label{def:distribution-valued-streett-supermartingale}
	A \emph{distribution-valued Streett supermartingale} (DVSSM) is a prefixed point $r$ of $K_{\mathbb{P}}$ such that $r(s)([0, \infty)) = 1$ for all $s \in S$.
	Specifically, $r$ is a measurable function $r \colon S \to \Giry \mathbb{N}$ such that the following condition holds.
	\[ \forall x \in A \setminus B,\quad 1 \oplus \nexttime(r)(x) \le r(x) \qquad
	\forall x \notin A \cup B,\quad \nexttime(r)(x) \le r(x) \]
\end{definition}

\begin{corollary}\label{cor:dvssm-sound-complete}
	DVSSMs are sound and complete for almost sure satisfaction of Streett conditions.
\end{corollary}
\begin{proof}
	By Proposition~\ref{prop:streett-distribution-lfp} and Lemma~\ref{lem:streett-null-recurrence}.
\end{proof}

\subsection{A Lexicographic Extension of Generalised Streett Supermartingales}\label{sec:lexgssm}

Since DVSSMs take values in distributions, they are difficult to synthesise automatically.
We therefore introduce another type of supermartingales called \emph{lexicographic generalised Streett supermartingales} (LexGSSMs).
LexGSSMs are, in short, defined as a combination of lexicographic ranking supermartingales \cite{AgrawalPOPL2018} and generalised Streett supermartingales (Definition~\ref{def:generalised-streett-supermartingale}).
To explain the idea, let us first recall lexicographic ranking supermartingales.

Lexicographic ranking supermartingales \cite{AgrawalPOPL2018} be used to verify not only positive almost-sure termination but also almost-sure termination of probabilistic programs.
They combine ideas from lexicographic ranking functions and ranking supermartingales.
Due to the probabilistic nature of target programs, the definition of ``lexicographic ordering'' used for lexicographic ranking supermartingales is slightly different from the standard one used for lexicographic ranking functions.

\begin{definition}\label{def:lexicographic-ranking-supermartingale-order}
	We define a ``strict'' lexicographic ordering $\succ$ on $[0, \infty)^m$ as follows: for each $(r_1, \dots, r_m), (r'_1, \dots, r'_m) \in [0, \infty)^m$,
	\begin{align}
		(r_1, \dots, r_m) \succ_{[l]} (r'_1, \dots, r'_m) \qquad&\coloneqq\qquad (\forall i, 1 \le i < l \implies r_i \ge r'_i) \land r_l \ge 1 + r'_l \\
		(r_1, \dots, r_m) \succ (r'_1, \dots, r'_m) \qquad&\coloneqq\qquad \exists l \in \OneToN{m}, (r_1, \dots, r_m) \succ_{[l]} (r'_1, \dots, r'_m).
	\end{align}
	If $(r_1, \dots, r_m) \succ_{[l]} (r'_1, \dots, r'_m)$, then we say that $l$ is a \emph{level} of $(r_1, \dots, r_m) \succ (r'_1, \dots, r'_m)$.
	We also define a ``non-strict'' lexicographic ordering on $[0, \infty)^m$ by the union $({\succeq}) \coloneqq ({\succ}) \cup ({\ge})$ where $(r_1, \dots, r_m) \ge (r'_1, \dots, r'_m)$ is defined as $r_i \ge r'_i$ for each $i$.
\end{definition}
Both $\succ$ and $\succeq$ are slightly different from the standard lexicographic orderings in the sense that (a) $\succ$ requires a decrease by at least $1$ in the first differing component, and (b) $\succeq$ is defined as the union with $\ge$ instead of $=$.
Lexicographic ranking supermartingales are defined as a function that decreases with respect to $\succ$ in expectation.

\begin{definition}
	Let $F \colon S \to \Giry S$ be a Markov chain and $S' \subseteq S$ be a measurable subset.
	A \emph{lexicographic ranking supermartingale} is a measurable function $r \colon S \to [0, \infty)^m$ such that for any $x \in S'$, $r(x) \succ (\nexttime r)(x)$ holds where the next-time operator $\nexttime$ is applied to $r$ pointwise.
\end{definition}

\begin{proposition}\label{prop:lexicographic-ranking-supermartingale-soundness}
	The existence of a lexicographic ranking supermartingale witnesses that the probability of staying in $S'$ for infinitely many steps is $0$, i.e., witnesses almost sure reachability to $S \setminus S'$ \cite[Theorem~3.4]{AgrawalPOPL2018}.
	\qed
\end{proposition}

We adapt this idea to Streett conditions to address null recurrence.
The idea is similar to GSSMs (Definition~\ref{def:generalised-streett-supermartingale}): given a Streett pair $(A, B)$, we require a strict decrease for states in $A \setminus B$, a non-strict decrease for states not in $A \cup B$, and do not require any condition for states in $B$.

\begin{definition}\label{def:lexicographic-generalised-streett-supermartingale}
	Let $F \colon S \to \Giry S$ be a Streett Markov chain with a Streett pair $(A, B)$.
	A \emph{lexicographic generalised Streett supermartingale} (LexGSSM) is a function $r \colon S \to [0, \infty)^m$ such that for any $x \in A \setminus B$, we have $(\nexttime r)(x) \prec r(x)$ and for any $x \notin A \cup B$, we have $(\nexttime r)(x) \preceq r(x)$.
\end{definition}

\begin{theorem}
	Let $F \colon S \to \Giry S$ be a Streett Markov chain with a Streett pair $(A, B)$.
	If there exists a LexGSSM, then the Streett condition is almost surely satisfied from any initial state.
\end{theorem}
\begin{proof}
	This is a consequence of the equivalence with lexicographic progress-measure supermartingales and their soundness (Proposition~\ref{prop:lexgssm-to-lexpmsm} and Theorem~\ref{thm:lexicographic-progress-measure}), which we will show later.
\end{proof}

Obviously, LexGSSMs subsume GSSMs, and thus the class of problems verifiable by LexGSSMs is a super-class of positive recurrence.
It is actually a proper super-class, as shown below.
\begin{example}\label{ex:lexgssm}
	Consider the Markov chain in Example~\ref{ex:generalised-streett-supermartingale} again, and suppose that a Streett condition is given as $A = \{ x \in \mathbb{N} \mid x \ge 2 \}$ and $B = \{ 1 \}$.
	In this case, there is no GSSM as we have seen in Example~\ref{ex:generalised-streett-supermartingale}.
	However, there exists a LexGSSM $r \colon \mathbb{N} \to [0, \infty)^2$ defined by $r(0) = (1, 0)$ and $r(x) = (0, x - 1)$ for $x > 0$.
	Note that the condition for LexGSSMs is satisfied even when $x = 0 \notin A \cup B$: $\nexttime r(0) = (0, \infty) \preceq r(0) = (1, 0)$.
	\qed
\end{example}

On the other hand, LexGSSMs are not complete for null recurrence.
\begin{example}\label{ex:lexgssm-not-complete}
	Consider the symmetric random walk on $\mathbb{Z}$ with a Streett condition given as $A = \mathbb{Z} \setminus \{ 0 \}$ and $B = \{ 0 \}$.
	This is the well-known recurrence property of the 1-dimensional symmetric random walk.
	Although the Streett condition is almost surely satisfied from any initial state, there is no LexGSSM for this problem.
	The proof is given in \referappendix{Proposition}{B.1}{prop:lexgssm-not-complete}.
	\qed
\end{example}

\begin{remark}
	The notions of GSSM, LexGSSM, and DVSSM can be understood in analogy with their counterparts for (almost-sure) termination: ranking supermartingales, lexicographic ranking supermartingales, and distribution-valued ranking supermartingales, respectively.
	In this correspondence, positive recurrence for Streett conditions corresponds to positive almost-sure termination, for which GSSMs and ranking supermartingales provide sound and complete certificates.
	Similarly, null recurrence corresponds to almost-sure termination, with DVSSMs and distribution-valued ranking supermartingales serving as sound and complete certificates.
	LexGSSMs lie between positive and null recurrence, just as lexicographic ranking supermartingales lie between positive almost-sure termination and almost-sure termination.
\end{remark}

\begin{toappendix}
\begin{proposition}\label{prop:lexgssm-not-complete}
	There is no LexGSSM for the recurrence property of the 1-dimensional symmetric random walk (Example~\ref{ex:lexgssm-not-complete}).
\end{proposition}
\begin{proof}
	We prove this by contradiction.
	Suppose that there exists a LexGSSM $r \colon \mathbb{Z} \to [0, \infty)^m$.
	Then, for any $x \in \mathbb{Z} \setminus \{ 0 \}$, we have
	\[ r(x) \quad\succ\quad \frac{1}{2} r(x - 1) + \frac{1}{2} r(x + 1) \]
	Let $l$ be the minimum level of the above lexicographic ordering that occurs infinitely many times for $x > 0$.
	Then, the $l$-th component $r_l \colon \mathbb{Z} \to [0, \infty)$ satisfies $r_l(x) \ge \frac{1}{2} r_l(x - 1) + \frac{1}{2} r_l(x + 1)$ for all $x > 0$ and $r_l(x) \ge 1 + \frac{1}{2} r_l(x - 1) + \frac{1}{2} r_l(x + 1)$ for infinitely many $x > 0$.
	Let $\Delta r_l(x) \coloneqq r_l(x + 1) - r_l(x)$.
	Then, we have $\Delta r_l(x - 1) \ge \Delta r_l(x)$ for all $x > 0$ and $\Delta r_l(x - 1) \ge 2 + \Delta r_l(x)$ for infinitely many $x > 0$.
	For sufficiently large $x$, this implies $\Delta r_l(x) \le -1$.
	This further implies that for sufficiently large $x$, $r_l(x)$ becomes negative, which contradicts the non-negativity of $r_l$.
\end{proof}
\end{toappendix}

\section{Progress-Measure Supermartingales}
\label{sec:progress-measure}

So far, we have considered Streett conditions.
We now turn to another acceptance condition, namely, parity conditions.
In the literature of model checking, \emph{parity progress measures} \cite{JurdzinskiSTACS2000} are well-known as certificates for parity conditions on non-probabilistic systems.
This naturally raises the question of whether parity progress measures can be extended to Markov chains, and, if so, how powerful they are as certificates for $\omega$-regular verification of Markov chains.
To the best of our knowledge, this question has not been addressed in the existing literature.

In this section, we give an affirmative answer to this question by proposing a new type of supermartingales, which we call \emph{progress-measure supermartingales} (PMSMs).
Furthermore, we show that the class of properties verifiable by a lexicographic extension of PMSMs, which we call \emph{lexicographic progress-measure supermartingales} (LexPMSMs), is equal to the class verifiable by LexGSSMs.
We show this results by providing mutual translations between LexPMSMs and LexGSSMs.

\subsection{Preliminaries on Parity Progress Measures}\label{sec:preliminaries-parity-progress-measures}

We review the notion of parity progress measures for parity graphs~\cite{JurdzinskiSTACS2000} and give an idea of how to extend it to Markov chains.
For brevity, we often refer to parity progress measures simply as \emph{progress measures} below.

A \emph{parity graph} is a tuple $(V, E, p)$ where $V$ is a set of vertices, $E \subseteq V \times V$ is a set of edges, and $p \colon V \to \OneToN{d}$ is a priority function.
Consider the problem of verifying that any infinite walk in a parity graph (i.e., an infinite sequence of vertices $v_1, v_2, \dots$ such that $(v_i, v_{i + 1}) \in E$ for each $i$) satisfies the parity condition.

The idea of progress measures is to assign a $d$-dimensional vector of natural numbers to each vertex in the graph such that the vectors decrease with respect to lexicographic ordering truncated by the priorities of vertices along edges.

\begin{definition}
	The strict lexicographic ordering $>_{\mathrm{lex}}$ on $\mathbb{N}^d$ is defined in the standard way.
	\[ (n_1, n_2, \dots, n_d) >_{\mathrm{lex}} (n'_1, n'_2, \dots, n'_d) \quad\iff\quad \exists l, (\forall i, 1 \le i < l \implies n_i = n'_i) \land n_l > n'_l \]
	The non-strict lexicographic ordering $\ge_{\mathrm{lex}}$ is defined by $({\ge}_{\mathrm{lex}}) \coloneqq ({>}_{\mathrm{lex}}) \cup ({=})$.
	The truncation to the first $i$ components is denoted by $(n_1, n_2, \dots, n_d) {\downharpoonright} i = (n_1, n_2, \dots, n_i)$.
	For each $i \in \OneToN{d}$ and $x, y \in \mathbb{N}^d$, we define \emph{non-strict truncated lexicographic ordering} $x \ge_i y$ by $x {\downharpoonright} i \ge_{\mathrm{lex}} y {\downharpoonright} i$, and the strict one $x >_i y$ is defined similarly.
\end{definition}

\begin{definition}
	A \emph{progress measure} for a parity graph $(V, E, p)$ is a function $r \colon V \to \mathbb{N}^d$ such that for each $(v, w) \in E$, we have $r(v) \ge_{p(v)} r(w)$ if $p(v)$ is even and $r(v) >_{p(v)} r(w)$ if $p(v)$ is odd \cite[Definition~3]{JurdzinskiSTACS2000}.
\end{definition}

The strict inequality for odd priorities prohibits odd priorities from being the minimum priority that occurs infinitely many times.
We give a proof of the soundness of progress measures here, as we will later extend them to the probabilistic case so that the same argument can be applied.

\begin{proposition}
	If there exists a progress measure $r$ for a parity graph $(V, E, p)$, then any infinite walk in the graph satisfies the parity condition \cite[Proposition~4]{JurdzinskiSTACS2000}.
\end{proposition}
\begin{proof}
	We slightly modify the original proof so that it does not depend on the finiteness of the graph.
	We prove by contradiction.
	Let $v_1 \to v_2 \to \dots$ be an infinite walk in the graph and $p_{\min}$ be the minimum priority of vertices that occurs infinitely many times in the walk.
	Assume that $p_{\min}$ is odd.
	Consider the first $p_{\min}$ components of the progress measure $r$.
	Since the lexicographic ordering is preserved by truncation (Lemma~\ref{lem:truncation-lexicographic-ordering}), and by definition of a progress measure, we have $r(v_1) \ge_{p_{\min}} r(v_2) \ge_{p_{\min}} \dots$.
	Here, the strict inequality $p(v_i) >_{p_{\min}} p(v_{i + 1})$ holds for infinitely many $i$'s because $p_{\min}$ is odd and $p(v_i) = p_{\min}$ holds for infinitely many $i$'s.
	This contradicts the well-foundedness of the strict lexicographic ordering $>_{p_{\min}}$.
\end{proof}

By analysing the proof above, we can see that the essence of progress measures is the use of a strict and non-strict lexicographic ordering such that (1) the strict one prohibits infinite descending chains and (2) truncation preserves the non-strict lexicographic ordering in the following sense.

\begin{lemma}\label{lem:truncation-lexicographic-ordering}
	If $i \le j$ and $x \ge_j y$, then $x \ge_i y$.
	\qed
\end{lemma}

\subsection{Progress Measures for Parity Markov Chains}

As we have discussed in Section~\ref{sec:preliminaries-parity-progress-measures}, progress measures are defined using truncated lexicographic orderings.
Now, we adapt this idea to Markov chains by replacing the standard lexicographic ordering with the one used for lexicographic ranking supermartingales (Definition~\ref{def:lexicographic-ranking-supermartingale-order}).
We define a binary relation $\succeq_i$ on $[0, \infty)^d$ by $x \succeq_i y$ if $x {\downharpoonright} i \succeq y {\downharpoonright} i$ where $i \in \OneToN{d}$.
The strict version $\succ_i$ is defined similarly.
Then, we define progress measures for Markov chains as follows.

\begin{definition}\label{def:progress-measure}
	Let $F \colon S \to \Giry S$ be a parity Markov chain with a priority function $p \colon S \to \OneToN{d}$.
	A \emph{progress-measure supermartingale} (PMSM) is a non-negative measurable function $r \colon S \to [0, \infty)^d$ such that the following conditions hold:
	for each $x \in S$, if $p(x)$ is even, then $r(x) \succeq_{p(x)} (\nexttime r)(x)$; and if $p(x)$ is odd, then $r(x) \succ_{p(x)} (\nexttime r)(x)$.
\end{definition}

\begin{theorem}\label{thm:progress-measure}
	Let $F \colon S \to \Giry S$ be a parity Markov chain with a priority function $p$.
	If there exists a PMSM, then the parity condition is satisfied almost surely for any initial state $x_0 \in S$.
\end{theorem}
\begin{proofsketch}
	Intuitively, this theorem holds because the strict lexicographic ordering $\succ$ prohibits infinite descending chains (Proposition~\ref{prop:lexicographic-ranking-supermartingale-soundness}) and truncation preserves the non-strict lexicographic ordering $\succeq$ similarly to Lemma~\ref{lem:truncation-lexicographic-ordering}.
	We give a formal proof in \referappendix{Appendix}{D.1}{sec:proof-progress-measure}.
\end{proofsketch}

\subsection{A Lexicographic Extension of Progress-Measure Supermartingales}

For each priority $i$, the $i$-th component of a PMSM $r \colon S \to [0, \infty)^d$ in Definition~\ref{def:progress-measure} is a function $r_i \colon S \to [0, \infty)$ whose codomain is one-dimensional.
We generalise this by allowing a vector of non-negative real numbers $r_i \colon S \to [0, \infty)^{m_i}$ for each priority $i$ where each vector is ordered lexicographically.
We call the resulting supermartingales \emph{lexicographic progress-measure supermartingales} (LexPMSMs).
This extension has better theoretical properties than PMSMs as we will see later in Section~\ref{sec:properties-lexpmsm}.

\begin{definition}
	Let $m_1, \dots, m_d > 0$ be positive integers.
	We extend the lexicographic ordering $\succ$ and $\succeq$ on $[0, \infty)^d$ to \emph{nested lexicographic ordering} $\succ^{(2)}$ and $\succeq^{(2)}$ on $[0, \infty)^{m_1} \times \dots \times [0, \infty)^{m_d}$ as follows.
	For each $(r_1, \dots, r_d), (r'_1, \dots, r'_d) \in [0, \infty)^{m_1} \times \dots \times [0, \infty)^{m_d}$,
	\begin{align}
		(r_1, \dots, r_d) \succ^{(2)} (r'_1, \dots, r'_d) \quad&\iff\quad \exists l, (\forall i, 1 \le i < l \implies r_i \succeq r'_i) \land r_l \succ r'_l \\
		r \succeq^{(2)} r' \quad&\iff\quad r \succ^{(2)} r' \lor (\forall i, r_i \succeq r'_i)
	\end{align}
\end{definition}
The definition of $\succ^{(2)}$ is a straightforward extension of the lexicographic ordering $\succ$ in Definition~\ref{def:lexicographic-ranking-supermartingale-order}.
If $m_1 = \dots = m_n = 1$, then $\succ^{(2)}$ and $\succeq^{(2)}$ are equivalent to $\succ$ and $\succeq$.
We define truncated nested lexicographic ordering $\succeq^{(2)}_i$ and $\succ^{(2)}_i$ in the same way as $\succeq_i$ and $\succ_i$.
One may think that handling nested lexicographic ordering $\succ^{(2)}$ is a bit complicated, but in fact, it can be flattened into unnested lexicographic ordering $\succ$ as shown below.

\begin{lemma}\label{lem:flatten-lexicographic-ordering}
	Let $\mathrm{flat} \colon [0, \infty)^{m_1} \times \dots \times [0, \infty)^{m_d} \to [0, \infty)^{m_1 + \dots + m_d}$ be a function that flattens a nested vector into a flat vector by concatenation.
	Then, we have the following equivalences.
	\begin{align}
		(r_1, \dots, r_d) \succ^{(2)} (r'_1, \dots, r'_d) \quad&\iff\quad \mathrm{flat}(r_1, \dots, r_d) \succ \mathrm{flat}(r'_1, \dots, r'_d) \\
		(r_1, \dots, r_d) \succeq^{(2)} (r'_1, \dots, r'_d) \quad&\iff\quad \mathrm{flat}(r_1, \dots, r_d) \succeq \mathrm{flat}(r'_1, \dots, r'_d)
		\tag*{\qed}
	\end{align}
\end{lemma}

By replacing $\succ$ and $\succeq$ in Definition~\ref{def:progress-measure} with $\succ^{(2)}$ and $\succeq^{(2)}$, we obtain the following definition.
\begin{definition}\label{def:lexpmsm}
	Let $F \colon S \to \Giry S$ be a parity Markov chain with a priority function $p \colon S \to \OneToN{d}$.
	A \emph{lexicographic progress-measure supermartingale} (LexPMSM) is a non-negative measurable function $r \colon S \to \prod_{i = 1}^d [0, \infty)^{m_i}$ such that the following conditions hold:
	for each $x \in S$, if $p(x)$ is even, then $r(x) \succeq^{(2)}_{p(x)} (\nexttime r)(x)$, and if $p(x)$ is odd, then $r(x) \succ^{(2)}_{p(x)} (\nexttime r)(x)$.
\end{definition}

\begin{theorem}\label{thm:lexicographic-progress-measure}
	Let $F \colon S \to \Giry S$ be a parity Markov chain and a priority function $p \colon S \to \OneToN{d}$.
	If there exists a LexPMSM, then the parity condition is satisfied almost surely for any initial state.
	\qed
\end{theorem}
\begin{proofsketch}
	We can reduce this theorem to Theorem~\ref{thm:progress-measure} by flattening the LexPMSM using Lemma~\ref{lem:flatten-lexicographic-ordering}.
	Details can be found in \referappendix{Appendix}{D.1}{sec:proof-progress-measure}.
\end{proofsketch}

\subsection{Properties of Lexicographic Progress-Measure Supermartingales}\label{sec:properties-lexpmsm}

\subsubsection{Reduced LexPMSMs}\label{sec:reduced-lexpmsm}

It is known that for parity progress measures for parity graphs, assuming that the components corresponding to even priorities are always $0$ does not affect their verification power, and thus we can drop those components without loss of generality.
We show that the same holds for LexPMSMs.

\begin{definition}\label{def:reduced-lexpmsm}
	A \emph{reduced LexPMSM} is a nonnegative measurable function $r \colon S \to \prod_{i = 1}^{\lceil d / 2 \rceil} [0, \infty)^{m_i}$ such that the following conditions hold:
	for each $x \in S$, if $p(x)$ is even, then $r(x) \succeq^{(2)}_{\lceil p(x) / 2 \rceil} (\nexttime r)(x)$, and if $p(x)$ is odd, then $r(x) \succ^{(2)}_{\lceil p(x) / 2 \rceil} (\nexttime r)(x)$.
	Here, $\lceil \cdot \rceil$ is the ceiling function.
\end{definition}
Here, we drop the components corresponding to even priorities using the bijection $\{ i \mid i \in \{ 1, \dots, d \}, \text{$i$ is odd} \} \ni i \mapsto \lceil i / 2 \rceil \in \{ 1, \dots, \lceil d / 2 \rceil \}$, whose inverse is given by $j \mapsto 2 j - 1$.
It is obvious that any reduced LexPMSM induces a LexPMSM by inserting $0$ at the components corresponding to even priorities.
The converse also holds as shown below.

\begin{proposition}\label{prop:lexpmsm-to-reduced-lexpmsm}
	Let $r \colon S \to \prod_{i = 1}^{d} [0, \infty)^{m_i}$ be a LexPMSM for a parity Markov chain $F \colon S \to \Giry S$ with a priority function $p \colon S \to \OneToN{d}$.
	Then, there exists a reduced LexPMSM $r' \colon S \to \prod_{i = 1}^{\lceil d / 2 \rceil} [0, \infty)^{m_i'}$ for the same parity Markov chain.
	Specifically, if we define the $i$-th component $r'_i$ of $r'$ as follows, then $r'$ is a reduced LexPMSM.
	\[ r'_1 = r_1, \qquad r'_2 = \mathrm{flat}(r_2, r_3), \qquad \dots \qquad, r'_{\lceil d / 2 \rceil} = \mathrm{flat}(r_{2 \lceil d / 2 \rceil - 2}, r_{2 \lceil d / 2 \rceil - 1}) \]
\end{proposition}
\begin{proof}[Proof of Proposition~\ref{prop:lexpmsm-to-reduced-lexpmsm}]
	By definition of $r'$, if $p(x)$ is odd, then $r'(x) \succ^{(2)}_{\lceil p(x) / 2 \rceil} (\nexttime r')(x)$ if and only if $r(x) \succ^{(2)}_{p(x)} (\nexttime r)(x)$; and if $p(x)$ is even, then $r'(x) \succeq^{(2)}_{\lceil p(x) / 2 \rceil} (\nexttime r')(x)$ if and only if $r(x) \succeq^{(2)}_{p(x) - 1} (\nexttime r)(x)$.
\end{proof}

\subsubsection{Equivalence to LexGSSMs}

Let $F \colon S \to \Giry S$ be a parity Markov chain with a priority function $p \colon S \to \OneToN{2 d}$.
As we have seen in Lemma~\ref{lem:parity-streett}, the parity condition defined by $p$ is equivalent to the Streett condition defined by the Streett pairs $(S_{\le 2 i - 1}, S_{\le 2 i - 2})$ for each $i = 1, \dots, d$.
We show that LexPMSMs for the parity condition are equivalent to LexGSSMs for the corresponding Streett condition by giving mutual translations between them.

\begin{proposition}\label{prop:lexpmsm-to-lexgssm}
	Let $r = (r_1, \dots, r_{2 d})$ be a LexPMSM where $r_i$ is the $i$-th component for each priority $i$.
	Then, $r {\downharpoonright} (2 i - 1) = (r_1, \dots, r_{2 i - 1})$ is a LexGSSM for the Streett pair $(S_{\le 2 i - 1}, S_{\le 2 i - 2})$ for each $i = 1, \dots, d$.
	\qed
\end{proposition}
\begin{appendixproof}[Proof of Proposition~\ref{prop:lexpmsm-to-lexgssm}]
	If $x \in S_{\le 2 i - 1} \setminus S_{\le 2 i - 2}$, then by definition of LexPMSMs, we have $(\nexttime (r_1, \dots, r_{2 i - 1}))(x) \prec (r_1, \dots, r_{2 i - 1})(x)$.
	If $x \in S \setminus (S_{\le 2 i - 1} \cup S_{\le 2 i - 2})$, then we have $p(x) > 2 i - 1$. In this case, we have $(\nexttime (r_1, \dots, r_{p(x)}))(x) \prec (r_1, \dots, r_{2 i - 1})(x)$ if $p(x)$ is odd, and $(\nexttime (r_1, \dots, r_{p(x)}))(x) \preceq (r_1, \dots, r_{2 i - 1})(x)$ if $p(x)$ is even.
	In either case, we have $(\nexttime (r_1, \dots, r_{2 i - 1}))(x) \preceq (r_1, \dots, r_{2 i - 1})(x)$.
\end{appendixproof}

\begin{proposition}\label{prop:lexgssm-to-lexpmsm}
	Any LexGSSM can be transformed into a (reduced) LexPMSM.
	Specifically, suppose that for each $i = 1, \dots, d$, we have a LexGSSM $r_i \colon S \to [0, \infty)^{m_i}$ for the Streett pair $(S_{\le 2 i - 1}, S_{\le 2 i - 2})$.
	Then, $r \coloneqq (r_1, r_2, \dots, r_d)$ is a reduced LexPMSM for the parity Markov chain.
	\qed
\end{proposition}
\begin{appendixproof}[Proof of Proposition~\ref{prop:lexgssm-to-lexpmsm}]
	By definition, the following conditions hold for each $i$.
	\begin{enumerate}
		\item If $p(x) = 2 i - 1$, then $r_i(x) \succ (\nexttime r_i)(x)$.
		\label{item:lexgssm-to-lexpmsm-1}
		\item If $p(x) \ge 2 i$, then $r_i(x) \succeq (\nexttime r_i)(x)$.
		\label{item:lexgssm-to-lexpmsm-2}
	\end{enumerate}
	We can rephrase this as follows.
	\begin{itemize}
		\item If $p(x) = 2 j - 1$, then for any $i < j$, we have $p(x) \ge 2 i$ and thus $r_i(x) \succeq (\nexttime r_i)(x)$ by Item ~\ref{item:lexgssm-to-lexpmsm-2}.
		If $p(x) = 2 j - 1$ and $i = j$, then we have $p(x) = 2 i - 1$ and thus $r_i(x) \succ (\nexttime r_i)(x)$ by Item~\ref{item:lexgssm-to-lexpmsm-1}.
		Therefore, if $p(x) = 2 j - 1$ is odd, then we have $r(x) \succ^{(2)}_{j} (\nexttime r)(x)$.
		\item If $p(x) = 2 j$, then for any $i \le j$, we have $p(x) \ge 2 i$ and thus $r_i(x) \succeq (\nexttime r_i)(x)$ by Item~\ref{item:lexgssm-to-lexpmsm-2}.
		Therefore, if $p(x) = 2 j$ is even, then we have $r(x) \succeq^{(2)}_{j} (\nexttime r)(x)$.
	\end{itemize}
	Therefore, $r$ is a reduced LexPMSM.
\end{appendixproof}
The proof of these propositions can be found in the
\iflong
appendix.
\else
long version \cite{arxiv}.
\fi
Note that Proposition~\ref{prop:lexgssm-to-lexpmsm} also gives a translation from GSSMs to PMSMs.

\begin{toappendix}
\section{Details of Proofs}

\subsection{Proof of Theorem~\ref{thm:progress-measure}}\label{sec:proof-progress-measure}

For convenience, we slightly generalise parity conditions (Definition~\ref{def:parity-condition}).

\begin{definition}[parity condition]
	A \emph{parity condition} over a set $S$ is specified by a pair $(p, I)$ of a \emph{priority function} $p \colon S \to \OneToN{d}$ and a subset $I \subseteq \OneToN{d}$.
	An infinite sequence $s_0 s_1 \dots \in S^{\omega}$ satisfies the parity condition $(p, I)$ if the minimum priority that occurs infinitely often in the sequence is in $I$.
	\[ \mathbf{Parity}(p, I) \coloneqq \{ s \in S^{\omega} \mid \min \mathbf{Inf}_p(s) \in I \} \]
\end{definition}
Obviously, Definition~\ref{def:parity-condition} is a special case where $I$ is the set of even numbers in $\OneToN{d}$.

We also generalise PMSMs (Definition~\ref{def:progress-measure}) accordingly.
That is, $r$ is a PMSM if the following conditions hold for each $x \in S$.
\begin{itemize}
	\item If $p(x) \in I$, then $r(x) \succeq_{p(x)} (\nexttime r)(x)$.
	\item If $p(x) \notin I$, then $r(x) \succ_{p(x)} (\nexttime r)(x)$.
\end{itemize}

Then, Theorem~\ref{thm:progress-measure} is restated as follows.
\begin{theorem}\label{thm:progress-measure-general-parity-condition}
	Let $F \colon S \to \Giry S$ be a parity Markov chain with parity condition $(p, I)$.
	If there exists a parity progress measure, then the parity condition is satisfied almost surely for any initial state $x_0 \in S$.
	\[ \mathbb{P}_{x_0}[\mathbf{Parity}(p, I)] = 1 \]
\end{theorem}
\begin{proof}
	Let $r$ be a parity progress measure.
	We first construct an auxiliary measurable function $\mathrm{lev} \colon S \to \OneToN{d} \cup \{ 1^{\sharp}, \dots, d^{\sharp} \}$ that satisfies the following conditions for each $x \in S$.
	\begin{itemize}
		\item If $p(x) \in I$, then $\mathrm{lev}(x) \in \OneToN{p(x)} \cup \{ p(x)^{\sharp} \}$. If $p(x) \notin I$, then $\mathrm{lev}(x) \in \OneToN{p(x)}$.
		\item If $\mathrm{lev}(x) = p(x)^{\sharp}$, then $(r_1(x), \dots, r_{p(x)}(x)) \ge ((\nexttime r_1)(x), \dots, (\nexttime r_{p(x)})(x))$, that is, for any $i$ such that $1 \le i \le p(x)$, $r_i(x) \ge (\nexttime r_i)(x)$.
		\item If $\mathrm{lev}(x) \in \OneToN{p(x)}$, then we have $(r_1(x), \dots, r_{p(x)}(x)) \succ ((\nexttime r_1)(x), \dots, (\nexttime r_{p(x)})(x))$, and $\mathrm{lev}(x)$ is a level. That is, for any $i$ such that $1 \le i < \mathrm{lev}(x)$, $r_i(x) \ge (\nexttime r_i)(x)$, and if $i = \mathrm{lev}(x)$, then $r_i(x) \ge 1 + (\nexttime r_i)(x)$.
	\end{itemize}
	The function $\mathrm{lev}$ is defined as follows.
	\begin{itemize}
		\item If $p(x) = i \notin I$, then we have
		\[ p^{-1}(\{ i \}) \subseteq \bigcup_{1 \le l \le i} R_l \quad\text{where}\quad R_l \coloneqq \bigcap_{1 \le j < l} \{ x \mid r_j(x) \ge (\nexttime r_j)(x) \} \cap \{ x \mid r_l(x) \ge 1 + (\nexttime r_l)(x) \}. \]
		Since $p^{-1}(\{ i \})$ and each $R_l$ are measurable sets, there exists a measurable partition $p^{-1}(\{ i \}) = \bigcup_{1 \le l \le i} R'_l$ such that $R'_l \subseteq R_l$.
		We define $\mathrm{lev}(x) = l$ if $x \in R'_l$.
		\item If $p(x) \in I$, then we have
		\[ p^{-1}(\{ i \}) \subseteq \bigcup_{1 \le l \le i} R_l \cap R_{i^{\sharp}} \quad\text{where}\quad R_{i^{\sharp}} \coloneqq \bigcap_{1 \le j < i} \{ x \mid r_j(x) \ge (\nexttime r_j)(x) \}. \]
		Similarly to the previous case, there exists a measurable partition $p^{-1}(\{ i \}) = \bigcup_{1 \le l \le i} R'_l \cap R'_{i^{\sharp}}$ such that $R'_l \subseteq R_l$ and $R'_{i^{\sharp}} \subseteq R_{i^{\sharp}}$.
		We define $\mathrm{lev}(x) = l$ if $x \in R'_l$ and $\mathrm{lev}(x) = i^{\sharp}$ if $x \in R'_{i^{\sharp}}$.
	\end{itemize}

	The rest of the proof combines the ideas of \cite{AgrawalPOPL2018} and \cite{JurdzinskiSTACS2000}.
	We prove by contradiction.
	Note that $\mathbb{P}_{x_0}[\mathbf{Parity}(p, I)] = 1$ is equivalent to $\mathbb{P}_{x_0}[\mathbf{Parity}(p, \compl{I})] = 0$ where $\compl{I} \coloneqq \OneToN{d} \setminus I$ is the complement of $I$.
	Assume there exists an initial state $x_0$ such that
	\[ \mathbb{P}_{x_0}[\mathbf{Parity}(p, \compl{I})] > 0. \]
	We repeatedly apply the following principle: if $\bigcup_{i = 0}^{\infty} E_i$ is a countable union of measurable sets such that $\mathbb{P}_{x_0}[\bigcup_{i = 0}^{\infty} E_i] > 0$, then there exists $i$ such that $\mathbb{P}_{x_0}[E_i] > 0$.
	\begin{itemize}
		\item By definition of $\mathbf{Parity}(p, \compl{I})$, we have
		\[ \mathbf{Parity}(p, \compl{I}) = \bigcup_{k \in \compl{I}} \mathbf{Parity}(p, \{ k \}). \]
		There exists $k \in \compl{I}$ such that $\min \mathbf{Inf}_p$ is $k$ with positive probability.
		\[ \mathbb{P}_{x_0}[\mathbf{Parity}(p, \{ k \})] > 0 \]
		\item If $\min \mathbf{Inf}_p = k$, then after sufficiently many steps, the priority function must be at least $k$.
		Specifically, we have
		\[ \mathbf{Parity}(p, \{ k \}) = \bigcup_{n = 1}^{\infty} \{ x \in S^{\omega} \mid \min \mathbf{Inf}_p(x) = k \land \forall m \ge n, p(x_m) \ge k \}. \]
		There exists $n$ such that after $n$ steps, the priority function is at least $k$ with positive probability.
		\[ \mathbb{P}_{x_0} [\{ x \in S^{\omega} \mid \min \mathbf{Inf}_p(x) = k \land \forall m \ge n, p(x_m) \ge k \}] > 0 \]
		\item Since $p(x_i) = k$ holds for infinitely many $i$'s, there exists $l$ such that $1 \le l \le k$ and $\mathrm{lev}(x_i) = l$ for infinitely many $i$'s.
		Consider the minimum of such $l$'s, which we write as $\min \mathbf{Inf}_{\mathrm{lev}}(x)$.
		More precisely, we take the minimum with respect to the order relation on $\OneToN{d} \cup \{ 1^{\sharp}, \dots, d^{\sharp} \}$ defined by $1 < 1^{\sharp} < 2 < \dots < d < d^{\sharp}$.
		Since we have $p(x_i) \ge k$ after sufficiently many steps, $\mathrm{lev}(x_i) \in \{ 1^{\sharp}, \dots, (k - 1)^{\sharp} \}$ cannot happen infinitely many times.
		\begin{align}
			&\{ x \in S^{\omega} \mid \min \mathbf{Inf}_p(x) = k \land \forall m \ge n, p(x_m) \ge k \} \\
			&= \bigcup_{1 \le l \le k} \{ x \in S^{\omega} \mid \min \mathbf{Inf}_p(x) = k \land \forall m \ge n, p(x_m) \ge k \land \min \mathbf{Inf}_{\mathrm{lev}}(x) = l \}
		\end{align}
		There exists $l$ such that $\min \mathbf{Inf}_{\mathrm{lev}} = l$ with positive probability.
		\[ \mathbb{P}_{x_0} [\{ x \in S^{\omega} \mid \min \mathbf{Inf}_p(x) = k \land \forall m \ge n, p(x_m) \ge k \land \min \mathbf{Inf}_{\mathrm{lev}}(x) = l \}] > 0. \]
		\item Similarly to $\min \mathbf{Inf}_{p}$, after sufficiently many steps, the level $\mathrm{lev}(x_i)$ must be at least $l$.
		\begin{align}
			&\{ x \in S^{\omega} \mid \min \mathbf{Inf}_p(x) = k \land \forall m \ge n, p(x_m) \ge k \land \min \mathbf{Inf}_{\mathrm{lev}}(x) = l \} \\
			&= \bigcup_{n' \ge n} \{ x \in S^{\omega} \mid \min \mathbf{Inf}_p(x) = k \land \forall m \ge n, p(x_m) \ge k \land \min \mathbf{Inf}_{\mathrm{lev}}(x) = l \land \forall m \ge n', \mathrm{lev}(x_m) \ge l \}
		\end{align}
		There exists $n' \ge n$ such that after $n'$ steps, the level $\mathrm{lev}(x_i)$ is at least $l$ with positive probability.
		\[ \mathbb{P}_{x_0}[\{ x \in S^{\omega} \mid \min \mathbf{Inf}_p(x) = k \land \forall m \ge n, p(x_m) \ge k \land \min \mathbf{Inf}_{\mathrm{lev}}(x) = l \land \forall m \ge n', \mathrm{lev}(x_m) \ge l \}] > 0 \]
		\item By omitting the condition $\min \mathbf{Inf}_p(x) = k$ and unifying $n$ and $n'$, we have the following.
		\begin{align}
			&\{ x \in S^{\omega} \mid \min \mathbf{Inf}_p(x) = k \land \forall m \ge n, p(x_m) \ge k \land \min \mathbf{Inf}_{\mathrm{lev}}(x) = l \land \forall m \ge n', \mathrm{lev}(x_m) \ge l \} \\
			&\subseteq \{ x \in S^{\omega} \mid \min \mathbf{Inf}_{\mathrm{lev}}(x) = l \land \forall m \ge n', p(x_m) \ge k \land \mathrm{lev}(x_m) \ge l \}
		\end{align}
		We simplify the current event as follows.
		\[ \mathbb{P}_{x_0} [\{ x \in S^{\omega} \mid \min \mathbf{Inf}_{\mathrm{lev}}(x) = l \land \forall m \ge n', p(x_m) \ge k \land \mathrm{lev}(x_m) \ge l \}] > 0 \]
		Intuitively, this condition implies that the $l$-th component $r_l$ must be non-increasing after $n'$ steps and strictly decreasing for infinitely many steps.
		\item By definition of a parity progress measure, we have $r_l(x_{n'}) \in [0, \infty) = \bigcup_{B \in \mathbb{N}} [0, B]$.
		\begin{align}
			&\{ x \in S^{\omega} \mid \min \mathbf{Inf}_{\mathrm{lev}}(x) = l \land \forall m \ge n', p(x_m) \ge k \land \mathrm{lev}(x_m) \ge l \} \\
			&= \bigcup_{B \in \mathbb{N}} \{ x \in S^{\omega} \mid \min \mathbf{Inf}_{\mathrm{lev}}(x) = l \land \forall m \ge n', p(x_m) \ge k \land \mathrm{lev}(x_m) \ge l \land r_l(x_{n'}) \le B \}
		\end{align}
		There exists $B \in \mathbb{N}$ such that $r_l(x_{n'})$ is at most $B$ with positive probability.
		\[ \mathbb{P}_{x_0}[\{ x \in S^{\omega} \mid \min \mathbf{Inf}_{\mathrm{lev}}(x) = l \land \forall m \ge n', p(x_m) \ge k \land \mathrm{lev}(x_m) \ge l \land r_l(x_{n'}) \le B \}] > 0 \]
	\end{itemize}
	Now, we write $M$ for the measurable set above.
	\[ M \coloneqq \{ x \in S^{\omega} \mid \min \mathbf{Inf}_{\mathrm{lev}}(x) = l \land \forall m \ge n', p(x_m) \ge k \land \mathrm{lev}(x_m) \ge l \land r_l(x_{n'}) \le B \} \]
	Similarly to \cite{AgrawalPOPL2018}, we construct a stochastic process, which leads to a contradiction.
	We define a filtration for $S^{\omega}$ by $\{ \mathcal{F}_m \}_{m \in \mathbb{N}} \coloneqq \{ \{ \pi_m^{-1}(C) \mid C \in \Sigma(S^{m}) \} \}_{m \in \mathbb{N}}$ where $\pi_m \colon S^{\omega} \to S^{m}$ is the projection of the first $m$ elements and $\Sigma(S^{m})$ is the product $\sigma$-algebra.
	Consider the following stochastic process $Y_m$.
	\[ Y_m(x) \quad\coloneqq\quad \begin{cases}
		0 & r_l(x_{n'}) > B \\
		B & r_l(x_{n'}) \le B \land m < n' \\
		r_l(x_m) & r_l(x_{n'}) \le B \land m \ge n' \land F(x) > m \\
		r_l(x_{F(x)}) & r_l(x_{n'}) \le B \land m \ge n' \land F(x) \le m \\
	\end{cases} \]
	where $F$ is the stopping time with respect to the filtration $\{ \mathcal{F}_m \}_{m = 0}^{\infty}$ and defined as follows.
	\[ F(x) \quad\coloneqq\quad \min \{ m \mid m \ge n' \land \mathrm{lev}(x_m) < l \} \]
	Let $D \coloneqq \{ x \in S^{\omega} \mid r_l(x_{n'}) \le B \}$.
	By definition, we have
	\[ Y_{n'}(x) = \begin{cases}
		0 & r_l(x_{n'}) > B \\
		r_{l}(x_{n'}) & \text{otherwise}
	\end{cases}, \]
	and thus, we have the following.
	\begin{equation}
		\mathbb{E}_{x_0}[Y_{n'}] \le \mathbb{E}_{x_0}[B \cdot \chi_{D}] = B \cdot \mathbb{P}_{x_0}[D]
		\label{eq:progress-measure-proof1}
	\end{equation}
	For each $m \ge n'$, we have
	\[ Y_{m + 1}(x) - Y_m(x) = \begin{cases}
		r_{l}(x_{m + 1}) - r_{l}(x_m) & \text{if } r_{l}(x_{n'}) \le B \land F(x) \ge m + 1 \\
		0 & \text{otherwise.}
	\end{cases} \]
	Let $D_{m} \coloneqq \{ x \mid r_{l}(x_{n'}) \le B \land F(x) \ge m + 1 \}$.
	Since $F$ is a stopping time, we have $D_{m} \in \mathcal{F}_m$.
	By identifying $D_m \subseteq S^{\omega}$ with the corresponding measurable set in $\Sigma(S^m)$, we have the following.
	\begin{align}
		&\mathbb{E}_{x_0}[Y_{m + 1}] - \mathbb{E}_{x_0}[Y_{m}] \\
		&= \iint \chi_{D_m}(x_1 \dots x_m) \cdot (r_{l}(x_{m + 1}) - r_{l}(x_m)) \,\mathrm{d} f_{x_m}(x_{m + 1}) \,\mathrm{d} (f_{m})_{x_0}(x_1 \dots x_m) \\
		&= \int \chi_{D_m}(x_1 \dots x_m) \cdot \left( \int r_{l}(x_{m + 1}) \,\mathrm{d} f_{x_m}(x_{m + 1}) - r_{l}(x_m) \right) \,\mathrm{d} (f_{m})_{x_0}(x_1 \dots x_m)
	\end{align}
	For each $x \in D_m$, if $\mathrm{lev}(x_m) = l$, then
	\[ \int r_{l}(x_{m + 1}) \,\mathrm{d} f_{x_m}(x_{m + 1}) - r_{l}(x_m) \le -1 \]
	and otherwise, we have
	\[ \int r_{l}(x_{m + 1}) \,\mathrm{d} f_{x_m}(x_{m + 1}) - r_{l}(x_m) \le 0 \]
	because $F(x) \ge m + 1$ implies $\mathrm{lev}(x_m) \ge l$.
	Thus
	\[ \mathbb{E}_{x_0}[Y_{m + 1}] - \mathbb{E}_{x_0}[Y_{m}] \le - \mathbb{P}_{x_0}[D_m \cap \{ x \mid \mathrm{lev}(x_m) = l \}]. \]
	Since we have $M \subseteq D_m$,
	\begin{equation}
		\mathbb{E}_{x_0}[Y_{m + 1}] - \mathbb{E}_{x_0}[Y_{m}] \le - \mathbb{P}_{x_0}[M \cap \{ x \mid \mathrm{lev}(x_m) = l \}].
		\label{eq:progress-measure-proof2}
	\end{equation}
	By~\eqref{eq:progress-measure-proof1} and~\eqref{eq:progress-measure-proof2}, for each $m \ge n'$, we have
	\begin{align}
		\mathbb{E}_{x_0}[Y_{m}] &= \mathbb{E}_{x_0}[Y_{n'}] + \sum_{i = n'}^{m - 1} (\mathbb{E}_{x_0}[Y_{i + 1}] - \mathbb{E}_{x_0}[Y_{i}]) \\
		&\le B \cdot \mathbb{P}_{x_0}[D] - \sum_{i = n'}^{m - 1} \mathbb{P}_{x_0}[M \cap \{ x \mid \mathrm{lev}(x_i) = l \}] \\
		&= B \cdot \mathbb{P}_{x_0}[D] - \mathbb{E}_{x_0}[\chi_M \cdot \sum_{i = n'}^{m - 1} \chi_{\{ x \mid \mathrm{lev}(x_i) = l \}}]
	\end{align}
	Let $\sharp_m(x) \coloneqq \sum_{i = n'}^{m - 1} \chi_{\{ x \mid \mathrm{lev}(x_i) = l \}}(x) = |\{ i \mid n' \le i < m \land \mathrm{lev}(x_i) = l \}|$.
	For any $j \in \mathbb{N}$, we have
	\[ M = \bigcup_{m \ge n'} (M \cap \{ x \mid \sharp_m(x) \ge j \}) \]
	Since the right-hand side is an increasing sequence of measurable sets, we have
	\[ \mathbb{P}_{x_0}[M] = \sup_{m \ge n'} \mathbb{P}_{x_0}[M \cap \{ x \mid \sharp_m(x) \ge j \}] \]
	There exists $m' \ge n'$ such that
	\[ \mathbb{P}_{x_0}[M \cap \{ x \mid \sharp_{m'}(x) \ge j \}] \ge \mathbb{P}_{x_0}[M] / 2 \]
	Now, we have
	\begin{align}
		\mathbb{E}_{x_0}[Y_{m'}] &\le B \cdot \mathbb{P}_{x_0}[D] - \mathbb{E}_{x_0}[\chi_M \cdot \sum_{i = n'}^{m - 1} \chi_{\{ x \mid \mathrm{lev}(x_i) = l \}}] \\
		&\le B \cdot \mathbb{P}_{x_0}[D] - \mathbb{E}_{x_0}[\chi_{M \cap \{ x \mid \sharp_{m'}(x) \ge j \}} \cdot j] \\
		&\le B \cdot \mathbb{P}_{x_0}[D] - j \cdot \mathbb{P}_{x_0}[M] / 2
	\end{align}
	We retrospectively choose $j$ sufficiently large so that $B \cdot \mathbb{P}_{x_0}[D] - j \cdot \mathbb{P}_{x_0}[M] / 2 < 0$ holds.
	Then, this contradicts the non-negativity of $Y_{m'}$.
\end{proof}

\begin{proof}[Proof of Theorem~\ref{thm:lexicographic-progress-measure}]
	Let $r \colon S \to \prod_{i = 1}^d [0, \infty)^{m_i}$ be a lexicographic parity progress measure.
	The key idea is that if we flatten $r$ into a single vector $r' \colon S \to [0, \infty)^{d'}$, then $r'$ is a parity progress measure for some parity condition $(p', I')$, which is equivalent to $(p, I)$.
	Specifically, we define a parity condition $(p' \colon S \to \OneToN{d'}, I')$ as follows.
	\[ d' \coloneqq \sum_{i=1}^d m_i \qquad p'(x) \coloneqq \sum_{i = 1}^{p(x)} m_i \qquad I' \coloneqq \{ \sum_{i = 1}^{j - 1} m_i + k \mid j \in I, 1 \le k \le m_j \} \]
	Then, $(p, I)$ and $(p', I')$ are equivalent.
	We define $r'(x) \coloneqq \mathrm{flat}(r(x))$.
	By definition, we have $\mathrm{flat}((r_1(x), \dots, r_{p(x)}(x))) = (r'_1(x), \dots, r'_{p'(x)}(x))$.
	By Lemma~\ref{lem:flatten-lexicographic-ordering}, $r'$ is a parity progress measure for the parity condition $(p', I')$.
\end{proof}

\end{toappendix}

\section{Algorithm for Synthesising LexPMSMs}
\label{sec:algorithm}

We provide a synthesis algorithm for LexPMSMs by extending the algorithm for lexicographic ranking supermartingales \cite{AgrawalPOPL2018}.
The description in this section focuses on presenting the idea and is separated from implementation details as much as possible.
Details of implementation are deferred to Section~\ref{sec:implementation}.
Although we focus on LexPMSMs in this section, a synthesis algorithm for LexGSSMs can be also obtained similarly (see \referappendix{Appendix}{E.2}{sec:lexgssm-synthesis}).

\subsection{Input: a pCFG with a Priority Partition}
As a representation of possibly infinite-state Markov chains, we consider (an abstract version of) probabilistic control flow graphs (pCFGs) without nondeterminism.

\begin{definition}
	An \emph{abstract probabilistic control flow graph} (abstract pCFG or simply pCFG) is a family of Markov kernels of the following form where $n$ is a number of real-valued program variables and $L$ is a finite set of locations equipped with discrete $\sigma$-algebra.
	\[ \{ f_l \colon \mathbb{R}^n \to \Giry (L \times \mathbb{R}^n) \}_{l \in L} \]
\end{definition}

A pCFG induces a Markov chain $L \times \mathbb{R}^n \to \Giry (L \times \mathbb{R}^n)$ and the next-time operator $\nexttime \colon \mathbf{Meas}(L \times \mathbb{R}^n, [0, \infty]) \to \mathbf{Meas}(L \times \mathbb{R}^n, [0, \infty])$ as in Section~\ref{sec:markov-chains}.
We often represent a measurable function $L \times \mathbb{R}^n \to [0, \infty]$ as a family $\{ g_l \colon \mathbb{R}^n \to [0, \infty] \}_{l \in L}$ of measurable functions.

A priority function for a pCFG is given as a measurable function $p \colon L \times \mathbb{R}^n \to \{ 1, \dots, d \}$.
In the synthesis algorithm that we present later, it is more convenient to represent a priority function as a family of measurable sets
$P = \{ P_{l, i} \subseteq \mathbb{R}^n \}_{l \in L, i \in \{1, \ldots, d\}}$
such that for each $l \in L$, $\{ P_{l, i} \}_{i = 1, \dots, d}$ is a partition of $\mathbb{R}^n$.
We call such a family $P$ a \emph{priority partition}.
These two representations are equivalent, as we can recover a priority function $p \colon L \times \mathbb{R}^n \to \{1, \ldots, d\}$ by $p(l, x) \coloneqq i$ for each $x \in P_{l, i}$.

\begin{remark}
	In the literature on supermartingale synthesis, it is common to take as input an invariant map $I = \{ I_l \subseteq \mathbb{R}^n \}_{l \in L}$, which is a measurable subset of states that are reachable from the initial state.
	In our setting, invariant maps can be included in the priority partition by considering the intersection $P'_{l, i} \coloneqq P_{l, i} \cap I_l$  of a priority partition $P$ and an invariant map $I$.
	Here, we relax the definition of priority partitions by allowing the union $\bigcup_{i = 1}^d P'_{l, i}$ to be a proper subset of $\mathbb{R}^n$.
	This extended definition of priority partitions can be justified by a reduction to the original definition obtained by assigning states outside the invariant map an even priority $d' \ge d$.
\end{remark}

\subsection{Output: a LexPMSM Map}
Given a pCFG with a priority partition, the aim of our algorithm is to synthesise a \emph{LexPMSM map} for the pCFG.
LexPMSM maps are defined based on reduced LexPMSMs in Definition~\ref{def:reduced-lexpmsm} but impose some condition on levels of lexicographic ordering, which we explain below.

We first define a \emph{level} of the nested lexicographic ordering $\succ^{(2)}$.
Let $r, r' \in \prod_{j = 1}^{\lceil d / 2 \rceil} [0, \infty)^{m_j}$.
We write $r_j, r'_j \in [0, \infty)^{m_j}$ for the $j$-th component of $r, r'$ and $r_{j, k}, r'_{j, k} \in [0, \infty)$ for the $k$-th element of $r_j, r'_j$.
Let $\mathbf{Lev} \coloneqq \{ (j, k) \mid j = 1, \dots, \lceil d / 2 \rceil; k = 1, \dots, m_j \}$ be the set of indices for the components of $r$ and $r'$.
We assume that elements of $\mathbf{Lev}$ are ordered lexicographically.
By Lemma~\ref{lem:flatten-lexicographic-ordering}, if $r \succ^{(2)} r'$, then there exists $(j, k)$ such that $r_{j, k} \ge 1 + r'_{j, k}$ and for any $(j', k') <_{\mathrm{lex}} (j, k)$, $r_{j', k'} \ge r'_{j', k'}$.
In this situation, we write $r \succ^{(2)}_{[(j, k)]} r'$ and say that $(j, k)$ is a \emph{level} of $r \succ^{(2)} r'$.

\begin{definition}\label{def:lexpmsm-map}
	A \emph{LexPMSM map} for a pCFG with a priority partition is a tuple
	\begin{equation}
		(\{ m_j \}_{j = 1, \dots, \lceil d / 2 \rceil},\quad \mathrm{lev},\quad \{ r_{(l, i), j, k} \}_{l \in L,\ i \in \{1, \ldots, d \},\ j \in \{1, \dots, \lceil d / 2 \rceil \},\ k \in \{ 1, \dots, m_j \}})
		\label{eq:lexpmsm-map}
	\end{equation}
	where $m_j$ is a positive integer, $\mathrm{lev} \colon L \times \{ 1, \dots, d \} \to \mathbf{Lev} \cup \{ \star \}$ is a (partial) function that assigns a level to each of the priority partition $\{ P_{l, i} \}_{l \in L, i \in \{1, \dots, d\}}$, and $r_{(l, i), j, k} \colon \mathbb{R}^n \to \mathbb{R}$ is a measurable function.
	Here, $r_{(l, i), j, k}$ defines the $(j, k)$-th component $r_{j, k} \colon L \times \mathbb{R}^n \to \mathbb{R}$ of a function $r \colon L \times \mathbb{R}^n \to \prod_{j = 1}^{\lceil d / 2 \rceil} \mathbb{R}^{m_j}$ by
	$r_{j, k}(l, x) \coloneqq \sum_{i=1}^d [x \in P_{l, i}] \cdot r_{(l, i), j, k}(x)$.
	That is, $r_{(l, i), j, k}$ specifies the restriction of $r_{j, k}$ to $P_{l, i}$.
	The function $\mathrm{lev}$ specifies a level of the nested lexicographic ordering.
	If $\mathrm{lev}(l, i) = \star$, it means that the corresponding nested lexicographic ordering is non-strict.
	The following conditions are required for LexPMSM maps.
	\begin{itemize}
		\item The function $r_{(l, i), j, k}$ is non-negative on $P_{l, i}$ for each $l, i, j, k$.
		\item If $\mathrm{lev}(l, i) = (j, k) \in \mathbf{Lev}$, then $j \le \lceil i / 2 \rceil$. If $\mathrm{lev}(l, i) = \star$, then $i$ is even.
		\item If $\mathrm{lev}(l, i) = (j, k)$, then for any $x \in P_{l, i}$, $r(l, x) \succ^{(2)}_{[(j, k)]} (\nexttime r)(l, x)$.
		\item If $\mathrm{lev}(l, i) = \star$, then for any $j \le \lceil i / 2 \rceil$, $k = 1, \dots, m_j$, and $x \in P_{l, i}$, $r_{j, k}(l, x) \ge (\nexttime r_{j, k})(l, x)$.
	\end{itemize}
	We often omit $m_j$ and $\mathrm{lev}$ when clear from the context.
	\qed
\end{definition}

By the conditions in Definition~\ref{def:lexpmsm-map}, a LexPMSM map indeed defines a LexPMSM:
if $\mathrm{lev}(l, i) = \star$, then we have a non-strict inequality $r(l, x) \succeq^{(2)}_{\lceil p(l, x) / 2 \rceil} (\nexttime r)(l, x)$ for any $x \in P_{l, i}$; and if $\mathrm{lev}(l, i) = (j, k)$, then we have a strict inequality $r(l, x) \succ^{(2)}_{\lceil p(l, x) / 2 \rceil} (\nexttime r)(l, x)$ for any $x \in P_{l, i}$ because $r(l, x) \succ^{(2)}_{[(j, k)]} (\nexttime r)(l, x)$ implies $r(l, x) \succ^{(2)}_{j} (\nexttime r)(l, x)$, which further implies $r(l, x) \succ^{(2)}_{\lceil p(l, x) / 2 \rceil} (\nexttime r)(l, x)$.
Therefore, we have the following proposition.

\begin{proposition}
	A LexPMSM map defines a reduced LexPMSM for a pCFG with a priority partition and thus ensures almost sure satisfaction of the corresponding parity condition.
	\qed
\end{proposition}

One notable difference between LexPMSM maps and LexPMSMs is that in a LexPMSM map, all elements in $P_{l, i}$ share the same level $\mathrm{lev}(l, i)$, while in a LexPMSM, each point $x \in P_{l, i}$ may have a different level.
Such a restriction is common in the literature \cite[Definition~4.6]{AgrawalPOPL2018} and allows us to synthesise a LexPMSM map efficiently.
On the other hand, there is no guarantee that this restriction can be imposed without loss of generality to the best of our knowledge.

\subsection{Synthesising LexPMSM Maps}

\newcommand{\Synthesise}{\hyperref[alg:line:synthesise]{\textsc{Synthesise}}}
\begin{figure}[tbp]
	\begin{algorithmic}[1]
		\Require A pCFG $F = \{ f_l \colon \mathbb{R}^n \to \Giry (L \times \mathbb{R}^n) \}_{l \in L}$ with a priority partition $P = \{ P_{l, i} \}_{l \in L, i \in \{1, \ldots, d\}}$
		\Ensure A LexPMSM map or ``no LexPMSM map found''
		\Function{Synthesise}{$F, P$}\label{alg:line:synthesise}
		\State $T \gets \{ (l, i) \mid l \in L, i \in \{1, \ldots, d\} \}$\label{alg:line:t1}
		\For{$j \in \{ 1, 2, \dots, \lceil d / 2 \rceil \}$}
		\State $T \gets \{ (l, i) \mid (l, i) \in T \land i \ge 2 j - 1 \}$ \label{alg:line:t2}
		\State $k \gets 0$
		\Repeat
		\State $c_0 \gets \{ P_{l, i} \Rightarrow R_{l, i} \ge \nexttime (\{R_{l, i}\}_{l \in L, i \in \{1, \ldots, d\}}) \mid (l, i) \in T \} \cup \{ P_{l, i} \Rightarrow R_{l, i} \ge 0 \mid l \in L, i \in \{1, \ldots, d\} \}$ \label{alg:line:c0}
		\State $c_1 \gets \{ P_{l, i} \Rightarrow R_{l, i} \ge 1 + \nexttime (\{R_{l, i}\}_{l \in L, i \in \{1, \ldots, d\}}) \mid (l, i) \in T \}$ \label{alg:line:c1}
		\State $\{ r_{(l, i), j, k} \}_{l \in L, i \in \{ 1, \dots, d \}} \gets \Call{Solve}{c_0, c_1}$
		\State $T' \gets \{ (l, i) \in T \mid\quad \models P_{l, i} \Rightarrow r_{(l, i), j, k} \ge 1 + \nexttime (\{r_{(l, i), j, k}\}_{l \in L, i \in \{1, \ldots, d\}}) \}$
		\State $T \gets T \setminus T'$ \label{alg:line:t3}
		\State $k \gets k + 1$
		\Until{$T' = \emptyset$}
		\If{$\{ (l, i) \in T \mid i = 2 j - 1 \} \neq \emptyset$}
			\State \Return{no LexPMSM map found} \label{alg:line:not-found}
		\EndIf
		\If{$k = 1$}
			\State $m_j \gets 1$
			\State $\{ r_{(l, i), j, 0} \}_{l \in L, i \in \{ 1, \dots, d \}} \gets \{ 0 \}$
		\Else
			\State $m_j \gets k - 1$
		\EndIf
		\EndFor
		\State \Return{$\{ r_{(l, i), j, k} \mid l \in L, i \in \{ 1, \dots, d \}, j \in \OneToN{\lceil d / 2 \rceil}, k \in \OneToN{m_j} \}$}
		\EndFunction
	\end{algorithmic}
	\caption{Algorithm for synthesizing a LexPMSM map.}
	\label{alg:progress-measure-synthesis}
\end{figure}

The algorithm \Synthesise\ in Fig.~\ref{alg:progress-measure-synthesis} extends the idea of the algorithm for synthesising lexicographic ranking supermartingales \cite{AgrawalPOPL2018}.
\Synthesise\ iteratively synthesises $\{ r_{(l, i), j, k} \}_{l \in L, i \in \{ 1, \dots, d \}}$ starting from level $(j, k) = (1, 1)$, i.e., from the most significant level.
For each iteration, \Synthesise\ constructs two sets of constraints $c_0$ and $c_1$ over a set of function variables $\{ R_{l, i} \colon \mathbb{R}^n \to \mathbb{R} \}_{l \in L, i = 1, \dots, d}$ (Line~\ref{alg:line:c0},~\ref{alg:line:c1}).
The first set $c_0$ corresponds to the non-strict inequality condition and the non-negativity of LexPMSM maps, while the second set $c_1$ corresponds to the strict inequality condition.
To be precise, the next-time operator $\nexttime$ in Line~\ref{alg:line:c0},~\ref{alg:line:c1} should be applied to $\lambda (l, x). \sum_{i = 1}^d [x \in P_{l, i}] \cdot R_{l, i}(x)$ here, but for simplicity of notation, we write simply as $\nexttime (\{R_{l, i}\}_{l \in L, i \in \{1, \ldots, d\}})$.
The set $T \subseteq L \times \{ 1, \dots, d \}$ keeps track of active constraints throughout the algorithm.
\textsc{Solve} is a procedure that finds measurable functions $\{ R_{l, i} \colon \mathbb{R}^n \to \mathbb{R} \}_{l \in L, i = 1, \dots, d} = \{ r_{(l, i), j, k} \}_{l \in L, i \in \{ 1, \dots, d \}}$ such that (1) all constraints in $c_0$ are satisfied and (2) at least one constraint in $c_1$ is satisfied if possible.
In other words, $c_0$ and $c_1$ are treated as hard and soft constraints, respectively.
It is preferable that \textsc{Solve} returns a solution that satisfies as many constraints in $c_1$ as possible (see Remark~\ref{rem:solve-implementation}), but it is not necessary for soundness and relative completeness.
Note that $R_{l,i} = 0$ is always a solution that satisfies all constraints in $c_0$; thus, \textsc{Solve} always returns a solution.
If the strict inequality for $(l, i) \in T$ is satisfied by the returned solution, then $(l, i)$ is removed from $T$ (Line~\ref{alg:line:t3}), since we do not need the constraint for $(l, i)$ in the subsequent iterations for synthesising less significant components.
The inner loop ends when strict inequality constraints cannot be satisfied for any remaining $(l, i) \in T$.
At this point, if there remains any $(l, 2 j - 1) \in T$, then it means that \Synthesise\ failed to find a LexPMSM map that satisfies the condition for odd priority $2 j - 1$; thus, \Synthesise\ returns ``no LexPMSM map found'' (Line~\ref{alg:line:not-found}).
Otherwise, \Synthesise\ proceeds by incrementing $j$.
When constraints for all odd priorities are resolved, \Synthesise\ returns the synthesised LexPMSM map.

\begin{example}
	Consider applying \Synthesise\ in Fig.~\ref{alg:progress-measure-synthesis} to the Streett Markov chain in Example~\ref{ex:no-streett-supermartingale}.
	Here, we slightly modify the program in Example~\ref{ex:no-streett-supermartingale} by changing the type of variables from natural numbers to real numbers.
	Since $\{ (m, n) \in \mathbb{R}^2 \mid m > -1, n \ge 0 \}$ is an invariant of the modified program, we consider the priority partition given as $P_{l_0, 2} = P_{l_1, 3} = \{ (m, n) \in \mathbb{R}^2 \mid m > -1, n \ge 0 \}$, and $P_{l, i} = \emptyset$ for other $(l, i)$.
	\Synthesise\ proceeds as follows.
	For simplicity, we assume that $T$ is initialised to $\{ (l_0, 2), (l_1, 3) \}$ because constraints for other $(l, i)$ are trivially satisfied.
	In the first iteration for $(j, k) = (1, 0)$, the set $c_0$ is given as follows. We omit non-negativity constraints for simplicity.
	\begin{align}
		\forall m > -1, \forall n \ge 0,\quad R_{l_0, 2}(m, n) \quad&\ge\quad R_{l_1, 3}(n, n) \label{eq:algorithm-example-1} \\
		\forall m > -1, \forall n \ge 0,\quad R_{l_1, 3}(m, n) \quad&\ge\quad \ifexpr{m > 0}{R_{l_1, 3}(m - 1, n)}{R_{l_0, 2}(m, n + 1)} \label{eq:algorithm-example-2}
	\end{align}
	In this case, \texttt{Solve} cannot find a solution such that $T' \neq  \emptyset$ because both $l_0$ and $l_1$ are visited infinitely often.
	Thus, we have $m_1 = 1$, and the first component $r_1 \colon L \times \mathbb{R} \to [0, \infty)^{m_1}$ of the LexPMSM is given as $r_1(l, (m, n)) = 0$.
	Then, $j$ is incremented to $2$ and $T$ is updated to $\{ (l_1, 3) \}$.
	In the next iteration for $(j, k) = (2, 0)$, $c_0$ consists of only the constraint~\eqref{eq:algorithm-example-2}.
	\textsc{Solve} finds $r_{l_0, 2}(m, n) = 0$ and $r_{l_1, 3}(m, n) = m + 1$, which satisfies the strict inequality for $(l_1, 3)$.
	Thus, $T$ is updated to $\emptyset$.
	In the last iteration for $(j, k) = (2, 1)$, we have $T' = \emptyset$ and \Synthesise\ terminates.
	The second component $r_2 \colon L \times \mathbb{R} \to [0, \infty)^{m_2}$ of the LexPMSM is given as $m_2 = 1$, $r_2(l_0, (m, n)) = 0$, and $r_2(l_1, (m, n)) = m + 1$.
	\qed
\end{example}

\begin{remark}\label{rem:solve-implementation}
	If a constraint solver supports optimisation, as discussed in \cite{AgrawalPOPL2018}, $\textsc{Solve}(c_0, c_1)$ can be implemented as follows.
	For each $(l, i) \in T$, we introduce an auxiliary variable $\epsilon_{l, i}$ and replace $c_0$ and $c_1$ with the following constraints.
	\[ \forall (l, i) \in T,\qquad P_{l, i} \Rightarrow R_{l, i} \ge \epsilon_{l, i} + \nexttime (\{R_{l, i}\}_{l \in L, i \in \{1, \ldots, d\}}) \quad \land \quad 0 \le \epsilon_{l, i} \le 1 \]
	We then instruct the solver to maximise the sum $\sum_{(l, i) \in T} \epsilon_{l, i}$ of auxiliary variables.
	The solution returned by this procedure necessarily satisfies $\epsilon_{l, i} \in \{ 0, 1 \}$: $\epsilon_{l, i} = 1$ indicates that the corresponding constraint in $c_1$ is satisfied.
	The same idea can be applied even when the solver does not support optimization: in that case, we solve the constraints together with an additional constraint $\sum_{(l, i) \in T} \epsilon_{l, i} > 0$ (or $\ge 1$).
	A solution obtained in this way may not satisfy $\epsilon_{l, i} \in \{ 0, 1 \}$, but if $\epsilon_{l, i} > 0$, then the corresponding constraint in $c_1$ can be satisfied, since the inequality $R_{l, i} \ge \epsilon_{l, i} + \nexttime (\{R_{l, i}\}_{l \in L, i \in \{1, \ldots, d\}})$ can be rescaled by multiplying both sides by $1 / \epsilon_{l, i}$ using the linearity of $\nexttime$.
	If $\sum_{(l, i) \in T} \epsilon_{l, i} > 0$ is not satisfiable, then $\textsc{Solve}(c_0, c_1)$ returns $0$ as a trivial solution for $c_0$.
\end{remark}

Now, we state the soundness and relative completeness of the algorithm.
The proofs are deferred to the
\iflong
appendix.
\else
long version \cite{arxiv}.
\fi
\begin{theorem}[soundness]\label{thm:algorithm-soundness}
	If \Synthesise\ returns $\{ r_{(l, i), j, k} \}$, then it is a LexPMSM map.
	\qed
\end{theorem}
\begin{appendixproof}[Proof of Theorem~\ref{thm:algorithm-soundness}]
	The following invariant $\mathbf{Inv}(T, j, k)$ holds in the inner loop of the algorithm (Fig.~\ref{alg:progress-measure-synthesis}): for each $(l, i) \in L \times \{1, \ldots, d\}$,
	\begin{align}
		(l, i) \in T \implies &i \ge 2 j - 1 \land \forall x \in P_{l, i}, r_{(l, i)} \downharpoonright (j, k) \succeq \nexttime (r \downharpoonright (j, k)) \\
		(l, i) \notin T \implies &\big(i < 2 j - 1 \land \text{$i$ is even} \land \forall x \in P_{l, i},r_{(l, i)} \downharpoonright (\lceil i / 2 \rceil, m_{\lceil i / 2 \rceil}) \succeq \nexttime (r \downharpoonright (\lceil i / 2 \rceil,  m_{\lceil i / 2 \rceil})) \\
			&\lor i < 2 j - 1 \land \text{$i$ is odd} \land \forall x \in P_{l, i},r_{(l, i)} \downharpoonright (\lceil i / 2 \rceil, m_{\lceil i / 2 \rceil}) \succ \nexttime (r \downharpoonright (\lceil i / 2 \rceil,  m_{\lceil i / 2 \rceil})) \\
			&\lor \forall x \in P_{l, i}, r_{(l, i)} \downharpoonright (j, k) \succ \nexttime (r \downharpoonright (j, k)) \big)
	\end{align}
	where we write $r \downharpoonright (j, k) = (r_{1, 1}, \dots, r_{1, m_1}, r_{2, 1}, \dots, r_{2, m_2}, \dots, r_{j, 1}, \dots, r_{j, k})$ and similarly, $r_{(l, i)} \downharpoonright (j, k) = (r_{(l, i), 1, 1}, \dots, r_{(l, i), 1, m_1}, r_{(l, i), 2, 1}, \dots, r_{(l, i), 2, m_2}, \dots, r_{(l, i), j, 1}, \dots, r_{(l, i), j, k})$.
	Precisely speaking, we have to ensure that a single level is assigned to each region $P_{l, i}$ according to the definition of LexPMSM maps (Definition~\ref{def:lexpmsm-map}), but we omit this part in the definition of the invariant for simplicity.
	However, it should be noted that whenever we write $\succeq$ or $\succ$ in the invariant, it implicitly means that levels satisfy this requirement.
	Note also that $\mathbf{Inv}(T, j, k)$ implicitly depends on $\{ m_{j'} \}_{j' = 1}^{j - 1}$ and $\{ r_{(l, i), j', k'} \mid (j', k') <_{\mathrm{lex}} (j, k) \land l \in L \land i = 1, \dots, d \}$, but we do not write them explicitly for simplicity.
	We show in Fig.~\ref{fig:algorithm-soundness-proof} the conditions satisfied at each point in the algorithm.
	At Line~\ref{alg:line:t2}, $\{ (l, i) \in T \mid i = 2 j - 2 \}$ is removed from $T$, which does not break the invariant because if $(l, 2 j - 2) \in T$, then we have $r_{(l, 2 j - 2)} \downharpoonright (j - 1, m_{j - 1}) \succeq \nexttime (r \downharpoonright (j - 1, m_{j - 1}))$.
	At Line~\ref{alg:line:t3}, if $(l, i) \in T'$, then we have $i \ge 2 j - 1$ and $r_{(l, i), j, k} \ge 1 + \nexttime (\{r_{(l, i), j, k}\}_{l \in L, i \in \{1, \ldots, d\}})$.
	Thus, removing $T'$ from $T$ does not break the invariant.
	When the algorithm ends, we have $T \subseteq \{ (l, i) \mid i \ge 2 \lceil d / 2 \rceil \}$.
	By the invariant, $r$ is a LexPMSM map.
\end{appendixproof}

\begin{toappendix}
\begin{figure}
	\begin{algorithmic}
		\Function{Synthesize}{$F, P$}
		\State $T \gets \{ (l, i) \mid l \in L, i \in \{1, \ldots, d\} \}$
		\LComment{$\mathbf{Inv}(T, 0, 0)$}
		\For{$j \in \{ 1, 2, \dots, \lceil d / 2 \rceil \}$}
		\LComment{$\mathbf{Inv}(T, j - 1, m_{j - 1}) \land T \subseteq \{ (l, i) \mid i \ge 2 j - 2 \}$}
		\State $T \gets \{ (l, i) \mid (l, i) \in T \land i \ge 2 j - 1 \}$
		\LComment{$\mathbf{Inv}(T, j, 0)$}
		\State $k \gets 0$
		\Repeat
		\LComment{$\mathbf{Inv}(T, j, k)$}
		\State $c_0 \gets \{ P_{l, i} \Rightarrow R_{l, i} \ge \nexttime (\{R_{l, i}\}_{l \in L, i \in \{1, \ldots, d\}}) \mid (l, i) \in T \} \cup \{ P_{l, i} \Rightarrow R_{l, i} \ge 0 \mid l \in L, i \in \{1, \ldots, d\} \}$
		\State $c_1 \gets \{ P_{l, i} \Rightarrow R_{l, i} \ge 1 + \nexttime (\{R_{l, i}\}_{l \in L, i \in \{1, \ldots, d\}}) \mid (l, i) \in T \}$
		\State $\{ r_{(l, i), j, k} \}_{l \in L, i \in \{ 1, \dots, d \}} \gets \Call{Solve}{c_0, c_1}$
		\LComment{$\forall (l, i) \in T, r_{(l, i), j, k} \ge \nexttime (\{r_{(l, i), j, k}\}_{l \in L, i \in \{1, \ldots, d\}})$}
		\State $T' \gets \{ (l, i) \in T \mid\quad \models P_{l, i} \Rightarrow r_{(l, i), j, k} \ge 1 + \nexttime (\{r_{(l, i), j, k}\}_{l \in L, i \in \{1, \ldots, d\}}) \}$
		\LComment{$T' \subseteq T \land \forall (l, i) \in T', r_{(l, i), j, k} \ge 1 + \nexttime (\{r_{(l, i), j, k}\}_{l \in L, i \in \{1, \ldots, d\}})$}
		\State $T \gets T \setminus T'$
		\LComment{$(T' = \emptyset \implies \mathbf{Inv}(T, j, k)) \land (T' \neq \emptyset \implies \mathbf{Inv}(T, j, k + 1))$}
		\State $k \gets k + 1$
		\LComment{$(T' = \emptyset \implies \mathbf{Inv}(T, j, k - 1)) \land (T' \neq \emptyset \implies \mathbf{Inv}(T, j, k))$}
		\Until{$T' = \emptyset$}
		\LComment{$\mathbf{Inv}(T, j, k - 1)$}
		\If{$\{ (l, i) \in T \mid i = 2 j - 1 \} \neq \emptyset$}
			\State \Return{no LexPMSM map found}
		\EndIf
		\LComment{$\mathbf{Inv}(T, j, k - 1) \land \{ (l, i) \in T \mid i = 2 j - 1 \} = \emptyset$}
		\If{$k = 1$}
			\State $m_j \gets 1$
			\State $\{ r_{(l, i), j, 0} \}_{l \in L, i \in \{ 1, \dots, d \}} \gets \{ 0 \}$
		\Else
			\State $m_j \gets k - 1$
		\EndIf
		\LComment{$m_j > 0 \land \mathbf{Inv}(T, j, m_j) \land \{ (l, i) \in T \mid i = 2 j - 1 \} = \emptyset$}
		\EndFor
		\LComment{$\mathbf{Inv}(T, \lceil d / 2 \rceil, m_{\lceil d / 2 \rceil}) \land T \subseteq \{ (l, i) \mid i \ge 2 \lceil d / 2 \rceil \}$}
		\State \Return{$\{ r_{(l, i), j, k} \mid l \in L, i \in \{ 1, \dots, d \}, j \in \OneToN{\lceil d / 2 \rceil}, k \in \OneToN{m_j} \}$}
		\EndFunction
	\end{algorithmic}
	\caption{Proof of Theorem~\ref{thm:algorithm-soundness}.}
	\label{fig:algorithm-soundness-proof}
\end{figure}
\end{toappendix}

The completeness of the algorithm is relative to the completeness of \textsc{Solve}.
We say that \textsc{Solve} is \emph{complete} if for any set $c_0, c_1$ of constraints over function variables, \textsc{Solve} always finds a solution that satisfies all constraints in $c_0$ and at least one constraint in $c_1$ whenever such a solution exists.
Of course, we can modify the above definition by restricting the class of solutions and the class of constraints, e.g., to linear functions and linear constraints, respectively.

\begin{theorem}[relative completeness]\label{thm:algorithm-completeness}
	Assume that \textsc{Solve} is complete.
	If there exists a LexPMSM map for a given pCFG with a priority function, then \Synthesise\ returns a LexPMSM map.
	\qed
\end{theorem}
\begin{appendixproof}[Proof of Theorem~\ref{thm:algorithm-completeness}]
	Let $\hat{r} = \{ \hat{r}_{(l, i), j, k} \mid l \in L; i = 1, \dots, d; j = 1, \dots, \lceil d / 2 \rceil;  k = 1, \dots, \hat{m}_i \}$ be a LexPMSM map whose level function is $\hat{\mathrm{lev}} \colon L \times \{ 1, \dots, d \} \to \mathbf{Lev} \cup \{ \star \}$.
	The algorithm is terminating because the set $T$ is finite and strictly decreases in each iteration of the inner loop.
	Thus, it suffices to show that the algorithm never returns ``no LexPMSM map found''.
	\[ \hat{T}_{j, k} \coloneqq \{ (l, i) \mid (\hat{\mathrm{lev}}(l, i) \neq \star \land \hat{\mathrm{lev}}(l, i) \ge_{\mathrm{lex}} (j, k)) \lor (\hat{\mathrm{lev}}(l, i) = \star \land i \ge 2 j - 1) \} \]
	Here, $(j, k)$ ranges over $\mathbf{Lev}$, and we extend the definition of $\hat{T}_{j, k}$ to $j = 0$ or $k = m_j + 1$ in the obvious way.
	Then, we have the following.
	\begin{itemize}
		\item $\hat{T}_{0, k} = L \times \{ 1, \dots, d \}$.
		\item For each $j = 1, \dots, \lceil d / 2 \rceil$, $\hat{T}_{j, 1} \supseteq \hat{T}_{j, 2} \supseteq \dots \supseteq \hat{T}_{j, m_j + 1}$, and $\hat{T}_{j, k} \setminus \hat{T}_{j, k + 1} = \{ (l, i) \mid \hat{\mathrm{lev}}(l, i) = (j, k) \}$.
		\item For each $j = 1, \dots, \lceil d / 2 \rceil$, $\hat{T}_{j, m_j + 1} \supseteq \hat{T}_{j + 1, 1}$ and $\hat{T}_{j, m_j + 1} \setminus \hat{T}_{j + 1, 1} = \{ (l, i) \mid \hat{\mathrm{lev}}(l, i) = \star \land i = 2 j \}$.
		\item For each $k = 1, \dots, m_j$, if $(l, i) \in \hat{T}_{j, k}$, then $\hat{r}_{(l, i), j, k} \ge \nexttime (\{\hat{r}_{(l, i), j, k}\}_{l \in L, i \in \{1, \ldots, d\}})$.
		\item For each $k = 1, \dots, m_j$, if $(l, i) \in \hat{T}_{j, k} \setminus \hat{T}_{j, k + 1}$, then $\hat{r}_{(l, i), j, k} \ge 1 + \nexttime (\{\hat{r}_{(l, i), j, k}\}_{l \in L, i \in \{1, \ldots, d\}})$.
	\end{itemize}
	Now, it is shown that ``no LexPMSM map found'' is never returned as in Fig.~\ref{fig:algorithm-completeness-proof}.
\end{appendixproof}

\begin{toappendix}
\begin{figure}
	\begin{algorithmic}
		\Function{Synthesize}{$F, P$}
		\State $T \gets \{ (l, i) \mid l \in L, i \in \{1, \ldots, d\} \}$
		\LComment{$T \subseteq \hat{T}_{0, 1}$}
		\For{$j \in \{ 1, 2, \dots, \lceil d / 2 \rceil \}$}
		\LComment{$T \subseteq \hat{T}_{j - 1, \hat{m}_{j - 1} + 1}$}
		\State $T \gets \{ (l, i) \mid (l, i) \in T \land i \ge 2 j - 1 \}$
		\LComment{$T \subseteq \hat{T}_{j, 1}$}
		\State $k \gets 0$
		\Repeat
		\State $c_0 \gets \{ P_{l, i} \Rightarrow R_{l, i} \ge \nexttime (\{R_{l, i}\}_{l \in L, i \in \{1, \ldots, d\}}) \mid (l, i) \in T \} \cup \{ P_{l, i} \Rightarrow R_{l, i} \ge 0 \mid l \in L, i \in \{1, \ldots, d\} \}$
		\State $c_1 \gets \{ P_{l, i} \Rightarrow R_{l, i} \ge 1 + \nexttime (\{R_{l, i}\}_{l \in L, i \in \{1, \ldots, d\}}) \mid (l, i) \in T \}$
		\LComment{If $(\hat{T}_{j, 1} \setminus \hat{T}_{j, \hat{m}_j + 1}) \cap T \neq \emptyset$, then there exists $k$ such that $\{ \hat{r}_{(l, i), j, k} \}$ satisfies $c_0$ and at least one constraint in $c_1$.}
		\State $\{ r_{(l, i), j, k} \}_{l \in L, i \in \{ 1, \dots, d \}} \gets \Call{Solve}{c_0, c_1}$
		\State $T' \gets \{ (l, i) \in T \mid\quad \models P_{l, i} \Rightarrow r_{(l, i), j, k} \ge 1 + \nexttime (\{r_{(l, i), j, k}\}_{l \in L, i \in \{1, \ldots, d\}}) \}$
		\LComment{$(\hat{T}_{j, 1} \setminus \hat{T}_{j, \hat{m}_j + 1}) \cap T \neq \emptyset \implies T' \neq \emptyset$}
		\State $T \gets T \setminus T'$
		\State $k \gets k + 1$
		\Until{$T' = \emptyset$}
		\LComment{$T \subseteq \hat{T}_{j, \hat{m}_j + 1}$}
		\If{$\{ (l, i) \in T \mid i = 2 j - 1 \} \neq \emptyset$}
			\LComment{false}
			\State \Return{no LexPMSM map found}
		\EndIf
		\If{$k = 1$}
			\State $m_j \gets 1$
			\State $\{ r_{(l, i), j, 0} \}_{l \in L, i \in \{ 1, \dots, d \}} \gets \{ 0 \}$
		\Else
			\State $m_j \gets k - 1$
		\EndIf
		\EndFor
		\State \Return{$\{ r_{(l, i), j, k} \mid l \in L, i \in \{ 1, \dots, d \}, j \in \OneToN{\lceil d / 2 \rceil}, k \in \OneToN{m_j} \}$}
		\EndFunction
	\end{algorithmic}
	\caption{Proof of Theorem~\ref{thm:algorithm-completeness}.}
	\label{fig:algorithm-completeness-proof}
\end{figure}
\end{toappendix}

\paragraph{Computational complexity.}
When synthesising Streett supermartingales \cite{AbateCAV2024}, we need to solve constraints of size $O(|L|)$ for each of $d / 2$ Streett pairs obtained from Lemma~\ref{lem:streett-pair-parity}.
On the other hand, when synthesising LexPMSM maps, we need to solve constraints $c_0$ and $c_1$ of size $O(d \cdot |L|)$ for each iteration of the inner loop.
The number of components of LexPMSMs (i.e.\ $\sum_i m_i$) is upper bounded by $O(d \cdot |L|)$ since the size of $T'$ strictly decreases for each iteration.

\begin{toappendix}
\subsection{Synthesising LexPMSM maps with invariants}\label{sec:invariant-synthesis}
It is possible to synthesise invariants together with LexPMSM maps by modifying the algorithm in Fig.~\ref{alg:progress-measure-synthesis} as follows.

\begin{definition}
	Let $F : S \to \Giry S$ be a Markov chain.
	A measurable subset $I \subseteq S$ is an \emph{invariant} if for each $s \in I$, the support of $F(s)$ is included in $I$: $\forall s \in I, \mathrm{supp}(F(s)) \subseteq I$.
\end{definition}

\begin{definition}
	Given a pCFG $\{ f_l \colon \mathbb{R}^n \to \Giry (L \times \mathbb{R}^n) \}_{l \in L}$, we say $\{ I_l \subseteq \mathbb{R}^n \}_{l \in L}$ is an \emph{invariant map} if $I \coloneqq \bigcup_{l \in L} \{ l \} \times I_l$ is an invariant of the Markov chain on $L \times \mathbb{R}^n$ induced by the pCFG, i.e.,
	\begin{equation}
		\forall l \in L, \forall x \in I_l, \mathrm{supp}(f_l(x)) \subseteq \bigcup_{l \in L} \{ l \} \times I_l.
		\label{eq:invariant-map}
	\end{equation}
\end{definition}

\begin{example}
	Consider the pCFG in Fig.~\ref{ex:no-streett-supermartingale}.
	Then, $\{ I_l \}_{\{ l = l_0, l_1 \}}$ is an invariant map if the following holds.
	\begin{align}
		&(m, n) \in I_{l_0} \implies (n, n) \in I_{l_1} \\
		&(m, n) \in I_{l_1} \land m > 0 \implies (m - 1, n) \in I_{l_1} \\
		&(m, n) \in I_{l_1} \land \lnot (m > 0) \implies (m, n + 1) \in I_{l_0}
	\end{align}
\end{example}

An invariant $I$ induces a sub-Markov chain $F|_I : I \to \Giry I$ in the obvious way, and we can restrict our attention to this sub-Markov chain when we want to prove almost-sure termination from states in $I$.
\begin{definition}
	Let $F \colon S \to \Giry S$ be a parity Markov chain with a priority function $p \colon S \to \OneToN{d}$ and $I$ be an invariant of $F$.
	A \emph{LexPMSM on $I$} is a non-negative measurable function $r \colon S \to \prod_{i = 1}^d [0, \infty)^{m_i}$ that satisfies the condition of Definition~\ref{def:lexpmsm} only for states in $I$:
	for each $x \in I$, if $p(x)$ is even, then $r(x) \succeq^{(2)}_{p(x)} (\nexttime r)(x)$, and if $p(x)$ is odd, then $r(x) \succ^{(2)}_{p(x)} (\nexttime r)(x)$.
\end{definition}

We define the notion of LexPMSM maps on invariant maps for pCFGs similarly.
By restricting the domain of $r$ to the invariant $I$, we obtain a LexPMSM  for the sub-Markov chain $F|_I$, which ensures almost-sure satisfaction of the parity condition from states in $I$.

Fig.~\ref{alg:progress-measure-synthesis-with-invariants} shows an algorithm for synthesising a LexPMSM map together with an invariant map.
The main difference from the algorithm in Fig.~\ref{alg:progress-measure-synthesis} is that we add constraints $c_I$, which ensure that $\{ I_l \}_{l \in L}$ is an invariant map.
In Fig.~\ref{alg:progress-measure-synthesis-with-invariants}, we assume that \textsc{Solve} can find a solution for $I$ and $R$.
The algorithm synthesises an invariant map for each iteration, and the conjunction of these invariant maps is returned as an invariant map for the whole algorithm.
\begin{figure}[tbp]
	\begin{algorithmic}
		\Require A pCFG $F = \{ f_l \colon \mathbb{R}^n \to \Giry (L \times \mathbb{R}^n) \}_{l \in L}$ with a priority partition $P = \{ P_{l, i} \}_{l \in L, i \in \{1, \ldots, d\}}$ and an initial configuration $(l_0, x_0) \in L \times \mathbb{R}^n$
		\Ensure An invariant map and a LexPMSM map, or ``not found''
		\Function{SynthesiseLexPMSMAndInvariant}{$F, P$}
		\State $\{ \mathbf{Inv}_l \}_{l \in L} \gets \{ \mathbb{R}^n \}_{l \in L}$
		\State $T \gets \{ (l, i) \mid l \in L, i \in \{1, \ldots, d\} \}$
		\For{$j \in \{ 1, 2, \dots, \lceil d / 2 \rceil \}$}
		\State $T \gets \{ (l, i) \mid (l, i) \in T \land i \ge 2 j - 1 \}$
		\State $k \gets 0$
		\Repeat
		\State $c_I \gets \{ x_0 \in I_{l_0} \} \cup \{ x \in I_l \land (l', x') \in \mathrm{supp}(f_l(x)) \implies x' \in I_{l'} \mid l, l' \in L \}$
		\State $c_0 \gets \{ I_l \land P_{l, i} \Rightarrow R_{l, i} \ge \nexttime (\{R_{l, i}\}_{l \in L, i \in \{1, \ldots, d\}}) \mid (l, i) \in T \} \cup \{ I_l \land P_{l, i} \Rightarrow R_{l, i} \ge 0 \mid l \in L, i \in \{1, \ldots, d\} \}$
		\State $c_1 \gets \{ I_l \land P_{l, i} \Rightarrow R_{l, i} \ge 1 + \nexttime (\{R_{l, i}\}_{l \in L, i \in \{1, \ldots, d\}}) \mid (l, i) \in T \}$
		\State $\{ I_l \}_{l \in L},\ \{ r_{(l, i), j, k} \}_{l \in L, i \in \{ 1, \dots, d \}} \gets \Call{Solve}{c_I \cup c_0, c_1}$
		\State $T' \gets \{ (l, i) \in T \mid\quad \models I_l \land P_{l, i} \Rightarrow r_{(l, i), j, k} \ge 1 + \nexttime (\{r_{(l, i), j, k}\}_{l \in L, i \in \{1, \ldots, d\}}) \}$
		\State $T \gets T \setminus T'$
		\State $k \gets k + 1$
		\If{$T' \neq \emptyset$}
			\State $\{ \mathbf{Inv}_l \}_{l \in L} \gets \{ \mathbf{Inv}_l \cap I_l \}_{l \in L}$
		\EndIf
		\Until{$T' = \emptyset$}
		\If{$\{ (l, i) \in T \mid i = 2 j - 1 \} \neq \emptyset$}
			\State \Return{not found}
		\EndIf
		\If{$k = 1$}
			\State $m_j \gets 1$
			\State $\{ r_{(l, i), j, 0} \}_{l \in L, i \in \{ 1, \dots, d \}} \gets \{ 0 \}$
		\Else
			\State $m_j \gets k - 1$
		\EndIf
		\EndFor
		\State \Return{$\{ \mathbf{Inv}_l \}_{l \in L},\quad \{ r_{(l, i), j, k} \mid l \in L, i \in \{ 1, \dots, d \}, j \in \OneToN{\lceil d / 2 \rceil}, k \in \OneToN{m_j} \}$}
		\EndFunction
	\end{algorithmic}
	\caption{Algorithm for synthesizing a LexPMSM map with invariants.}
	\label{alg:progress-measure-synthesis-with-invariants}
\end{figure}

\subsection{Synthesising LexGSSMs}
\label{sec:lexgssm-synthesis}

Suppose we have a pCFG $\{ f_l \colon \mathbb{R}^n \to \Giry (L \times \mathbb{R}^n) \}_{l \in L}$.
A Streett pair for the pCFG is given as $(A, B) = (\bigcup_{l \in L} \{ l \} \times A_l, \bigcup_{l \in L} \{ l \} \times B_l)$ where $A_l, B_l \subseteq \mathbb{R}^n$ are measurable subsets for each $l \in L$.
We call $\{ (A_l, B_l) \}_{l \in L}$ a \emph{Streett pair map} for the pCFG.
Note that a Streett pair map represents a single Streett pair, and this is not the conjunction of multiple Streett pairs.

\begin{definition}
	Suppose we have a pCFG $\{ f_l \colon \mathbb{R}^n \to \Giry (L \times \mathbb{R}^n) \}_{l \in L}$, a Streett pair map $\{ (A_l, B_l) \}_{l \in L}$, and an invariant map $\{ I_l \}_{l \in L}$.
	A \emph{LexGSSM map on the invariant map} is a pair of a measurable function and a family of measurable functions
	\[ \mathrm{lev}_1 \colon L \to \{ 1, \dots, m \}, \qquad \mathrm{lev}_2 \colon L \to \{ 1, \dots, m \} \cup \{ \star \}, \qquad \{ r_{l, j} \colon \mathbb{R}^n \to \mathbb{R} \}_{l \in L, j \in \OneToN{m}} \]
	such that the function $r \colon L \times \mathbb{R}^n \to \mathbb{R}^m$ defined by $r(l, x) \coloneqq (r_{l, 1}(x), \dots, r_{l, m}(x))$ satisfies the following conditions.
	\begin{itemize}
		\item The function $r_{l, j}$ is non-negative on $I_l$.
		\item If $\mathrm{lev}_1(l) = j$, then $r(l, x) \succ_{[j]} (\mathbb{X} r)(l, x)$ for any $x \in I_l \cap A_l \setminus B_l$.
		\item If $\mathrm{lev}_2(l) = j \neq \star$, then $r(l, x) \succ_{[j]} (\mathbb{X} r)(l, x)$ for any $x \in I_l \setminus (A_l \cup B_l)$.
		\item If $\mathrm{lev}_2(l) = \star$, then $r(l, x) \ge (\mathbb{X} r)(l, x)$ for any $x \in I_l \setminus (A_l \cup B_l)$.
	\end{itemize}
	Similarly to Definition~\ref{def:lexpmsm-map}, this definition ensures that the function $r$ is non-negative and satisfies the ranking condition for LexGSSMs (Definition~\ref{def:lexicographic-generalised-streett-supermartingale}) for each state in the invariant.
	We often omit the level functions $\mathrm{lev}_1$ and $\mathrm{lev}_2$ when they are clear from the context.
\end{definition}

Fig.~\ref{alg:lex-gssm-synthesis} shows an algorithm for synthesising LexGSSM maps.
This is a straightforward extension of the algorithm in \cite{AgrawalPOPL2018}.

\begin{figure}[tbp]
	\begin{algorithmic}
		\Require A pCFG $F = \{ f_l \colon \mathbb{R}^n \to \Giry (L \times \mathbb{R}^n) \}_{l \in L}$ with an invariant map $I = \{ I_l \}_{l \in L}$ and a Streett pair map $\{ (A_l, B_l) \}_{l \in L}$
		\Ensure A LexGSSM map, or ``not found''
		\Function{SynthesiseLexGSSM}{$F, I, \{ (A_l, B_l) \}_{l \in L}$}
		\State $S \gets L$
		\State $T \gets L$
		\State $m \gets 0$
		\Repeat
		\State $c_0 \gets \{ I_l \Rightarrow R_l \ge 0 \mid l \in L \}$
		\State $c_0 \gets c_0 \cup \{ I_l \land \lnot A_l \land \lnot B_l \Rightarrow R_l \ge \nexttime (\{ R_l \}_{l \in L}) \mid l \in S \}$
		\State $c_0 \gets c_0 \cup \{ I_l \land A_l \land \lnot B_l \Rightarrow R_l \ge \nexttime (\{ R_l \}_{l \in L}) \mid l \in T \}$
		\State $c_1 \gets \{ I_l \land \lnot A_l \land \lnot B_l \Rightarrow R_l \ge 1 + \nexttime (\{ R_l \}_{l \in L}) \mid l \in S \}$
		\State $c_1 \gets c_1 \cup \{ I_l \land A_l \land \lnot B_l \Rightarrow R_l \ge 1 + \nexttime (\{ R_l \}_{l \in L}) \mid l \in T \}$
		\State $\{ r_{l, m} \}_{l \in L} \gets \Call{Solve}{c_0, c_1}$
		\State $S' \gets \{ l \in S \mid\quad \models I_l \land \lnot A_l \land \lnot B_l \Rightarrow r_{l, m} \ge 1 + \nexttime (\{ r_{l, m} \}_{l \in L}) \}$
		\State $S \gets S \setminus S'$
		\State $T' \gets \{ l \in T \mid\quad \models I_l \land A_l \land \lnot B_l \Rightarrow r_{l, m} \ge 1 + \nexttime (\{ r_{l, m} \}_{l \in L}) \}$
		\State $T \gets T \setminus T'$
		\State $m \gets m + 1$
		\Until{$T = \emptyset$ or $S' \cup T' = \emptyset$}
		\If{$T \neq \emptyset$}
			\State \Return{not found}
		\EndIf
		\State \Return{$\{ r_{l, j} \mid l \in L, j \in \OneToN{m} \}$}
		\EndFunction
	\end{algorithmic}
	\caption{Algorithm for synthesizing a LexGSSM map with invariants.}
	\label{alg:lex-gssm-synthesis}
\end{figure}

\paragraph{Computational complexity.}
Each component requires solving constraints of size $O(|L|)$, as in the case of synthesising SSMs/GSSMs.
Hence, the additional cost is proportional to the number of components in the lexicographic functions, which is upper bounded by $O(|L|)$, i.e., the size of $T$.

\end{toappendix}

\section{Implementation and Experiments}
\label{sec:implementation}

\subsection{Polynomial Templates}
We implement the synthesis algorithm for LexPMSMs described in Section~\ref{sec:algorithm} using polynomial templates.
That is, we instantiate \textsc{Solve} in Fig.~\ref{alg:progress-measure-synthesis} so that it searches for a solution from polynomials over program variables with unknown coefficients where the degree of the polynomials is specified in a configuration file.
Assuming that expressions in input programs (pCFGs) are polynomials, the constraints $c_0$ and $c_1$ in Line~\ref{alg:line:c0},~\ref{alg:line:c1} can be expressed as \emph{polynomial quantified entailments} (PQEs) \cite{ChatterjeeATVA2025}, which are constraints of the form
$\forall \tilde{x} \in \mathbb{R}^n, \bigwedge_{i = 1}^m (p_i(\tilde{t}, \tilde{x}) \bowtie 0) \implies p(\tilde{t}, \tilde{x}) \bowtie 0$
where $p_1, \dots, p_m, p$ are polynomials, $\tilde{t}$ are unknown parameters to be solved, and $\bowtie \in \{ \ge, > \}$.
We use the technique in Remark~\ref{rem:solve-implementation} to handle constraints in $c_1$.
Then, we solve the resulting PQEs using an off-the-shelf solver \textsc{PolyQEnt} \cite{ChatterjeeATVA2025}.
\textsc{PolyQEnt} uses positivity theorems for polynomials, such as Farkas' lemma, Handelman's theorem, and Putinar's theorem, to efficiently solve PQEs.
These positivity theorems have completeness guarantees under certain conditions as discussed in \cite{ChatterjeeATVA2025}.
We also implement synthesis algorithms for LexGSSMs, GSSMs, and SSMs in a similar manner.

In theory, our supermartingales can reason about continuous distributions.
However, our current implementation supports only discrete distributions with finite supports, as all benchmarks we consider can be represented using such distributions.
This restriction is not inherent to our algorithm: supporting continuous distributions is straightforward, because integrals of polynomial templates over continuous distributions can be computed symbolically provided that $k$-th moments of the distributions are known for each $k$.
It is often desirable to synthesise supermartingales together with their supporting invariants.
We note that the invariant-synthesis approach in \cite{AbateCAV2024} could be incorporated into our algorithm for LexPMSMs by taking the conjunction of invariants obtained during the iteration of the synthesis process.
A naive algorithm for invariant synthesis is presented in \referappendix{Appendix}{E.1}{sec:invariant-synthesis}.
As a final remark, our implementation takes \emph{quantitative fixed-point equation systems} as the input format, which are systems of fixed-point equations over $[0, 1]$-valued functions and serve as an alternative representation of pCFGs with priority partitions.
These two representations are known to be equivalent \cite{CirsteaCALCO2017}.

\subsection{Experiments}
We evaluate our implementation $\textsc{MuVal}^{\textsc{QFL}}$ (\url{https://github.com/hiroshi-unno/coar}) to demonstrate the effectiveness of our novel supermartingale-based certificates and to compare their verification power.
Benchmarks are drawn from existing literature \cite{AbateCAV2024} and our examples (Example~\ref{ex:generalised-streett-supermartingale}\footnote{We modified Example~\ref{ex:generalised-streett-supermartingale} as in \referappendix{Example}{F.1}{ex:gssm-benchmark}, where the desired probability distribution is generated by a loop.},~\ref{ex:no-streett-supermartingale}, and~\ref{ex:lexgssm}).
Compared to \cite{AbateCAV2024}, $\textsc{MuVal}^{\textsc{QFL}}$ currently does not support invariant and parameter synthesis.
Thus, we manually provide invariants and omit benchmarks that require parameter synthesis.
Table~\ref{tab:experimental-results} shows our experimental results.
Our experiments were conducted on Apple M3 Pro (12-core CPU) with 18 GB of memory.
$\textsc{MuVal}^{\textsc{QFL}}$ successfully synthesised LexPMSMs and LexGSSMs for all benchmarks, including Example~\ref{ex:generalised-streett-supermartingale},~\ref{ex:no-streett-supermartingale}, and~\ref{ex:lexgssm}.
In contrast, GSSMs and SSMs could not be synthesised for some benchmarks, as expected from the theoretical results.
We observed that LexPMSMs often require fewer lexicographic components than LexGSSMs, which leads to faster synthesis times.
This observation is consistent with the following theoretical result: LexGSSMs can be translated to LexPMSMs without increasing the number of components (Proposition~\ref{prop:lexgssm-to-lexpmsm}), while the converse translation (Proposition~\ref{prop:lexpmsm-to-lexgssm}) increases the number of components.
Although these translations are not necessarily optimal in terms of the number of components, they suggest that LexPMSMs require no more components than LexGSSMs in general.

\begin{table}[tb]
	\caption{Experimental results.
	The degree of polynomial templates was set to $1$.
	Entries in parentheses ``(..)'' indicate that the synthesis algorithm returned ``not found'', while others indicate that the algorithm returned a solution.}
	\label{tab:experimental-results}
	\small
	\begin{tabular}{c|cccc}
		& \multicolumn{4}{c}{time (sec)} \\
		benchmark & LexPMSM & LexGSSM & GSSM & SSM \\
		\hline
		\texttt{ex\_3\_8} & 0.845 & 0.793 & 0.384 & (0.367) \\
		\texttt{ex\_3\_9} & 0.532 & 0.553 & 0.473 & (0.434) \\
		\texttt{ex\_3\_9\_rw} & 0.807 & 0.864 & 0.633 & (0.690) \\
		\texttt{ex\_4\_11} & 1.554 & 1.596 & (0.465) & (0.502) \\
		\texttt{EvenOrNegative} & 0.666 & 1.067 & 0.351 & 0.366 \\
		\texttt{PersistRW} & 0.425 & 0.442 & 0.220 & 0.213 \\
		\texttt{RecurRW} & 0.534 & 0.511 & 0.295 & 0.288 \\
		\texttt{GuaranteeRW} & 6.168 & 9.098 & 3.568 & 4.721 \\
		\texttt{Temperature2} & 12.832 & 18.550 & 5.750 & 6.213
	\end{tabular}
\end{table}

\begin{toappendix}
\subsection{Benchmarks}
\begin{example}\label{ex:gssm-benchmark}
	The Markov chain in Example~\ref{ex:generalised-streett-supermartingale} and~\ref{ex:lexgssm} can be represented by the following program.
	\[ \while{\mathtt{true}}{\ifstmt{x \ge 1}{\assign{x}{x - 1}}{(\assign{x}{2}; \while{\mathtt{prob}(1/2)}{\assign{x}{2 * x}})}} \]
	\begin{center}
		\begin{tikzpicture}
			\node[initial, state] (s0) {$l_0$};
			\node[state, right=10em of s0] (s1) {$l_1$};
			\draw[->] (s0) edge[bend left=10] node[above, align=center] {$[\lnot (x \ge 1)]$ \\ $\assign{x}{2}$} (s1);
			\draw[->] (s1) edge[bend left=10] node[below, align=center] {$1 / 2$} (s0);
			\draw[->] (s1) edge[loop right] node[right, align=center] {$1 / 2$\\$\assign{x}{2 * x}$} (s1);
			\draw[->] (s0) edge[loop above] node[left, align=center] {$[x \ge 1]$\\$\assign{x}{x - 1}$} (s0);
		\end{tikzpicture}
	\end{center}
	The Streett condition is given as follows.
	\[ (A, B) \quad=\quad (\{ (l_0, x) \mid x \ge 1 \},\quad \{ (l_0, x) \mid x = 0 \} \cup \{ (l_1, x) \mid x \in \mathbb{N} \}) \]

	For experiments, we slightly modify the above example and assume that the variable $x$ takes values in $\mathbb{R}$.
	We relate $\mathbb{R}$-valued variables to $\mathbb{N}$-valued variables using the floor function $\lfloor \cdot \rfloor \colon \mathbb{R} \to \mathbb{N}$.
	Now, the Streett condition is as follows.
	\[ (A, B) \quad=\quad (\{ (l_0, x) \mid \lfloor x \rfloor \ge 1 \},\quad \{ (l_0, x) \mid \lfloor x \rfloor = 0 \} \cup \{ (l_1, x) \mid x \in \mathbb{R} \}) \]
	This Streett pair corresponds to the following parity condition.
	\begin{align}
		p^{-1}(2) &= \{ (l_0, x) \mid \lfloor x \rfloor = 0 \} \cup \{ (l_1, x) \mid x \in \mathbb{R} \} \\
		p^{-1}(3) &= \{ (l_0, x) \mid \lfloor x \rfloor \ge 1 \} \\
		p^{-1}(4) &= \{ (l_0, x) \mid \lfloor x \rfloor < 0 \}
	\end{align}
	We have the following invariant.
	\[ I \quad=\quad \{ (l, x) \mid l \in \{ l_0, l_1 \}, x \ge 0 \} \]
	Then, we obtain the following fixed-point equation system.
	\begin{align}
		X_{l_0, 2}(x) &=_{\nu} \ifexpr{x \ge 1}{X_{l_0}(x - 1)}{X_{l_1}(2)} \\
		X_{l_1, 2}(x) &=_{\nu} \frac{1}{2} X_{l_1}(2 x) + \frac{1}{2} X_{l_0}(x) \\
		X_{l_0, 3}(x) &=_{\mu} \ifexpr{x \ge 1}{X_{l_0}(x - 1)}{X_{l_1}(2)} \\
		X_{l_0}(x) &=_{\nu} \ifexpr{x \ge 1}{X_{l_0, 3}(x)}{\ifexpr{x \ge 0}{X_{l_0, 2}(x)}{Y()}} \\
		X_{l_1}(x) &=_{\nu} \ifexpr{x \ge 0}{X_{l_1, 2}(x)}{Y()} \\
		Y() &=_{\nu} Y()
	\end{align}
	We may replace each occurrence of $X_{l_0}(x)$ in the right-hand side of the equations with $\ifexpr{x \ge 1}{X_{l_0, 3}(x)}{\ifexpr{x \ge 0}{X_{l_0, 2}(x)}{Y()}}$ and remove the equation for $X_{l_0}(x)$ since this does not change the solution.
	This equation system can be understood as modifying the pCFG so that each location has a single priority.
	We introduce a new location $(l_0, 2)$ and $(l_0, 3)$.
	We then copy all outgoing transitions of $l_0$ to both $(l_0, 2)$ and $(l_0, 3)$ and remove outgoing transitions from $l_0$.
	Finally, we add conditional branching from $l_0$ to $(l_0, 2)$ and $(l_0, 3)$ according to the priority partition.
	\begin{center}
		\begin{tikzpicture}
			\node[state] (l0old) at (-5, 0) {$l_0$};
			\draw[->] (-6, 0) -- (l0old);
			\draw[->] (l0old) -- (-4, 0);
			\node at (-2.5, 0) {$\mapsto$};
			\node[state] (l0) at (0, 0) {$l_0$};
			\node[state] (l02) at (2.5, 0.5) {$l_0, 2$};
			\node[state] (l03) at (2.5, -0.5) {$l_0, 3$};
			\draw[->] (-1, 0) -- (l0);
			\draw[->] (l0) -- node[midway, above] {$[x = 0]$} (l02);
			\draw[->] (l0) -- node[midway, below] {$[x > 0]$} (l03);
			\draw[->] (l02) -- (3.5, 0.5);
			\draw[->] (l03) -- (3.5, -0.5);
		\end{tikzpicture}
	\end{center}

	As for the Streett pair $(A, B) = (\{ (l_0, x) \mid \lfloor x \rfloor \ge 2 \}, \{ (l_0, x) \mid \lfloor x \rfloor = 1 \})$, we obtain the following equation system.
	\begin{align}
		X_{l_0, 2}(x) &=_{\nu} \ifexpr{x \ge 1}{X_{l_0}(x - 1)}{X_{l_1}(2)} \\
		X_{l_0, 3}(x) &=_{\mu} \ifexpr{x \ge 1}{X_{l_0}(x - 1)}{X_{l_1}(2)} \\
		X_{l_0, 4}(x) &=_{\nu} \ifexpr{x \ge 1}{X_{l_0}(x - 1)}{X_{l_1}(2)} \\
		X_{l_1, 4}(x) &=_{\nu} \frac{1}{2} X_{l_1}(2 x) + \frac{1}{2} X_{l_0}(x) \\
		X_{l_0}(x) &=_{\nu} \ifexpr{x \ge 2}{X_{l_0, 3}(x)}{\ifexpr{x \ge 1}{X_{l_0, 2}(x)}{\ifexpr{x \ge 0}{X_{l_0, 4}(x)}{Y()}}} \\
		X_{l_1}(x) &=_{\nu} \ifexpr{x \ge 0}{X_{l_1, 4}(x)}{Y()} \\
		Y() &=_{\nu} Y()
	\end{align}
\end{example}
\begin{example}
	The Streett Markov chain in Example~\ref{ex:no-streett-supermartingale} corresponds to the following fixed-point equation system.
	\begin{align}
		X_{l_0}(n) &=_{\nu} X_{l_1}(n, n) \\
		X_{l_1}(m, n) &=_{\mu} \ifexpr{m > 0}{X_{l_1}(m - 1, n)}{X_{l_0}(n + 1)}
	\end{align}
	For experiments, we modify the example so that the variables $m$ and $n$ take values in $\mathbb{R}$.
	We have the following invariant.
	\[ I \quad=\quad \{ (l, m, n) \mid l \in \{ l_0, l_1 \}, m > -1, n \ge 0 \} \]
	Taking the invariant into account, we obtain the following fixed-point equation system.
	\begin{align}
		X_{l_0}(m, n) &=_{\nu} \ifexpr{m > -1 \land n \ge 0}{X_{l_1}(n, n)}{Y()} \\
		X_{l_1}(m, n) &=_{\mu} \ifexpr{m > -1 \land n \ge 0}{\ifexpr{m > 0}{X_{l_1}(m - 1, n)}{X_{l_0}(m, n + 1)}}{Y()} \\
		Y() &=_{\nu} Y()
	\end{align}
\end{example}
\end{toappendix}

\section{Related Work}

\paragraph{Supermartingales for $\omega$-regular properties}
The most relevant to our work is Streett supermartingales \cite{AbateCAV2024} and limit-deterministic B\"uchi supermartingales \cite{HenzingerCAV2025} for verifying almost sure satisfaction of $\omega$-regular properties.
These two notions are, to the best of our knowledge, the only known supermartingale-based certificates for general $\omega$-regular properties, but as we have shown in Example~\ref{ex:no-streett-supermartingale}, their power is not very strong compared to what we proposed in this paper.
Earlier work \cite{ChakarovTACAS2016} had proposed certificates for recurrence properties, but these can be regarded as special cases of the Streett supermartingale.
Quantitative verification for $\omega$-regular properties are also studied in \cite{AbateCAV2025,HenzingerCAV2025} using the idea of combining stochastic invariants \cite{ChatterjeePOPL2017} with supermartingales for \emph{qualitative} verification of $\omega$-regular properties.
In \cite{AbateCAV2025}, stochastic invariants are witnessed by stochastic invariant indicator \cite{TakisakaTOPLAS2021,ChatterjeeCAV2022}, and Streett acceptance conditions are witnessed by Streett supermartingale \cite{AbateCAV2024}.
On the other hand, in \cite{HenzingerCAV2025}, stochastic invariants are witnessed by repulsing supermartingales \cite{ChatterjeePOPL2017}, and B\"uchi acceptance conditions are witnessed by the method in \cite{ChakarovTACAS2016}.
Note that the latter work considers products of MDPs and limit-deterministic B\"uchi automata, which necessarily introduce nondeterminism even when the original system is a Markov chain.
While the introduction of (demonic) nondeterminism does not affect the soundness of our supermartingale-based certificates, it complicates the completeness arguments. Therefore, in this paper, we restrict our attention to the setting without nondeterminism.
The completeness of ranking supermartingales in the presence of nondeterminism has been studied in~\cite{FuVMCAI2019}; however, extending our completeness results to this setting is left for future work.

\paragraph{Lexicographic ranking supermartingales}
Our development in Section~\ref{sec:lexgssm} and Section~\ref{sec:progress-measure} is inspired by lexicographic ranking supermartingales \cite{AgrawalPOPL2018} for almost-sure termination.
In fact, our definition of LexGSSMs and LexPMSMs can be seen as a generalisation of lexicographic ranking supermartingales.
On the other hand, it is known that the strong non-negativity condition used in \cite{AgrawalPOPL2018} sometimes limits the applicability of synthesis algorithms, and thus relaxations of non-negativity are studied in \cite{ChatterjeeFAC2023,TakisakaCAV2024}.
We believe that our definitions of LexGSSMs and LexPMSMs can also be relaxed similarly, but we leave this for future work.

\paragraph{Parity progress measures for probabilistic systems}
Markov chains can be seen as a special case of \emph{stochastic games} where there is no nondeterministic player.
A reduction from qualitative stochastic parity games (i.e., stochastic games with parity objectives and almost-sure winning criterion) to ordinary parity games is given in \cite{ChatterjeeCSL2003}.
In their proof, a probabilistic extension of parity progress measures \cite{JurdzinskiSTACS2000} is used.
However, their extension is different from our LexPMSMs and relies on finiteness of the state space of stochastic games, whereas our LexPMSMs are applicable to Markov chains with infinite state spaces.

\paragraph{Non-probabilistic $\omega$-regular verification}
B\"uchi ranking functions \cite{ChatterjeeFM2025} are certificates for verifying $\omega$-regular properties of nondeterministic programs.
The idea is to use a ranking function that ensures infinitely many visits to the accepting states of a B\"uchi condition and is conceptually close to the idea of ranking supermartingales for $\omega$-regular verification.
A sound and complete method for such properties was previously proposed in \cite{UnnoPOPL2023}, following directly from a corresponding result for branching-time temporal properties expressed in the modal $\mu$-calculus, a logic strictly more expressive than $\omega$-regular properties, where this result for the modal $\mu$-calculus is established by reducing modal $\mu$-calculus model checking to $\mu$CLP, a fixed-point equation system, and further to pfwCSP, a predicate constraint solving problem extending CHCs as proposed in \cite{UnnoCAV2021}.
The quantitative fixed-point equation systems studied in this paper can be viewed as a quantitative extension of the quantifier-free fragment of $\mu$CLP.

\section{Conclusions and Future Work}
We proposed several new supermartingale-based certificates for qualitative verification of $\omega$-regular properties of Markov chains.
These certificates include GSSMs, LexGSSMs, DVSSMs, PMSMs, and LexPMSMs; and they can verify larger classes of problems than the existing approach of Streett supermartingales.
We studied their hierarchy as summarised in Fig.~\ref{fig:hierarchy-of-supermartingales}.
We implemented synthesis algorithms and demonstrated their effectiveness through experiments.
Future work includes extending our approach to Markov decision processes or stochastic games that involve angelic nondeterminism, and applying such extensions to solve a first-order extension of quantitative fixed-point logics with nondeterministic choice operations \cite{McIverTOCL2007}.

\section*{Data Availability Statement}
$\textsc{MuVal}^{\textsc{QFL}}$ source code can be found at \url{https://github.com/hiroshi-unno/coar}. We also provide an artifact \cite{artifact} to reproduce the evaluation in Section~\ref{sec:implementation}, including a snapshot of $\textsc{MuVal}^{\textsc{QFL}}$, docker images, and benchmarks.

\begin{acks}
We would like to thank the anonymous reviewers for their helpful comments and suggestions.
We would also like to thank Takeshi Tsukada for his insightful comments.
This study was supported by JST K Program Grant Number JPMJKP24U2; JSPS KAKENHI Grant Number JP23K24820, JP25H00446, JP25K21183, JP25K24739.
\end{acks}

\bibliographystyle{ACM-Reference-Format}
\bibliography{ref,arxiv}

@misc{arxiv,
  title         = {A Hierarchy of Supermartingales for $\omega$-Regular Verification},
  author        = {Satoshi Kura and Hiroshi Unno},
  year          = {2025},
  eprint        = {2512.00270},
  archiveprefix = {arXiv},
  primaryclass  = {cs.LO},
  url           = {https://arxiv.org/abs/2512.00270}
}

@software{artifact,
  author    = {Kura, Satoshi and
               Unno, Hiroshi},
  title     = {A Hierarchy of Supermartingales for
               $\omega$-Regular Verification -- Artifact
               },
  month     = apr,
  year      = 2026,
  publisher = {Zenodo},
  doi       = {10.5281/zenodo.19396030},
  url       = {https://doi.org/10.5281/zenodo.19396030}
}

@inproceedings{AbateCAV2024,
  title = {Stochastic {{Omega-Regular Verification}} and {{Control}} with {{Supermartingales}}},
  booktitle = {Computer {{Aided Verification}}},
  author = {Abate, Alessandro and Giacobbe, Mirco and Roy, Diptarko},
  year = 2024,
  series = {Lecture {{Notes}} in {{Computer Science}}},
  volume = {14683},
  pages = {395--419},
  publisher = {Springer Nature Switzerland},
  address = {Cham},
  doi = {10.1007/978-3-031-65633-0_18},
  url = {https://link.springer.com/10.1007/978-3-031-65633-0_18},
  urldate = {2025-03-06},
  isbn = {978-3-031-65632-3 978-3-031-65633-0},
  langid = {english}
}

@inproceedings{AbateCAV2025,
  title = {Quantitative Supermartingale Certificates},
  booktitle = {Computer Aided Verification},
  author = {Abate, Alessandro and Giacobbe, Mirco and Roy, Diptarko},
  year = 2025,
  pages = {3--28},
  publisher = {Springer Nature Switzerland},
  address = {Cham},
  doi = {10.1007/978-3-031-98679-6_1},
  isbn = {978-3-031-98679-6}
}

@article{AgrawalPOPL2018,
  title = {Lexicographic Ranking Supermartingales: An Efficient Approach to Termination of Probabilistic Programs},
  shorttitle = {Lexicographic Ranking Supermartingales},
  author = {Agrawal, Sheshansh and Chatterjee, Krishnendu and Novotn{\'y}, Petr},
  year = 2018,
  month = jan,
  journal = {Proceedings of the ACM on Programming Languages},
  volume = {2},
  number = {POPL},
  pages = {1--32},
  issn = {2475-1421, 2475-1421},
  doi = {10.1145/3158122},
  url = {http://dl.acm.org/doi/10.1145/3158122},
  urldate = {2020-02-25},
  langid = {english}
}

@book{Baier2008,
  title = {Principles of Model Checking},
  author = {Baier, Christel and Katoen, Joost-Pieter},
  year = 2008,
  publisher = {MIT Press},
  address = {Cambridge, Mass.},
  isbn = {978-0-262-02649-9},
  langid = {english}
}

@inproceedings{ChakarovCAV2013,
  title = {Probabilistic {{Program Analysis}} with {{Martingales}}},
  booktitle = {Computer {{Aided Verification}}},
  author = {Chakarov, Aleksandar and Sankaranarayanan, Sriram},
  year = 2013,
  series = {Lecture {{Notes}} in {{Computer Science}}},
  volume = {8044},
  pages = {511--526},
  publisher = {Springer Berlin Heidelberg},
  address = {Saint Petersburg, Russia},
  doi = {10.1007/978-3-642-39799-8_34},
  url = {http://link.springer.com/10.1007/978-3-642-39799-8_34},
  urldate = {2020-03-02},
  isbn = {978-3-642-39798-1 978-3-642-39799-8}
}

@inproceedings{ChakarovTACAS2016,
  title = {Deductive {{Proofs}} of {{Almost Sure Persistence}} and {{Recurrence Properties}}},
  booktitle = {Tools and {{Algorithms}} for the {{Construction}} and {{Analysis}} of {{Systems}}},
  author = {Chakarov, Aleksandar and Voronin, Yuen-Lam and Sankaranarayanan, Sriram},
  year = 2016,
  series = {Lecture {{Notes}} in {{Computer Science}}},
  volume = {9636},
  pages = {260--279},
  publisher = {Springer Berlin Heidelberg},
  address = {Berlin, Heidelberg},
  doi = {10.1007/978-3-662-49674-9_15},
  url = {http://link.springer.com/10.1007/978-3-662-49674-9_15},
  urldate = {2025-07-05},
  isbn = {978-3-662-49673-2 978-3-662-49674-9},
  langid = {english}
}

@inproceedings{ChatterjeeATVA2025,
  title = {{{PolyQEnt}}: {{A Polynomial Quantified Entailment Solver}}},
  shorttitle = {{{PolyQEnt}}},
  booktitle = {Automated {{Technology}} for {{Verification}} and {{Analysis}}},
  author = {Chatterjee, Krishnendu and Goharshady, Amir Kafshdar and Goharshady, Ehsan Kafshdar and Karrabi, Mehrdad and Saadat, Milad and Seeliger, Maximilian and {\v Z}ikeli{\'c}, {\DJ}or{\dj}e},
  year = 2025,
  series = {Lecture {{Notes}} in {{Computer Science}}},
  volume = {16145},
  pages = {411--424},
  publisher = {Springer Nature Switzerland},
  address = {Cham},
  doi = {10.1007/978-3-032-08707-2_19},
  url = {https://link.springer.com/10.1007/978-3-032-08707-2_19},
  urldate = {2025-11-12},
  isbn = {978-3-032-08706-5 978-3-032-08707-2},
  langid = {english}
}

@inproceedings{ChatterjeeCAV2016,
  title = {Termination {{Analysis}} of {{Probabilistic Programs Through Positivstellensatz}}'s},
  booktitle = {Computer {{Aided Verification}}},
  author = {Chatterjee, Krishnendu and Fu, Hongfei and Goharshady, Amir Kafshdar},
  year = 2016,
  series = {Lecture {{Notes}} in {{Computer Science}}},
  volume = {9779},
  pages = {3--22},
  publisher = {Springer International Publishing},
  address = {Cham},
  doi = {10.1007/978-3-319-41528-4_1},
  url = {http://link.springer.com/10.1007/978-3-319-41528-4_1},
  urldate = {2020-03-02},
  isbn = {978-3-319-41527-7 978-3-319-41528-4},
  langid = {english}
}

@inproceedings{ChatterjeeCAV2022,
  title = {Sound and {{Complete Certificates}} for {{Quantitative Termination Analysis}} of {{Probabilistic Programs}}},
  booktitle = {Computer {{Aided Verification}}},
  author = {Chatterjee, Krishnendu and Goharshady, Amir Kafshdar and Meggendorfer, Tobias and {\v Z}ikeli{\'c}, {\DJ}or{\dj}e},
  year = 2022,
  series = {Lecture {{Notes}} in {{Computer Science}}},
  volume = {13371},
  pages = {55--78},
  publisher = {Springer International Publishing},
  address = {Cham},
  doi = {10.1007/978-3-031-13185-1_4},
  url = {https://link.springer.com/10.1007/978-3-031-13185-1_4},
  urldate = {2022-08-08},
  isbn = {978-3-031-13184-4 978-3-031-13185-1},
  langid = {english}
}

@inproceedings{ChatterjeeCSL2003,
  title = {Simple Stochastic Parity Games},
  booktitle = {Computer Science Logic},
  author = {Chatterjee, Krishnendu and Jurdzi{\'n}ski, Marcin and Henzinger, Thomas A.},
  year = 2003,
  series = {Lecture {{Notes}} in {{Computer Science}}},
  volume = {2803},
  pages = {100--113},
  publisher = {Springer Berlin Heidelberg},
  doi = {10.1007/978-3-540-45220-1_11},
  isbn = {978-3-540-45220-1}
}

@article{ChatterjeeFAC2023,
  title = {On {{Lexicographic Proof Rules}} for {{Probabilistic Termination}}},
  author = {Chatterjee, Krishnendu and Kafshdar Goharshady, Ehsan and Novotn{\'y}, Petr and Z{\'a}rev{\'u}cky, Ji{\v r}{\'i} and {\v Z}ikeli{\'c}, {\DJ}or{\dj}e},
  year = 2023,
  month = jun,
  journal = {Formal Aspects of Computing},
  volume = {35},
  number = {2},
  pages = {1--25},
  issn = {0934-5043, 1433-299X},
  doi = {10.1145/3585391},
  url = {https://dl.acm.org/doi/10.1145/3585391},
  urldate = {2025-11-09},
  langid = {english}
}

@inproceedings{ChatterjeeFM2025,
  title = {Sound and Complete Witnesses for Template-Based Verification of {{LTL}} Properties on Polynomial Programs},
  booktitle = {Formal Methods},
  author = {Chatterjee, Krishnendu and Goharshady, Amir and Goharshady, Ehsan and Karrabi, Mehrdad and {\v Z}ikeli{\'c}, {\DJ}or{\dj}e},
  year = 2025,
  pages = {600--619},
  publisher = {Springer Nature Switzerland},
  address = {Cham},
  doi = {10.1007/978-3-031-71162-6_31},
  isbn = {978-3-031-71162-6}
}

@inproceedings{ChatterjeePOPL2017,
  title = {Stochastic Invariants for Probabilistic Termination},
  booktitle = {Proceedings of the 44th {{ACM SIGPLAN Symposium}} on {{Principles}} of {{Programming Languages}} - {{POPL}} 2017},
  author = {Chatterjee, Krishnendu and Novotn{\'y}, Petr and {\v Z}ikeli{\'c}, {\DH}or{\dj}e},
  year = 2017,
  pages = {145--160},
  publisher = {ACM Press},
  address = {Paris, France},
  doi = {10.1145/3009837.3009873},
  url = {http://dl.acm.org/citation.cfm?doid=3009837.3009873},
  urldate = {2020-03-14},
  isbn = {978-1-4503-4660-3},
  langid = {english}
}

@inproceedings{CirsteaCALCO2017,
  title = {Parity Automata for Quantitative Linear Time Logics},
  booktitle = {7th Conference on Algebra and Coalgebra in Computer Science ({{CALCO}} 2017)},
  author = {Cirstea, Corina and Shimizu, Shunsuke and Hasuo, Ichiro},
  year = 2017,
  series = {Leibniz International Proceedings in Informatics (Lipics)},
  volume = {72},
  pages = {7:1--7:18},
  publisher = {Schloss Dagstuhl -- Leibniz-Zentrum f\"ur Informatik},
  address = {Dagstuhl, Germany},
  issn = {1868-8969},
  doi = {10.4230/LIPIcs.CALCO.2017.7},
  url = {https://drops.dagstuhl.de/entities/document/10.4230/LIPIcs.CALCO.2017.7},
  isbn = {978-3-95977-033-0},
  urn = {urn:nbn:de:0030-drops-80468}
}

@inproceedings{CookPOPL2007,
  title = {Proving That Programs Eventually Do Something Good},
  booktitle = {Proceedings of the 34th Annual {{ACM SIGPLAN-SIGACT}} Symposium on {{Principles}} of Programming Languages},
  author = {Cook, Byron and Gotsman, Alexey and Podelski, Andreas and Rybalchenko, Andrey and Vardi, Moshe Y.},
  year = 2007,
  month = jan,
  pages = {265--276},
  publisher = {ACM},
  address = {Nice France},
  doi = {10.1145/1190216.1190257},
  url = {https://dl.acm.org/doi/10.1145/1190216.1190257},
  urldate = {2026-03-24},
  isbn = {978-1-59593-575-5},
  langid = {english}
}

@inproceedings{FuVMCAI2019,
  title = {Termination of {{Nondeterministic Probabilistic Programs}}},
  booktitle = {Verification, {{Model Checking}}, and {{Abstract Interpretation}}},
  author = {Fu, Hongfei and Chatterjee, Krishnendu},
  year = 2019,
  series = {Lecture {{Notes}} in {{Computer Science}}},
  volume = {11388},
  pages = {468--490},
  publisher = {Springer International Publishing},
  address = {Cham},
  doi = {10.1007/978-3-030-11245-5_22},
  url = {http://link.springer.com/10.1007/978-3-030-11245-5_22},
  urldate = {2020-03-15},
  isbn = {978-3-030-11244-8 978-3-030-11245-5},
  langid = {english}
}

@book{Gradel2002,
  title = {Automata {{Logics}}, and {{Infinite Games}}: {{A Guide}} to {{Current Research}}},
  shorttitle = {Automata {{Logics}}, and {{Infinite Games}}},
  author = {Gr{\"a}del, Erich and Thomas, Wolfgang and Wilke, Thomas},
  year = 2002,
  series = {Lecture {{Notes}} in {{Computer Science}}},
  volume = {2500},
  publisher = {Springer Berlin Heidelberg},
  address = {Berlin, Heidelberg},
  doi = {10.1007/3-540-36387-4},
  url = {http://link.springer.com/10.1007/3-540-36387-4},
  urldate = {2025-02-10},
  copyright = {http://www.springer.com/tdm},
  isbn = {978-3-540-00388-5 978-3-540-36387-3},
  langid = {english}
}

@inproceedings{HenzingerCAV2025,
  title = {Supermartingale Certificates for Quantitative Omega-Regular Verification and Control},
  booktitle = {Computer Aided Verification},
  author = {Henzinger, Thomas A. and Mallik, Kaushik and Sadeghi, Pouya and {\v Z}ikeli{\'c}, {\DJ}or{\dj}e},
  year = 2025,
  pages = {29--55},
  publisher = {Springer Nature Switzerland},
  address = {Cham},
  doi = {10.1007/978-3-031-98679-6_2},
  isbn = {978-3-031-98679-6}
}

@article{HuangOOPSLA2019,
  title = {Modular Verification for Almost-Sure Termination of Probabilistic Programs},
  author = {Huang, Mingzhang and Fu, Hongfei and Chatterjee, Krishnendu and Goharshady, Amir Kafshdar},
  year = 2019,
  month = oct,
  journal = {Proceedings of the ACM on Programming Languages},
  volume = {3},
  number = {OOPSLA},
  pages = {1--29},
  issn = {24751421},
  doi = {10.1145/3360555},
  url = {http://dl.acm.org/citation.cfm?doid=3366395.3360555},
  urldate = {2020-03-04},
  langid = {english}
}

@inproceedings{JurdzinskiSTACS2000,
  title = {Small {{Progress Measures}} for {{Solving Parity Games}}},
  booktitle = {{{STACS}} 2000},
  author = {Jurdzi{\'n}ski, Marcin},
  year = 2000,
  series = {Lecture {{Notes}} in {{Computer Science}}},
  volume = {1770},
  pages = {290--301},
  publisher = {Springer Berlin Heidelberg},
  address = {Berlin, Heidelberg},
  doi = {10.1007/3-540-46541-3_24},
  url = {http://link.springer.com/10.1007/3-540-46541-3_24},
  urldate = {2020-04-08},
  isbn = {978-3-540-67141-1 978-3-540-46541-6}
}

@inproceedings{Kenyon-RobertsLICS2021,
  title = {Supermartingales, {{Ranking Functions}} and {{Probabilistic Lambda Calculus}}},
  booktitle = {2021 36th {{Annual ACM}}/{{IEEE Symposium}} on {{Logic}} in {{Computer Science}} ({{LICS}})},
  author = {{Kenyon-Roberts}, Andrew and Ong, C.-H. Luke},
  year = 2021,
  month = jun,
  pages = {1--13},
  publisher = {IEEE},
  address = {Rome, Italy},
  doi = {10.1109/LICS52264.2021.9470550},
  url = {https://ieeexplore.ieee.org/document/9470550/},
  urldate = {2023-09-19},
  isbn = {978-1-6654-4895-6}
}

@inproceedings{KuraTACAS2019,
  title = {Tail Probabilities for Randomized Program Runtimes via Martingales for Higher Moments},
  booktitle = {Tools and {{Algorithms}} for the {{Construction}} and {{Analysis}} of {{Systems}}},
  author = {Kura, Satoshi and Urabe, Natsuki and Hasuo, Ichiro},
  year = 2019,
  series = {Lecture {{Notes}} in {{Computer Science}}},
  volume = {11428},
  pages = {135--153},
  publisher = {Springer},
  address = {Prague, Czech Republic},
  doi = {10.1007/978-3-030-17465-1_8},
  url = {https://doi.org/10.1007/978-3-030-17465-1_8}
}

@article{McIverPOPL2018,
  title = {A New Proof Rule for Almost-Sure Termination},
  author = {McIver, Annabelle and Morgan, Carroll and Kaminski, Benjamin Lucien and Katoen, Joost-Pieter},
  year = 2018,
  month = jan,
  journal = {Proceedings of the ACM on Programming Languages},
  volume = {2},
  number = {POPL},
  pages = {1--28},
  issn = {2475-1421, 2475-1421},
  doi = {10.1145/3158121},
  url = {http://dl.acm.org/doi/10.1145/3158121},
  urldate = {2020-03-14},
  langid = {english}
}

@article{McIverTOCL2007,
  title = {Results on the Quantitative {$\mu$}-Calculus {{qM$\mu$}}},
  author = {McIver, Annabelle and Morgan, Carroll},
  year = 2007,
  month = jan,
  journal = {ACM Transactions on Computational Logic},
  volume = {8},
  number = {1},
  pages = {3},
  publisher = {Association for Computing Machinery (ACM)},
  issn = {1529-3785, 1557-945X},
  doi = {10.1145/1182613.1182616},
  url = {https://dl.acm.org/doi/10.1145/1182613.1182616},
  urldate = {2025-07-24},
  langid = {english}
}

@book{Meyn1996,
  title = {Markov {{Chains}} and {{Stochastic Stability}}},
  author = {Meyn, Sean P.},
  year = 1996,
  series = {Communications and {{Control Engineering Ser}}},
  publisher = {Springer London, Limited},
  address = {London},
  collaborator = {Tweedie, Richard L.},
  isbn = {978-1-4471-3267-7},
  langid = {english}
}

@inproceedings{TakisakaCAV2024,
  title = {Lexicographic {{Ranking Supermartingales}} with {{Lazy Lower Bounds}}},
  booktitle = {Computer {{Aided Verification}}},
  author = {Takisaka, Toru and Zhang, Libo and Wang, Changjiang and Liu, Jiamou},
  year = 2024,
  series = {Lecture {{Notes}} in {{Computer Science}}},
  volume = {14683},
  pages = {420--442},
  publisher = {Springer Nature Switzerland},
  address = {Cham},
  doi = {10.1007/978-3-031-65633-0_19},
  url = {https://link.springer.com/10.1007/978-3-031-65633-0_19},
  urldate = {2024-11-29},
  isbn = {978-3-031-65632-3 978-3-031-65633-0},
  langid = {english}
}

@article{TakisakaTOPLAS2021,
  title = {Ranking and {{Repulsing Supermartingales}} for {{Reachability}} in {{Randomized Programs}}},
  author = {Takisaka, Toru and Oyabu, Yuichiro and Urabe, Natsuki and Hasuo, Ichiro},
  year = 2021,
  month = jun,
  journal = {ACM Transactions on Programming Languages and Systems},
  volume = {43},
  number = {2},
  pages = {1--46},
  issn = {0164-0925, 1558-4593},
  doi = {10.1145/3450967},
  url = {https://dl.acm.org/doi/10.1145/3450967},
  urldate = {2025-04-23},
  langid = {english}
}

@article{UnnoPOPL2023,
  title = {Modular {{Primal-Dual Fixpoint Logic Solving}} for {{Temporal Verification}}},
  author = {Unno, Hiroshi and Terauchi, Tachio and Gu, Yu and Koskinen, Eric},
  year = 2023,
  month = jan,
  journal = {Proceedings of the ACM on Programming Languages},
  volume = {7},
  number = {POPL},
  pages = {2111--2140},
  issn = {2475-1421},
  doi = {10.1145/3571265},
  url = {https://dl.acm.org/doi/10.1145/3571265},
  urldate = {2023-02-23},
  langid = {english}
}

@inproceedings{UrabeLICS2017,
  title = {Categorical Liveness Checking by Corecursive Algebras},
  booktitle = {2017 32nd {{Annual ACM}}/{{IEEE Symposium}} on {{Logic}} in {{Computer Science}} ({{LICS}})},
  author = {Urabe, Natsuki and Hara, Masaki and Hasuo, Ichiro},
  year = 2017,
  month = jun,
  pages = {1--12},
  publisher = {IEEE},
  address = {Reykjavik, Iceland},
  doi = {10.1109/LICS.2017.8005151},
  url = {http://ieeexplore.ieee.org/document/8005151/},
  urldate = {2020-03-27},
  isbn = {978-1-5090-3018-7}
}

@inproceedings{VardiSFCS1985,
  title = {Automatic Verification of Probabilistic Concurrent Finite State Programs},
  booktitle = {26th {{Annual Symposium}} on {{Foundations}} of {{Computer Science}} (Sfcs 1985)},
  author = {Vardi, Moshe Y.},
  year = 1985,
  pages = {327--338},
  publisher = {IEEE},
  address = {Portland, OR, USA},
  doi = {10.1109/SFCS.1985.12},
  url = {http://ieeexplore.ieee.org/document/4568158/},
  urldate = {2026-03-24},
  isbn = {978-0-8186-0644-1}
}

@inproceedings{WangPLDI2019,
  title = {Cost Analysis of Nondeterministic Probabilistic Programs},
  booktitle = {Proceedings of the 40th {{ACM SIGPLAN Conference}} on {{Programming Language Design}} and {{Implementation}}},
  author = {Wang, Peixin and Fu, Hongfei and Goharshady, Amir Kafshdar and Chatterjee, Krishnendu and Qin, Xudong and Shi, Wenjun},
  year = 2019,
  month = jun,
  pages = {204--220},
  publisher = {ACM},
  address = {Phoenix AZ USA},
  doi = {10.1145/3314221.3314581},
  url = {https://dl.acm.org/doi/10.1145/3314221.3314581},
  urldate = {2025-10-23},
  isbn = {978-1-4503-6712-7},
  langid = {english}
}

@InProceedings{UnnoCAV2021,
  title = {Constraint-Based Relational Verification},
  booktitle = {Computer Aided Verification},
  author = {Hiroshi Unno and Tachio Terauchi and Eric Koskinen},
  year = {2021},
  pages = {742--766},
  publisher = {Springer Berlin Heidelberg},
  address = {Cham},
  doi = {10.1007/978-3-030-81685-8_35},
  isbn = {978-3-030-81685-8}
}

\end{document}